\documentclass{article}

\usepackage{hyperref}
\usepackage{amsfonts,amssymb,amsmath,exscale,bbm}
\usepackage{footnpag} 
\usepackage[all,knot]{xy}  
\usepackage{color}
\usepackage{graphicx}
\usepackage{epsfig}
\usepackage{subfig}
\usepackage{enumerate}
\usepackage{stmaryrd}

\usepackage{amsmath}
\usepackage{amsfonts}
\usepackage{graphicx}
\usepackage{psfrag}
\usepackage{array}
\usepackage{bbm}
\usepackage{hyperref}
\usepackage{amsthm}
\theoremstyle{definition}
\newtheorem{definition}{Definition}
\newtheorem{theorem}{Theorem}

\newtheorem{proposition}{Proposition}
\newtheorem{corollary}{Corollary}

\newcommand{\cA}{{\mathcal A}}
\newcommand{\cB}{{\mathcal B}}

\newcommand{\cF}{{\mathcal F}}
\newcommand{\cG}{{\mathcal G}}
\newcommand{\cH}{{\mathcal H}}

\newcommand{\cL}{{\mathcal L}}
\newcommand{\cM}{{\mathcal M}}

\newcommand{\cS}{{\mathcal S}}
\newcommand{\cT}{{\mathcal T}}

\newcommand{\cZ}{{\mathcal Z}}

\def\inv{{\mbox{\tiny -1}}}

\newcommand\beq{\begin{equation}}
\newcommand\eeq{\end{equation}}
\newcommand{\be}{\begin{equation}}
\newcommand{\ee}{\end{equation}}
\newcommand{\bes}{\begin{eqnarray}}
\newcommand{\ees}{\end{eqnarray}}
\newcommand{\bea}{\begin{eqnarray}}
\newcommand{\eea}{\end{eqnarray}}

\def\vphi{{\varphi}}
\def\vphib{\overline{{\varphi}}}

\newcommand{\one}{\mbox{$1 \hspace{-1.0mm}  {\bf l}$}}

      \def\nn{{\nonumber}}

\newcommand{\su}{\mathfrak{su}}

\newcommand{\SU}{\mathrm{SU}}

\def\extd{\mathrm {d}}
\newcommand{\e}{\epsilon}

\newcommand\acts\triangleright

\newcounter{letter} \newcounter{numeral} \newcounter{Numeral}

\def\vphi{\varphi}

\def\e{\mbox{e}}
\def\extd{\mathrm {d}}

\newtheorem{theo}{Theorem}
\newtheorem{lemma}[theo]{Lemma}

\newcommand{\prf }{\medskip\noindent{\bf Proof\ }}

\begin{document}


\begin{titlepage}
\begin{flushright}
LPT-ORSAY 13-25, AEI-2013-167\\
\end{flushright}

\vspace{20pt}

\begin{center}

{\Large\bf Renormalization of an $SU(2)$ Tensorial Group Field Theory} \\
\medskip
{\Large \bf in Three Dimensions}
\vspace{15pt}

{\large Sylvain Carrozza$^{a,b}$, Daniele Oriti$^{b}$ and Vincent Rivasseau$^{a,c} $}

\vspace{15pt}

$^{a}${\sl Laboratoire de Physique Th\'{e}orique, CNRS UMR 8627,\\
 Universit\'{e} Paris Sud, F-91405 Orsay Cedex, France, EU\\
}
\vspace{5pt}

$^{b}${\sl Max Planck Institute for Gravitational Physics,\\
Albert Einstein Institute, Am M\"uhlenberg 1, 14476 Golm, Germany, EU\\
}
\vspace{5pt}
$^{c}${\sl Perimeter Institute, Waterloo, Canada\\
}
\vspace{5pt}

{\sl Emails:   sylvain.carrozza@aei.mpg.de,  daniele.oriti@aei.mpg.de,  rivass@th.u-psud.fr
}

\vspace{10pt}

\begin{abstract}
We address in this paper the issue of renormalizability for SU(2) Tensorial Group Field
Theories (TGFT) with geometric Boulatov-type conditions in three dimensions. 
We prove that interactions up to $\phi^6$-tensorial type are 
just renormalizable without any anomaly.
Our new models define the renormalizable TGFT version of the Boulatov model and
provide therefore a new approach to quantum gravity in three dimensions.
Among the many new technical results established in this paper are a general classification of just renormalizable models with gauge invariance condition, and in particular concerning properties of melonic graphs, the second order expansion 
of melonic two point subgraphs needed for wave-function renormalization. 
\end{abstract}

\end{center}

\noindent  Pacs numbers:  11.10.Gh, 04.60.-m
\\

\setcounter{footnote}{0}

\end{titlepage}


\section*{Introduction}
\addcontentsline{toc}{section}{Introduction}

Tensorial group field theories (TGFTs)  \cite{GFT1,GFT2,GFT3, vincentTensor} are promising candidates for a background independent formulation of quantum gravity. They represent the convergence of developments in loop quantum gravity \cite{LQG}, in its covariant, simplicial implementation in terms of spin foam models \cite{SF,zakopane}, and of the extension of the formalism of matrix models for 2d gravity \cite{MM} to higher dimensions. Group field theories (GFTs) \cite{GFT1,GFT2,GFT3} can be seen as a second quantization of  loop quantum gravity, adapted to a discrete setting, such that spin networks (the quantum states of geometry in LQG) are created/annihilated with their interaction processes being assigned a Feynman amplitude which corresponds to the definition of a spin foam model. Accordingly, the data labeling field, states and histories (Feynman diagrams) of GFTs are group elements, Lie algebra elements or group representations. These data are very useful to extract geometric content from GFT structures, beside their combinatorial aspects and to characterize better their quantum dynamics. In this context promising models for 4d quantum gravity have been developed, e.g. \cite{EPRL, BO-Immirzi}. These are models based on the group manifold $SU(2)$ or $SO(3,1)$, constructed by imposing additional \lq\lq simplicity\rq\rq conditions, motivated by simplicial geometry and classical continuum gravity, onto GFT models describing topological BF theory, already characterized by a {\it gauge invariance} condition imposed on the GFT fields. 

The Feynman diagrams of GFTs are cellular complexes, and the perturbative GFT dynamics is defined by the sum over them, in principle extended to include arbitrary topologies. Recently, work on (colored) tensor models \cite{tensor, tensorReview,universality}, generalizing matrix models to define a perturbative sum over cellular complexes of arbitrary dimension, have led to a detailed understanding of the combinatorial features, statistical properties and universality aspects of such sums. The progress has been remarkable, leading for example to: 1) the definition of a large-N expansion \cite{large-N} (where $N$ is the size of the tensor index set), and the identification of the dominant configurations in this expansion, which turn out to be special types of spherical complexes called \lq\lq melons\rq\rq; 2) the proof that random un-symmetric rank-d tensors have natural polynomial interactions based on  $U(N)^{\otimes d}$ invariance\footnote{There have been also interesting applications to statistical 
physics, in particular dimers \cite{bonzomerbin} and spin glasses \cite{BGS}}. The incorporation of these key insights, coming from simpler tensor models, into the GFT formalism defines what we call {\it tensorial group field theories} possessing the richer pre-geometric content suggested by loop quantum gravity and spin foam models, added to the solid mathematical backbone of tensor models.

All these approaches define a fundamental quantum dynamics for degrees of freedom which are discrete, characterized by algebraic and combinatorial data only, thus pre-geometric. The key open issue is to extract from this the microscopic quantum dynamics and effective continuous limit of spacetime and geometry, with an effective dynamics that has to be related to (some modified form of) General Relativity. This transition to a continuum, geometric description has been dubbed \lq\lq geometrogenesis\rq\rq, and suggested to be associated to one or several phase transitions of the underlying quantum gravity system, with the further suggestion that the relevant phase corresponds to a {\it condensate} of the microscopic degrees of freedom \cite{GFTfluid, lorenzoGFT, vincentTensor}. This picture is even partially realized in \cite{GFTcosmo}. The problem can be approached in purely statistical terms in tensor models \cite{critical}, but the extra data of GFTs allow one to make use of the results of loop quantum gravity \cite{LQG} to read out continuum physics from specific models. 

In fact, as non-trivial quantum field theories, TGFTs offer a very convenient setting to approach this problem. Effective continuum physics can be looked for in their symmetries \cite{joseph, GFTdiffeos, virasoro}, or in collective effects to be extracted, for example, via mean field techniques \cite{danielelorenzo, danieleflorianetera,effHamilt}, or encoded in simplified models \cite{gfc}. The most powerful tool they offer, however, is the renormalization group. It is indeed the renormalization group that should govern the flow from the microscopic dynamics of few pre-geometric TGFT degrees of freedom to their effective macroscopic dynamics, involving an infinite number of them (modulo, of course, further approximation of the resulting continuum theory).  

The study of (perturbative) renormalizability of TGFTs has been one of the main directions of developments in recent years. This includes important, if preliminary calculations of radiative corrections in TGFTs with a direct interpretation in terms of quantum gravity, in both 3d \cite{josephvalentin} and 4d \cite{aldo}, and various steps in a systematic program \cite{GFTrenorm,lin,valentinmatteo,bgriv,addendum,josephsamary,josephaf,scaling,COR,gelounlivine,SVT} whose goal is a complete proof of renormalizability of realistic TGFT models for 4d quantum gravity, including (or reproducing at some effective level) all the ingredients and data that seem to be relevant for a proper encoding of quantum geometry. The next step in the same program would be a full characterization of the renormalization group flow of the same models, as encoded in the RG equations 
and in particular their beta functions. Important results on this second point have been obtained in \cite{josephsamary,josephaf}, where asymptotic freedom has been established for some simple TGFT models, but also argued to be a general feature in the TGFT formalism. Indeed wave function renormalization seems generically stronger in the tensorial context than in the scalar, vector or matrix case. This feature would make them prime candidates for a geometrogenesis scenario, as a quantum gravity analog of quark confinement in QCD. 
The last step would finally be a detailed study of their constructive aspects\footnote{Indeed constructibility of TGFTs can be assessed via rigorous \emph{constructive}
analysis in their dilute perturbative phase, through the \emph{loop vertex expansion} \cite{Rivasseau:2007fr}. This tool has been already applied to tensor models \cite{Magnen:2009at, universality} and has indeed a very general range of applicability as far as field theories are concerned \cite{Rivasseau:2010ke}.}.  

Such systematic renormalization analysis requires first of all a clear definition of the TGFT models one is working with. As field theories, TGFTs involve a choice of a propagator and  of a class of interactions.    

Concerning the kinetic term, the usual quantum gravity TGFT models suggested by loop quantum gravity are {\it ultralocal} with trivial kinetic operators (delta functions or simple projectors). These seem appropriate from the perspective of simplicial gravity path integrals, but generally do not allow the definition of renormalization group scales. It is also true that these models are still highly non-trivial due to specific symmetries and other conditions imposed on the fields and to the peculiar non-local nature of the interactions, thus it is possible that they can provide an alternative, less direct definition of such scales. This possibility however has not been explored yet. Such scales are instead defined in a very straightforward manner in proper {\it dynamical} TGFTs (first considered, with different motivations, in \cite{generalisedGFT}), characterized by kinetic operators given by differential operators on the group manifold, such as the Laplace-Beltrami operator. There are even indications \cite{josephvalentin} that ultralocal models turn into dynamical models as soon as radiative corrections are considered, since the kinetic terms with Laplacian operators are required as counter-terms. For these reasons, we consider these dynamical models in this paper.

As for the interactions, in usual quantum field theories, these are specified by the requirement of {\it locality}, which in turns translates into the simple identification of field arguments in the interaction terms entering the action. From this formal perspective, TGFTs are {\it non-local}, in that the field arguments in the interaction terms generically have a combinatorially non-trivial pattern of convolutions. Indeed, they fall into two classes, each corresponding to a suggested alternative notion of locality. TGFT models corresponding to spin foam models and inspired by LQG impose {\it simpliciality} of the interactions, whereby the combinatorics of field convolutions describes the gluing of $(d-1)$-simplices across shared $(d-2)$-simplices to form $d$-simplices. This comes from the wish to have Feynman diagrams corresponding to simplicial complexes and weighted by a group-theoretic version of a simplicial gravity path integral. In turn, work on tensor model universality and on TGFT renormalization has suggested the notion of {\it traciality}, in turn coming from the mentioned $U(N)^d$ invariance. We detail this notion in the following, as we are going to work with interactions incorporating it. Once more, these two notions of locality and the resulting types of interactions are not disconnected, even though their exact relation is not yet understood: integration of fields in a path integral for TGFTs based on simpliciality does in fact result in effective interactions (for the remaining fields) characterized by $U(N)^d$ invariance \cite{uncoloring}. Moreover, the combinatorics of such tensor invariant can be represented by polytopes with triangular faces (in turn obtainable by gluing tetrahedra around common vertices) \cite{uncoloring}. 

The first TGFT models in 3d and 4d were shown to be perturbatively renormalizable at all orders in \cite{bgriv,josephsamary}. These were Abelian models with tensor invariance and Laplacian kinetic term, with no additional constraints on the fields. The next step was to include {\it gauge invariance}, which in turns results in the presence of a discrete gauge connection at the level of the Feynman amplitudes of the theory. This step was taken in \cite{COR} where an Abelian TGFT model in 4d incorporating such condition was shown to be super-renormalizable, and a general classification of Abelian models in any dimension in terms of their divergences was defined. The generalization to gauge invariant TGFTs required several non-trivial adaptations of standard notions from the renormalization of local quantum field theories to be achieved. We take advantage of such refined, generalized notions in this paper. Indeed, we take here a further step towards renormalization of realistic TGFT models for 4d gravity, and study for the first time the renormalizability of a  {\it non-Abelian} TGFT model, specializing to the 3d case and to the group manifold $SU(2)$. Other just renormalizable models of Abelian type in 5 and 6 dimensions have been shown renormalizable in \cite{SVT}.

We define the models we work with in section \ref{models}. We first discuss generic non-Abelian models, which include a gauge invariance condition under the diagonal action of $SU(2)$ on this group manifold, use a Laplacian kinetic term and tensor invariant interactions. In the same section, we define all the generalized QFT notions that are needed for the renormalization analysis, e.g. {\it face-connectedness} and (quasi-)locality, recalling or further generalizing the definitions given in \cite{COR}. We recall as well, in section \ref{powercounting} the Abelian power counting of divergences, for arbitrary dimension and Abelian group, obtained in \cite{COR}. We analyze further this divergence structure, as it will be relevant for the non-Abelian case as well, and use this classification to identify just-renormalizable models
in this category.

In section \ref{su2_model} we introduce the non-Abelian model. It is a model in the same class as the previous Abelian ones, but based on the group manifold $SU(2)^3$. It is a modification of the Boulatov model \cite{boulatov} in two key aspects. First the interaction is based on tensor-invariant colored gluings rather 
than the initial interaction proposed by Boulatov. Second it has a Laplacian term which changes the amplitudes. This term 
has not been given yet a clear geometric interpretation in terms of discrete gravity actions.
Without it, the model would correspond to a quantization of topological BF theory discretized on a cellular complex 
described by gluing generalized polytopes with triangular faces. 

 We introduce all the relevant interactions and the needed counter-terms, and identify all the divergent subgraphs. 

We perform the renormalization of this model in section \ref{multiscale}, via rigorous multi-scale expansion in the style of \cite{Riv}. 
The model turns out to be just-renormalizable (in contrast to the super-renormalizability of the Abelian case) up to interactions of degree 6.
The renormalizability analysis involves a number of interesting technical discoveries, in particular about various properties of melonic graphs. Among them, 
the structure of external faces of melonic diagrams, their inclusion relations, and the central role shown for the notion of face-connectedness, in particular concerning the expansion of divergences around their local contributions.

Finally, in section \ref{finiteness} we prove the finiteness of the renormalized series, that is we establish a BPHZ theorem for our TGFT model.

\section{TGFT models with closure constraint}
\label{models}

In this section, we recall general properties of TGFT's with closure constraint (gauge invariance) and Laplacian propagator, as defined in \cite{COR}. We then review the main conclusions of this first study that are
relevant to the present paper. These include a refined notion of connectedness, hence of quasi-locality, as well as an optimal Abelian power-counting. We finally discuss the relevance of this Abelian
power-counting in a generic non-Abelian context. 

\subsection{Definition and Feynman amplitudes}



A generic TGFT is a quantum field theory of a tensorial field, with entries in a Lie group. In this paper we assume $G$ to be compact, and the field to be a rank-$d$\footnote{Throughout this article, we assume $d \geq 3$.} complex function $\vphi( g_1 ,
\dots , g_d )$. The statistics is then defined by a partition function
\beq
\cZ = \int \extd \mu_C (\vphi , \vphib) \, \e^{- S(\vphi , \vphib )}\,,
\eeq
where $\extd \mu_C (\vphi , \vphib)$ is a Gaussian measure characterized by its covariance $C$ (i.e. propagator), and $S$ is the interaction part of the action. As in any quantum field theory,
possible interactions are determined by a \textit{locality principle}, while the definition of the dynamics (including possible constraints on the degrees of freedom of the fields) is completed by the propagator $C$, which generically breaks locality.

\
GFTs used in the context of loop quantum gravity and spin foam models use a notion of {\it simpliciality}, i.e. the requirement that interaction vertices correspond to $d$-simplices, obtained by gluing along sub-faces the $(d-1)$-simplices associated to each field. TGFTs propose a new notion of locality, in the form of \textit{tensor invariance}, initially proposed in the realm of tensor models (whence the extra characterization of these GFTs as \lq tensorial\rq). It can be thought of as a limit of a $U(N)^{\otimes d}$ invariance,
where $N$ is a cut-off on representation labels (e.g. spins) in the harmonic expansion of the field. In simpler terms, \textit{tensor invariants} are convolutions of a certain number of fields $\vphi$
and $\vphib$ such that any $k$-th index of a field $\vphi$ is
contracted with a $k$-th index of a conjugate field $\overline{\vphi}$. They are dual to $d$-\textit{colored graphs}, built from two types of nodes and $d$ types of colored edges: each white (resp. black) dot
represents a field $\vphi$ (resp. $\vphib$), while a contraction of two indices in position $k$ is associated to an edge with color label $k$. Connected such graphs, called $d$\textit{-bubbles}, generate the set of 
\textit{connected tensor invariants}. See Figure \ref{examples_bubbles} for examples in dimension $d = 3$. We assume that the interaction part of the action is a sum of such connected invariants
\beq
S(\vphi , \vphib) = \sum_{b \in \cB} t_b I_b (\vphi , \vphib)\,, 
\eeq
where $\cB$ is a finite set of $d$-bubbles, and $I_b$ is the connected invariant dual to the bubble $b$.

\begin{figure}[h]
\begin{center}
\includegraphics[scale=0.5]{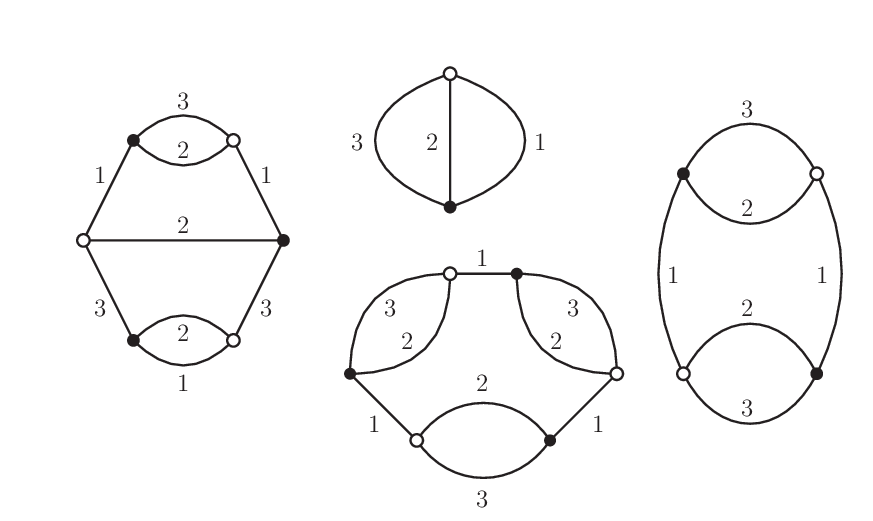}
\caption{Examples of $3$-bubbles.}
\label{examples_bubbles}
\end{center}
\end{figure}

\
The Gaussian measure $\extd \mu_C$ implements both the dynamics, through a Laplacian propagator $$\left( m^2 - \sum_{\ell = 1}^{d} \Delta_\ell \right)^{-1} \,,$$ and the gauge invariance\footnote{We use the term 'gauge invariance' in accordance with quantum gravity and lattice gauge theory usage. It refers to the discrete gauge invariance appearing at the level of the Feynman amplitudes, rather than to a gauge symmetry of the quantum field theory itself.} condition
\beq \label{gauge}
\forall h \in G \,, \qquad \vphi(h g_1, \dots , h g_d ) = \vphi(g_1, \dots , g_d )\,.
\eeq
The implications of this condition can be understood in two main ways \cite{GFT1,GFT2,GFT3,LQG,SF,EPRL,BO-Immirzi}. In full  generality, it imposes a gauge invariance of the quantum states of the model, represented as $d$-valent graphs labeled by group (or conjugate Lie algebra) elements on their links and located at the vertices of the same graphs; equivalently \cite{BO-Immirzi}, it implies that the $d$ Lie algebra elements associated to the $d$ links incident to one such vertex sum to zero. The same gauge invariance can be seen at the level of the Feynman amplitudes of the model, which acquire the form of lattice gauge theory amplitudes. Indeed, the implementation of this constraint also introduces a notion of discrete gauge connection on the Feynman diagrams of the TGFT model. For models where a geometric interpretation of the combinatorial $(d-1)$-simplices corresponding to the TGFT fields is possible, the same requirement implies the \lq closure\rq\,of the $d$ faces of such $(d-1)$-simplices to form a closed boundary hypersurface for them. This condition is therefore a necessary ingredient for the consistent interpretation of these models as encoding simplicial geometry.   
The resulting covariance can be expressed as an integral over a Schwinger parameter $\alpha$ of a product of heat kernels on $G$ at time $\alpha$:
\bes
\int \extd \mu_C (\vphi , \vphib) \, \vphi(g_1 , \dots , g_d ) \vphib(g_1' , \dots , g_d' ) &=& C(g_1, \dots , g_d ; g_1' , \dots , g_d' ) \\
&\equiv& \int_{0}^{+ \infty} \extd \alpha \, \e^{- \alpha m^2} \int \extd h \prod_{\ell = 1}^{d} K_{\alpha} (g_\ell h g_\ell'^{\inv})\,.
\ees
This decomposition of the propagator provides an intrinsic notion of scale, parametrized by $\alpha$. Divergences result from the UV region (i.e. $\alpha \to 0$), hence the need to introduce a cut-off ($\alpha \geq \Lambda$), and subsequently to remove
it via renormalization.

\begin{figure}[b]
\begin{center}
\includegraphics[scale=0.5]{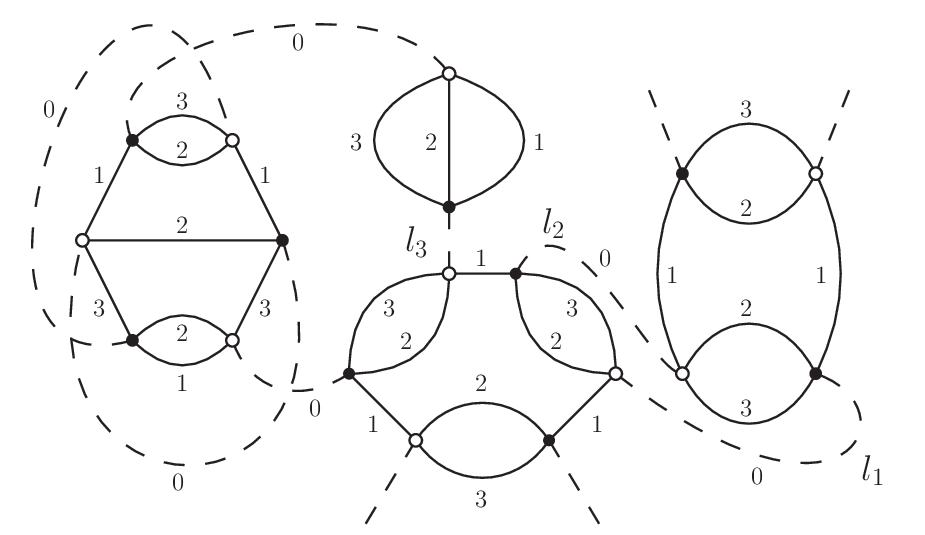}
\caption{A graph with $4$ vertices, $6$ lines and $4$ external legs in $d = 3$.}
\label{example_graph}
\end{center}
\end{figure}

\
The perturbative expansion of the theory is captured by Feynman graphs whose vertices are $d$-bubbles, and whose propagators are associated to an additional type of colored edges, of color $\ell = 0$, represented as dashed lines. When seen on the same footing,
these $d+1$ types of colored edges form $(d+1)$-colored graphs. To a Feynman \textit{graph} $\cG$, whose elements are $d$-bubble vertices ($V(\cG)$) and color-$0$ lines ($L(\cG)$), is therefore uniquely associated a $(d+1)$-colored graph $\cG_c$, called the
\textit{colored extension} of $\cG$. See Figure \ref{example_graph} for an example of Feynman graph in $d = 3$.
The connected Schwinger functions are given by a sum over line-connected Feynman graphs:
\beq
\cS_N = \sum_{\cG \; \mathrm{connected}, N(\cG)= N} \frac{1}{s(\cG)} \left(\prod_{b \in \cB} (- t_b)^{n_b (\cG)}\right) \cA_\cG \,,
\eeq
where $N(\cG)$ is the number of external legs of a graph $\cG$, $n_b (\cG)$ its number of vertices of type $b$, and $s(\cG)$ a symmetry factor. The amplitude $\cA_\cG$ of $\cG$ is expressed in terms of holonomies along its faces,
which can be easily defined in the colored extension $\cG_c$: a \textit{face} $f$ of color $\ell$ is a maximal connected subset of edges of color $0$ and $\ell$. In $\cG$, $f$ is a set of color-$0$ lines, from which the holonomies are constructed.
We finally use the following additional notations: $\alpha(f) \equiv \underset{e \in f}{\sum} \alpha_e$ is the sum of the Schwinger parameters appearing in the face $f$; $\epsilon_{ef} = \pm 1$ or $0$ is the adjacency or incidence matrix, encoding the line content of faces
and their relative orientations; the faces are split into closed ($F$) and opened ones ($F_{ext}$); $g_{s(f)}$ and $g_{t(f)}$ denote boundary variables in open faces, with functions $s$ and $t$ mapping open faces to their ``source" and ``target" boundary variables. 
The amplitude $\cA_\cG$ takes the form:
\begin{eqnarray}\label{ampl}
\cA_\cG &=& \left[ \prod_{e \in L(\cG)} \int \extd \alpha_{e} \, e^{- m^2 \alpha_e} \int \extd h_e \right] 
\left( \prod_{f \in F (\cG)} K_{\alpha(f)}\left( \overrightarrow{\prod_{e \in f}} {h_e}^{\epsilon_{ef}} \right) \right) \nn\\
&&\left( \prod_{f \in F_{ext}(\cG)} K_{\alpha(f)} \left( g_{s(f)}
\left[\overrightarrow{\prod_{e \in f}} {h_e}^{\epsilon_{ef}}\right] g_{t(f)}^{\inv} \right) \right) \,.
\end{eqnarray}

\
An important feature of the amplitude of $\cG$ is a $G^{V(\cG)}$ gauge symmetry:
\beq
h_e \mapsto g_{t(e)} h_e g_{s(e)}^{\inv}\,,
\eeq
where $t(e)$ (resp. $s(e)$) is the target (resp. source) vertex of an (oriented) edge $e$, and one of the two group elements is trivial for open lines. As we have anticipated, it is the gauge invariance \eqref{gauge} imposed on the TGFT field that is responsible of this gauge invariance at the level of the Feynman amplitudes, and for their expression \eqref{ampl} as a lattice gauge theory on $\cG$. When $\cG$ is connected, it is convenient to gauge fix the $h$ variables along a spanning
tree $\cT$ of the graph:
$$
h_e = \one
$$
in the integrand of (\ref{ampl}), for every line $e \in \cT$. We will use such gauge fixing in the following.

\subsection{Subgraphs, connectedness and quasi-locality}



We collect here a number of definitions and results, first introduced in \cite{COR}, which are key for the analysis of the non-Abelian model we will perform in the following. Among them, the new notions of {\it subgraph}, {\it face-connectedness}, {\it contractiblity}, {\it melopoles} and {\it traciality} already show that TGFTs require a non-trivial adaptation of standard QFT concepts, in order to unravel the combinatorial structure of the Feynman diagrams and to study the renormalizability.

\begin{definition}
A \textit{subgraph} $\cH$ of a graph $\cG$ is a subset of  lines of $\cG$, hence $\cG$ has exactly $2^{L(\cG)}$ subgraphs. 
$\cH$ is then completed by first adding the vertices that touch its lines. 
The faces closed in $\cG$ which  pass only through lines of $\cH$ form the set of \textit{internal faces} of $\cH$. The 
external faces of $\cH$ are the maximal open connected pieces of either open or closed faces of $\cG$ that pass through lines of $\cH$. 
Finally all the external legs or half-lines of $\cG \setminus \cH$ touching the vertices of $\cH$ are considered \textit{external legs} of $\cH$.
\end{definition}

We denote $L(\cH)$ and $F(\cH)$ the set of lines and internal faces of $\cH$, and $N(\cH)$ and $F_{ext}(\cH)$ the set of external legs and external faces. When no confusion is possible we also write $L$, $F$ etc for the cardinality 
of the corresponding sets. Moreover, the subgraph made of the lines $l_1, \ldots, l_k$ will simply be denoted $\{l_1 , \ldots , l_k \}$. 

\

\noindent {\bf{Example.}} In Figure \ref{example_graph}, $\cH_{12} = \{ l_1 , l_2 \}$ has two lines ($L(\cH_{12})=2$) which touch two vertices, giving $V(\cH_{12})=2$. Six additional half-lines are hooked up to these two bubbles, giving a total of $N( \cH_{12} )=6$ external legs. Finally, $\cH_{12}$ has four faces in total: two of them are internal, of color $2$ and $3$ respectively, hence $F(\cH_{12}) = 2$; the two others are external faces of color $1$, hence $F_{ext}(\cH_{12}) = 2$. Note that the connected pieces of (the colored extension of) $\cH_{12}$ which consist of two external legs and a single colored line should not be considered as external faces.   

\

On top of the usual notion of connectedness of subgraphs, to which we will refer as \textit{vertex-connectedness} in order to avoid any confusion, we will heavily rely on the similar concept of \textit{face-connectedness}. While the former focuses on incidence relations between lines and vertices, the latter puts the emphasis on incidence relations between lines and faces.


\begin{definition}
\begin{enumerate}[(i)]
\item The \textit{face-connected components} of a subgraph $\cH$ are defined as the subsets of lines of the maximal factorized rectangular blocks of its $\epsilon_{ef}$ incidence matrix (with entries in $L(\cH) \times F(\cH)$). 
\item A subgraph $\cH$ is called \emph{face-connected} if it has a single face-connected component.
\item Let $\cG$ be a graph. The face-connected subgraphs $\cH_1 , \ldots , \cH_k \subset \cG$ are said to be \textit{face-disjoint} if they form exactly $k$ face-connected components in their union $\cH_1 \cup \dots \cup \cH_k$. 
\end{enumerate}
\end{definition}

The notion of face-connectedness is finer than vertex-connectedness, in the sense that any face-connected subgraph is also vertex-connected. It should also be noted that with the previous definition, the face-disjoint subgraphs $\cH_1, \ldots , \cH_k \subset \cG$ can consist of strictly less than $k$ face-connected components in $\cG$ itself. What really matters is that there \textit{exists} a subgraph of $\cG$ into which $\cH_1 , \ldots , \cH_k$ form $k$ face-connected components. 
In other words, this is another instance of the importance of the underlying color structure in TGFT diagrams; it is this color structure that allows to encode fully the topology of the diagrams and of their dual cellular complexes \cite{crystallization}.

\

\noindent {\bf{Examples.}} In Figure \ref{example_graph}, $\cH_{12} = \{ l_1 , l_2 \}$ and $\cH_{123} = \{ l_1 , l_2 , l_3 \}$ are both vertex-connected, while only $\cH_{12}$ is face-connected. $\cH_{123}$ has two face-connected components: $\{ l_3 \}$ and $\{ l_1 , l_2 \}$. In Figure \ref{ex_con_comp}, $\cH_{1} = \{l_1\}$ and $\cH_2 = \{ l_2 \}$ are face-disjoint because they are their own face-connected components in $\cH_1 \cup \cH_2 = \{ l_1 , l_2\}$. On the other hand, they are not face-connected components of $\cH_{123}$, which is itself face-connected. This illustrates the subtelty in the definition of face-disjointness we just pointed out.

\begin{figure}[ht]
\begin{center}
\includegraphics[scale=0.5]{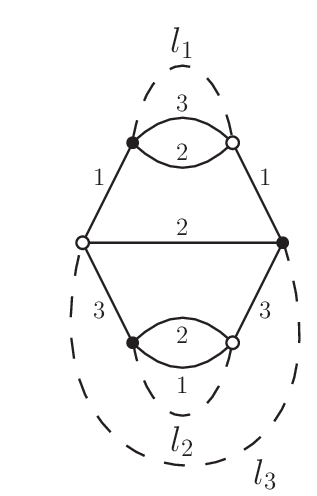}
\caption{$\cH_{1} = \{l_1\}$ and $\cH_2 = \{ l_2 \}$ are face-disjoint.}
\label{ex_con_comp}
\end{center}
\end{figure}

\

It is convenient to define elementary operations on TGFT graphs at the level of their underlying colored graphs. There, dipoles play a central role.
\begin{definition}
Let $\cG$ be a graph, and $\cG_c$ its colored extension. For any integer $k$ such that $1 \leq k \leq d+1$, a $k$-dipole is a line of $\cG$ whose image in $\cG_c$ links two nodes $n$ and
$\overline{n}$ which are connected by exactly $k - 1$ additional colored lines.
\end{definition}

\begin{definition}
Let $\cG$ be a graph, and $\cG_c$ its colored extension. The contraction of a $k$-dipole $d_k$ is an operation in $\cG_c$ that consists in:
\begin{enumerate}[(i)]
 \item deleting the two nodes $n$ and $\overline{n}$ linked by $d_k$, together with the $k$ lines that connect them;
 \item reconnecting the resulting $d - k + 1$ pairs of open legs according to their colors.
\end{enumerate}
We call $\cG_c / d_k$ the resulting colored graph, and $\cG / d_k$ its pre-image. See Figure \ref{k_dipole}.
\end{definition}

\begin{figure}[h]
\begin{center}
\includegraphics[scale=0.6]{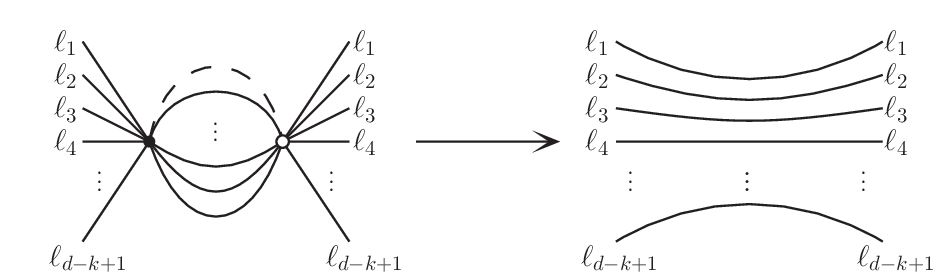}
\caption{Contraction of a $k$-dipole line.}
\label{k_dipole}
\end{center}
\end{figure}

\begin{definition}
We call contraction of a subgraph $\cH \subset \cG$ the successive contractions of all the lines of $\cH$. The resulting graph is independent of the order in which the lines of $\cH$ are contracted,
and is noted $\cG / \cH$.
\end{definition}

\begin{proposition}
Let $\cH$ be a subgraph of $\cG$, and $\cH_c$ its colored extension. The contracted graph $\cG / \cH$ is obtained by:
\begin{enumerate}[(a)]
 \item deleting all the internal faces of $\cH$;
 \item replacing all the external faces of $\cH_c$ by single lines of the appropriate color.
\end{enumerate}
\end{proposition}

Contracting a subgraph $\cH \subset \cG$ can heavily modify the connectivity properties of $\cG$, depending on the nature of the dipoles this operation involves.
\begin{proposition}\label{disconnected}
\begin{enumerate}[(i)]
 \item For any vertex-connected graph $\cG$, if $e$ is a line of $\cG$ contained in a $d$-dipole, then $\cG / e$ is vertex-connected.
 \item For any $1 \leq q \leq d - k + 1$, there exists a connected graph $\cG$ and a $k$-dipole $e$ such that $\cG / e$ has exactly $q$ connected components. 
\end{enumerate}
\end{proposition}

The following definition takes this possible loss of connectedness into account, in order to formulate a notion of quasi-locality adapted to TGFTs with gauge constraint, which we called \textit{traciality}.
\begin{definition}
Let $\cG$ be a vertex-connected graph, and $\cH$ be one of its face-connected subgraphs. 
\begin{enumerate}[(i)]
\item If $\cH$ is a tadpole\footnote{In this paper, we call \textit{tadpole} any graph with a single vertex.}, $\cH$ is \textit{contractible} if, for any group elements assignment $(h_e)_{e \in L(\cH)}$:
\beq
\left( \forall f \in F(\cH)\,, \;  \overrightarrow{\prod_{e \in f}} {h_e}^{\epsilon_{ef}} = \one \right) \Rightarrow \left( \forall e \in L(\cH) \, , \; h_e = \one \right)\,.
\eeq
\item In general, $\cH$ is contractible if it admits a spanning tree $\cT$ such that 
$\cH / \cT$ is a contractible tadpole.
\item $\cH$ is \textit{tracial} if it is contractible and the contracted graph $\cG / \cH$ is connected.
\end{enumerate}
\end{definition}

\

Finally, we recall the notion of melopole, a special class of tracial tadpole subgraphs which were responsible for all the divergences in \cite{COR}.
\begin{definition}
In a graph $\cG$, a \emph{melopole}
is a single-vertex subgraph $\cH$ (hence $\cH$ is made of tadpole lines attached to a single vertex in the ordinary sense), 
such that there is at least one ordering (or ``Hepp's sector") of its $k$ lines as $l_1, \cdots , l_k$ such that $ \{l_1 ,  \dots , l_{i} \} / \{l_1 , \dots , l_{i-1} \} $ is a $d$-dipole for $1 \le i \le k$. See Figure \ref{ex_melop}.
\end{definition}

\begin{figure}[ht]
\begin{center}
\includegraphics[scale=0.5]{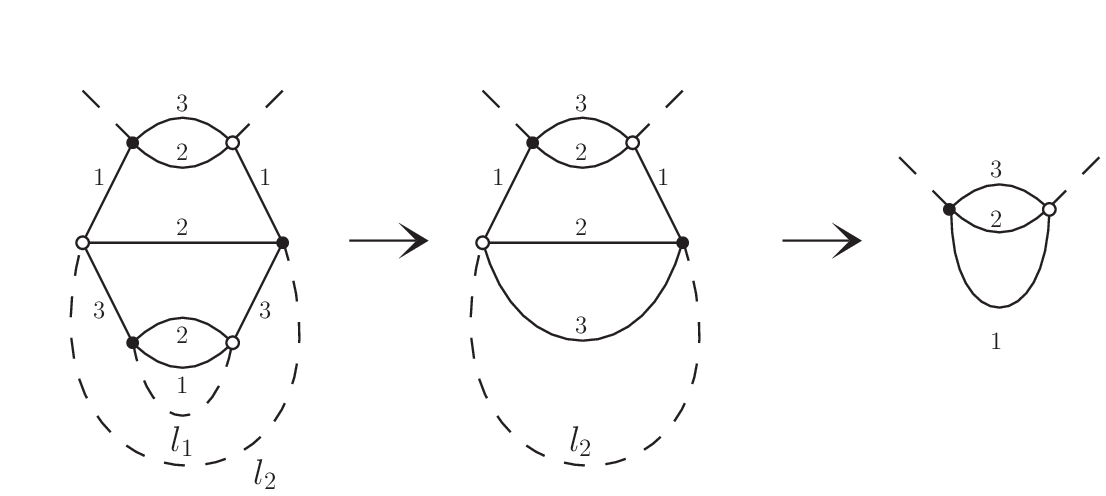}
\caption{A melopole with two lines, in $d=3$. $\{ l_1 \}$ and $\{l_1 , l_2 \} / \{l_1 \}$ are $3$-dipoles, as illustrated by the successive contractions of $l_1$ and $l_2$.}
\label{ex_melop}
\end{center}
\end{figure}

\begin{proposition}
Any face-connected melopole is tracial.
\end{proposition} 


In just-renormalizable models, a larger class of tracial subgraphs will dominate, which extend the notion of melopole to an arbitrary number of vertices.

\begin{definition}
In a graph $\cG$, a \emph{melonic subgraph}
is a face-connected subgraph $\cH$ containing at least one maximal tree $\cT$ such that $\cH / \cT$ is a melopole.\footnote{Remember that the notion of face-connectedness only takes the internal faces into account. The present definition is chosen so that at least one internal face of $\cG$ runs through any line of any melonic subgraph. $\cG$ itself is considered melonic if it is melonic as a subgraph of itself. This definition will ensure Lemma \ref{1PI}.}
\end{definition}

In Figure \ref{ex_melonic}, we give three simple examples of melonic subgraphs in $d=3$, with $2$, $4$ and $6$ external legs respectively. 

\begin{figure}[ht]
  \centering
  \subfloat{\label{tg62ll}\includegraphics[scale=0.6]{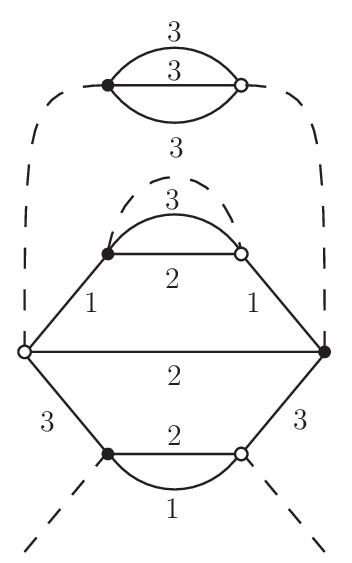}}   	       
  \subfloat{\label{h4ll}\includegraphics[scale=0.6]{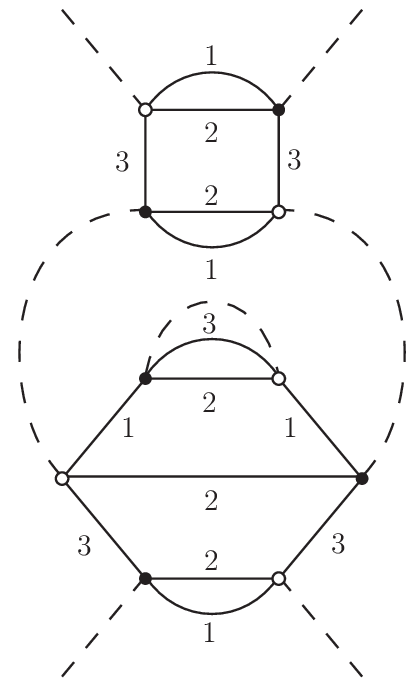}}     
  \subfloat{\label{g61ll}\includegraphics[scale=0.6]{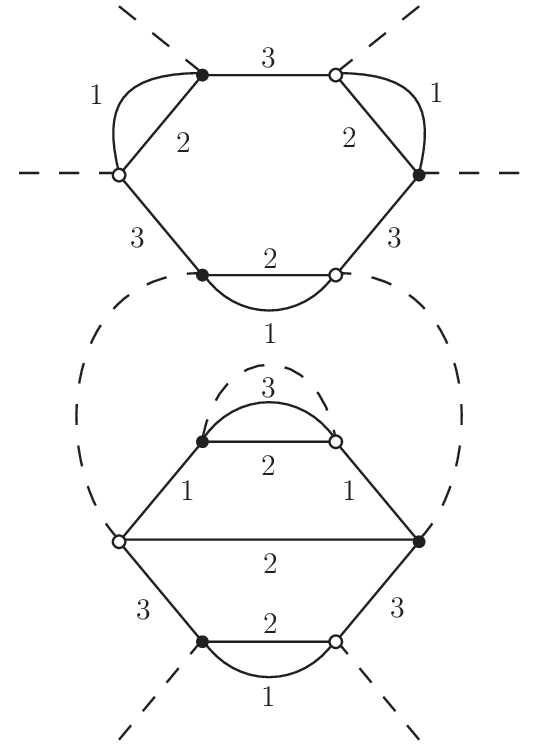}}     
  \caption{Examples of melonic subgraphs in $d=3$, with $2$, $4$ and $6$ external legs.}\label{ex_melonic}
\end{figure}

\begin{proposition}
Any melonic subgraph is tracial.
\end{proposition} 

\subsection{Abelian power-counting}

The main general result of \cite{COR} is an Abelian power-counting theorem. Derived in a multi-scale form, it identifies the divergence degree $\omega$, providing a bound on the asymptotic behaviour of 
the amplitudes when the cut-off $\Lambda$ is removed. As we will show, this bound holds for general group $G$, not necessarily Abelian. However, the bound is {\it optimal} when the group is Abelian. We call $D$ the dimension of $G$.

\begin{definition}
Let $\cH$ be a subgraph of $\cG$. The \textit{degree of divergence} of $\cH$ is defined by
\beq  
\omega (\cH)  = -2 L (\cH) + D(F(\cH) - R(\cH)) \,,
\eeq
where $R(\cH)$ is the rank of the $\epsilon_{lf}$ incidence matrix of $\cH$.
$\cH$ is \textit{divergent} when $\omega(\cH) \geq 0$, and \textit{convergent} otherwise.
\end{definition}

When $G$ is Abelian, the divergences of any graph $\cG$ are fully captured by its divergent subgraphs. On the other hand, if $G$ is not Abelian, a twisted degree of divergence is needed to account for the exact structure of divergences \cite{valentinmatteo}. We shall however show that: a) the Abelian power-counting still holds as a bound in the non-Abelian case; b) the degree of divergence and its twisted non-Abelian version coincide for contractible subgraphs. A reasonably detailed proof of these claims will be provided in section \ref{sec:multiscale}, in their multi-scale version, and for the group $G = \SU(2)$. But the intuitive reasons behind these are rather simple. First, thanks to their colored structure, TGFT graphs do not contain any \textit{tadfaces} (i.e. faces running several times through the same line), therefore decays can be successively extracted from the propagators by simple convolutions, in the very same way as in the Abelian case. Second, in a contractible subgraph $\cH$, flat connections are fully captured by the neighborhood of $h_e = \one$ for any $e \in L(\cH)$, in which case the Abelianized amplitude is nothing but a saddle point approximation of the full amplitude, hence correctly capturing its divergences.

\section{Abelian divergence degree and just-renormalizability}
\label{powercounting}
 
 \subsection{Analysis of the Abelian divergence degree}

In this section, we present a detailed analysis of the degree of divergence \cite{COR}
\beq  \omega(\cH)  =  -2L  + D (F-R)\,.
\eeq

We consider a face-connected subgraph $\cH \subset \cG$ with $V$ vertices, $L$ lines with $d \geq 3$ strands each, $F$ internal faces, $N$ external legs. 
$R$  is the rank of the $\epsilon_{lf}$ incidence matrix, $D$ is the Lie group dimension, and we denote by $v_{max}$ the maximal valency of $d$-bubble interactions. When $F = 0$, face-connectedness imposes $L=1$, and one trivially has $\omega(\cH) = -2$. From now on, we assume $F \geq 1$. Face-connectedness imposes that each line of $\cH$ appears in at least one of its internal faces. For $1 \leq k \leq v_{max} / 2$, $n_{2 k}$ is the number of bubbles with valency $2 k$ in $\cG$. 
We are particularly interested in determining which values of $d$, $D$ and $v_{max}$ are likely to support just-renormalizable theories.

Remember that 
the incidence matrix has entries $0, +1$ or $-1$ since the graphs we consider have no tadfaces.

Since we are going to make extensive use of contractions of graphs along trees, as a way to gauge fix the amplitudes\footnote{We use the term gauge fixing in the sense of lattice gauge theory or spin foam models: it is the procedure by which we eliminate the redundant group variables appearing in the amplitudes.}, we first establish the change in divergence degree under such a contraction.

\begin{lemma}
Under contraction of a tree  $\cT$, $F$ and $R$ each do not change so that 
$[F-R] (\cH) =[ F-R] (\cH / \cT)$.
\end{lemma}
\prf  That $F$ does not change is easy to show: existing faces can only get shorter under contraction of a tree line but cannot disappear
(this is true also for \emph{open} faces). 

$R$ does not change because of the tree-gauge invariance. This fact can be shown in very concrete terms. Given a tree $\cT$ with $\vert\cT \vert = T= V-1$ lines
we can define the $L \times T$ matrix $\eta_{l,\ell}$ which has entries $0, +1$ or $-1$ in the following way: for any oriented line $l=(v,v')$
we consider the \emph{unique path} $P_\cT(l)$ in the tree $\cT$ going from vertex $v$ to $v'$ and define $\eta_{l,\ell}$ to be zero if this path does not contain
$\ell$ and $\pm 1$ if it does, the sign taking into account the orientations of the path and of the line $\ell$. Remark that $\eta_{\ell\ell}=1$ for all $\ell$.

Then for each (closed) face $f$, made of $l_1,\dots, l_p$, it is easy to check that the induced loop on $\cT$ made of gluing
the paths $P_\cT(l_1), \dots ,  P_\cT(l_p)$, which is contractible, must take each tree line $\ell$ 
an equal number of times and with opposite signs so that 
\beq E(f, \ell) =\sum_l \epsilon_{lf}\eta_{l,\ell} = 0, \quad \to  \epsilon_{\ell f} = -  \sum_{l \not = \ell} \epsilon_{lf}\eta_{l,\ell} .
\eeq 
Therefore the line $\epsilon_{\ell f}$ is a combination of the other lines, and the incidence matrix after contracting $\ell$, which has one line less, but one which was a linear combination of the other ones, maintains the same rank.
\qed

We shall consider now a \textit{tensorial rosette} \cite{addendum}, namely the subgraph $\cH / \cT$ obtained after contraction of a \textit{spanning tree}; it has $L- (V-1)$ lines and a single vertex. The goal is to gain a better control over its degree of divergence and the various contributions to it. The key procedure to achieve the goal is to apply $k$-dipole contractions to the tensorial rosette, and establish how they affect the divergence degree. 

Note that $\cH / \cT$ is not necessarily face-connected, since the contraction of tree lines affects how faces are connected to one another. 
Recall also that a line is a $k$-dipole if it belongs to exactly $(k-1)$ faces of length $1$.

Under a $k$-dipole contraction we know that a single line and possibly several faces disappear, hence the rank
of the incidence matrix can either remain the same or go down by 1 unit. Moreover
only the faces of length 1 can eventually disappear; if there exist such faces
the rank must go down by exactly 1, since we delete a column which is not a combination of the others.

\begin{itemize}  
\item $F \to F- (k - 1)$ and $R \to R-1$, hence $F-R \to F- R- (k-2)$  if $k \ge 2$,

\item $F \to F$, and $R \to R$ or $R \to R-1$,  hence $F-R \to F- R$ or $F-R \to F- R+1$ 
if $k=1$.
\end{itemize}

By definition, a rosette (with external legs) is a melopole if 
and only if there is an ordering of its lines such that all contractions are $d$-dipoles. In that case, we find that $F-R =  (d-2)[L- (V-1)]$. If the rosette is not a melopole, there is at least one step where $F-R$ decreases by less than $(d-2)$, so we expect such a subgraph to be suppressed with respect to a melopole. However, $k$-dipole contractions with $k < d$ need not conserve vertex-connectedness, so we need to refine this argument. To do so, we write the divergence degree of any rosette in terms of the quantity
\beq
\rho \equiv F - R - \left( d-2 \right) \widetilde{L} \,,
\eeq
where $\widetilde{L}$ is the number of lines of the rosette. It will be convenient in the following to consider (vertex)-disjoint unions of rosettes, to which $\rho$ is extended by linearity. These disjoint unions of rosettes will simply be called rosettes from now on, and their single-vertex components will be said to be \textit{connected}.

Since $\widetilde{L} = L - V + 1$ is the number of lines of any rosette of the graph $\cH$, and $F-R$ does not depend on $\cT$ either, we know that $\rho(\cH / \cT)$ is independent of $\cT$. It is therefore a function on equivalent classes of rosettes. This way we obtain a nice splitting of $\omega$, between a rosette dependent contribution and additional combinatorial terms capturing the characteristics of the initial graph:
\beq
\omega(\cH) = D \left(d - 2 \right) + \left[ D \left(d - 2 \right) - 2 \right] L - D \left(d - 2 \right) V + D \rho( \cH / \cT ) \,.
\eeq
The first three terms do not depend on the rank $R$, and provided that $\rho$ can be understood, will give a simple classification of divergences. To establish this central result about the values of $\rho$, one first needs to prove a technical lemma, about $1$-dipole contractions. 
\begin{lemma}\label{d1}
Let $\cG$ be a face-connected rosette (with $F(\cG) \geq 1$), and $\ell$ a $1$-dipole line in $\cG$. If $\cG / \ell$ has more vacuum connected components than $\cG$, then
\beq
R( \cG / \ell ) = R ( \cG ) - 1 \,.
\eeq 
\end{lemma}
\begin{proof}
As stated before, such a move either lowers $R$ by $1$ or leaves it unchanged. We just have to show that given our hypothesis, we are in the first situation. 
We first remark that lines and faces can be oriented in such a way that $\epsilon_{lf} = +1$ or $0$. We can for instance positively orient lines from white to black nodes, and faces accordingly. 
With this convention, we can exploit the colored structure of the graphs in the following way: for any color $1 \leq i \leq d$, each line appears in exactly one face of color $i$. For vacuum graphs, all these faces are closed
and correspond to entries in the $\epsilon_{lf}$ matrix, implying
\beq
\sum_{f \, {\rm{of} \, \rm{color}} \, i } \epsilon_{l f} = 1
\eeq
for any $i$ and any $l$.
Given the hypothesis on $\cG / \ell$, we know that up to permutations of lines and columns, $\epsilon_{lf}$ takes the form:
\[
\left(
\begin{array}{c|c}
  \raisebox{-10pt}{{\huge\mbox{{$M_1$}}}} & \raisebox{-10pt}{{\huge\mbox{{$0$}}}} \\ \hline
  \ast \, \cdots \, \ast \, 1 & 1 \, \varepsilon_2 \, \cdots \, \varepsilon_{d - 1}\,  0\,  \cdots\,  0 \\ \hline
  \raisebox{-10pt}{{\huge\mbox{{$0$}}}} & \raisebox{-10pt}{{\huge\mbox{{$M_2$}}}} \\
\end{array}
\right)
\]
where $M_2$ is the $\epsilon_{lf}$ matrix of a vacuum graph, one of the additional vacuum components created by the contraction of $\ell$. $M_1$ is the $\epsilon_{lf}$ matrix associated to the complement (possibly several connected components) in $\cG / \ell$.
The additional line corresponds to $\ell$, and because $\cG$ is face-connected, it must contain at least a $1$ under $M_1$, and a $1$ above $M_2$. This leaves up to $d-2$ additional non-trivial entries in this line above $M_2$, denoted by the variables $\varepsilon_i = 0$ or $1$. Let us call
$i_1$ the color of the face associated to the first column of $M_2$. Non-zero $\varepsilon$'s are necessarily associated to different colors: call them $i_2$ up to $i_{d-1}$. This implies that the remaining color, $i_d$, only appears in faces
of $M_2$ that do not intersect with $\ell$. Calling $C_f$ the columns of $M_2$, and $C_{f_1}$ its first column, one has:
$$
C_{f_1} + \sum_{f \, {\rm{of} \, \rm{color}} \, i_1 \,; \,f \neq f_1} C_f = \sum_{f \, {\rm{of} \, \rm{color}} \, i_d \,} C_f\,.
$$     
The operation $$C_{f_1} \to C_{f_1} + \sum_{f \, {\rm{of} \, \rm{color}} \, i_1 \,; \,f \neq f_1} C_f - \sum_{f \, {\rm{of} \, \rm{color}} \, i_d \,} C_f$$
cancels the first column of $M_2$, and when operated on the whole matrix does not change the line $\ell$. 
We conclude that $R(\cG) = {\rm{rank}}(M_1) + {\rm{rank}}(M_2) + 1 = R( \cG / \ell) + 1$. 
\end{proof}

The essential property of the quantity $\rho$ is that it is bounded from above, and is extremal for melopoles. More precisely we have:
\begin{proposition}\label{rho}
Let $\cG$ be a connected rosette. 
\begin{enumerate}[(i)]
\item If $\cG$ is a vacuum graph, then 
$$ \rho(\cG) \leq 1$$
and
$$ \rho (\cG) = 1 \Leftrightarrow \cG \; \mathrm{is} \; \mathrm{a} \; \mathrm{melopole}\,.$$
\item If $\cG$ is not a vacuum graph, i.e. has external legs, then 
$$ \rho(\cG) \leq 0$$
and
$$ \rho (\cG) = 0 \Leftrightarrow \cG \; \mathrm{is} \; \mathrm{a} \; \mathrm{melopole}\,.$$
\end{enumerate}
\end{proposition}
\begin{proof}
It is easy to see that $\rho$ is conserved under $d$-dipole contractions. In particular a simple computation shows that $\rho(\cG) = 1$ when $\cG$ is a vacuum melopole, and $\rho(\cG) = 0$ when $\cG$ is a non-vacuum melopole. 
We can prove the general bounds and the two remaining implications in (i) and (ii) by induction on the number of lines $L$ of the rosette $\cG$. 
\begin{itemize}
 \item If $L = 1$, $\cG$ can be both vacuum or non-vacuum. In the first situation, $\cG$ cannot be anything else than the
fundamental melon with two nodes. It has exactly $d$ faces, a rank $R = 1$, so that $\rho (\cG) = 1$. In the second situation,
namely when $\cG$ is non-vacuum, the number of faces is strictly smaller than $d$, as at least one strand running through the single line of $\cG$ must correspond to an external face. Since on the other hand 
the rank is $0$ when $F(\cG) = 0$ and $1$ otherwise, we see that $\rho (\cG) \leq 0$, and $\rho (\cG) = 0$ whenever the number of faces is exactly $(d-1)$. In this case, the unique line of $\cG$ is a $d$-dipole, therefore $\cG$ is a melopole.
 \item Let us now assume that $L \geq 2$ and that properties (i) and (ii) hold for a number of lines $L' \leq L -1$. If $\cG$ is not face-connected (and therefore non-vacuum), we can decompose it into face-connected components $\cG_1 , \ldots , \cG_k$ with $k \geq 2$. Each of these components has a number of lines strictly smaller than $L$, so by induction hypothesis $\rho(\cG) = \underset{i}{\sum} \rho( \cG_i ) \leq 0$. Moreover, $\rho( \cG ) = 0$ if and only if $\rho ( \cG_i ) = 0$ for any $i$, in which case $\cG$ is a melopole since all the $\cG_i$'s are themselves melopoles. This being said, we assume from now on that $\cG$ is face-connected, and pick up a $k$-dipole line $\ell$ in $\cG$ ($1 \leq k \leq d$). 

Let us first suppose that $k \geq 2$. 
$\cG / \ell$ has $\widetilde{L} \leq L - 1$ lines in its rosettes, and $(F-R) (\cG / \ell ) = (F - R) (\cG) - (k - 2)$, which implies $\rho(\cG / \ell ) \geq \rho( \cG ) + (d - k)$ (with equality if and only if $\cG / \ell$ is itself a rosette). Moreover, $\cG / \ell$ is possibly disconnected and consists in $q$ vertex-connected components with $1 \leq q \leq d-k + 1$, yielding $q$ connected rosettes (after possible contractions of tree lines).
By the induction hypothesis, we therefore have $\rho( \cG ) \leq q - (d -k) \leq 1$, and $\rho(\cG) = 1$ if and only
if $\cG / \ell$ consists of $d - k + 1$ connected vacuum melopoles, in which case $\cG$ itself is a vacuum melopole. Similarly, $\rho(\cG) = 0$ if and only if $\cG / \ell$ consists of $d-k$ vacuum melopoles
and $1$ non-vacuum melopole, in which case $\cG$ is a non-vacuum melopole. 

If $k =1$, we either have $R (\cG / \ell) = R ( \cG ) - 1$ or $R (\cG / \ell) = R ( \cG )$, which respectively imply $\rho( \cG ) \leq \rho( \cG / \ell) - (d - 1)$ or $\rho( \cG ) \leq \rho( \cG / \ell) - (d - 2)$. 
The first situation is strictly analogous to the $k \geq 2$ case, therefore the same conclusions follow. In the second
situation, we resort to lemma \ref{d1}. Since $\cG$ has been assumed face-connected, and $L \geq 2$ implies $F(\cG) \geq 1$, the lemma is applicable: $\cG / \ell$ cannot have more vacuum connected components than $\cG$. In particular, if $\cG$ is non-vacuum, $\rho( \cG / \ell ) \leq 0$, 
therefore $\rho( \cG ) \leq - (d - 2) < 0$. Likewise, $\rho(\cG) \leq 0$ when $\cG$ is vacuum. 

We conclude that the two properties (i) and (ii) are true at rank $L$.
\end{itemize}
\end{proof}

\begin{corollary}
Let $\cH$ be a vertex-connected subgraph. If $\cH$ admits a melopole rosette (in particular, if $\cH$ is melonic), then all its rosettes are melopoles.
\end{corollary}
\begin{proof}
The quantity $\rho(\cH / \cT)$ is independent of the particular spanning tree $\cT$ one is considering. Therefore, if $\cH / \cT$ is a melopole then this holds for any other spanning tree $\cT'$. 
\end{proof}

\subsection{Just-renormalizable models}

We are now in good position to establish a list of potentially just-renormalizable theories. Indeed, by simply rewriting $L$ and $V$ as
\beq
L = \sum_{k = 1}^{v_{max} / 2} k \, n_{2k} - \frac{N}{2} \;, \qquad V = \sum_{k = 1}^{v_{max} / 2} n_{2 k}\,,
\eeq
one obtains the following bound on the degree of non-vacuum face-connected subgraphs:

\beq
\omega \leq D \left( d - 2 \right) - \frac{ D(d - 2) - 2}{2} N 
+ \sum_{k = 1}^{v_{max} / 2} \left[ \left( D(d-2) - 2 \right) k - D \left( d - 2 \right) \right] n_{2 k}\,.
\eeq
Since we also know this inequality to be saturated (by melonic graphs), it yields a necessary condition for just-renormalizable theories:
\beq\label{condition}
v_{max} = \frac{2 D (d-2)}{D(d-2) - 2}\,,
\eeq
and in such cases
\beq\label{vertex_true}
\omega = \frac{D (d - 2) - 2}{2}\left( v_{max} - N \right) - \sum_{k = 1}^{v_{max}/2 - 1} \left[ D \left( d - 2 \right) - \left( D(d-2) - 2 \right) k\right] n_{2 k} + D \rho \,.
\eeq
We immediately deduce that only $n$-point functions with $n \leq v_{max}$ can diverge, which is a necessary condition for renormalization. Equation (\ref{condition}) has exactly five non-trivial solutions (i.e. $v_{max} > 2$), which yields five classes of potentially
just-renormalizable interacting theories. Two of them are $\vphi^6$ models, the three others being of the $\vphi^4$ type. A particularly interesting model from a quantum gravity perspective is the $\vphi^6$ theory with $d=3$ and $D=3$,
which can incorporate the essential structures of 3d quantum gravity (model A)\footnote{The relevance of the other cases, in particular the four dimensional case C, for quantum gravity is uncertain. Current (T)GFT models for 4d quantum gravity \cite{GFT1,GFT2,GFT3,EPRL,BO-Immirzi}, in fact, are not given by simple field theories on a group manifold but, due to the simplicity constraints,  either by functions on homogeneous spaces (obtained by the quotient of the Lorentz group $SO(3,1)$ or the rotation group $SO(4)$ by an $SO(3)$ subgroup) or by functions on the full group but subject to the condition that only their value on a submanifold of the same is dynamically relevant. As it stands, therefore, the above analysis does not apply, and a new analysis should be performed.}. We will focus on this case in the following sections, but we already notice that the same methods could as well be applied to any of the four other types of
candidate theories. Table \ref{theories} summarizes the essential properties of these would-be just-renormalizable theories, called of type A up to E.

\begin{table}[h]
\centering
\begin{tabular}{| c || c | c | c | c |}
    \hline
  Type & $d$ & $D$ & $v_{max}$ & $\omega$  \\ \hline\hline
A & 3 & 3 & 6 & $3 - N/ 2 - 2 n_2 - n_4 + 3 \rho$ \\ \hline
B & 3 & 4 & 4 & $4 - N - 2 n_2 + 4 \rho$ \\ \hline
C & 4 & 2 & 4 & $4 - N - 2 n_2 + 2 \rho$\\ \hline
D & 5 & 1 & 6 & $3 - N/ 2 - 2 n_2 - n_4 + \rho$ \\ \hline
E & 6 & 1 & 4 & $4 - N - 2 n_2 + \rho$\\ 
    \hline
  \end{tabular}
\caption{Classification of potentially just-renormalizable models.}
\label{theories}
\end{table}

Models D and E have been studied and shown renormalizable in \cite{SVT}.
Non-vacuum divergences of models A and B will only have melonic contributions, while models C, D and E will also include submelonic terms. There could be: up to $\rho = -1$ divergent $2$-point graphs in model C; up to
$\rho = -2$ divergent $2$-point graphs and $\rho = -1$ divergent $4$-point graphs in model D; up to $\rho = -2$ divergent $2$-point graphs in model C. These require a (presumably simple) refinement of proposition \ref{rho}.
As for models A and B, we do not need any further understanding of $\rho$.

\

Finally, one also remarks that face-connectedness did not play any role in the derivation of expression (\ref{vertex_true}). Indeed, it is as well valid for vertex-connected unions of non-trivial face-connected subgraphs, which as we will see in the last section of this article, is also relevant to renormalizability.

\subsection{Properties of melonic subgraphs}

Since they will play a central role in the remainder of this paper, we conclude this section by a set of properties verified by melonic subgraphs, especially non-vacuum ones.

\

The first thing one can notice is that by mere definition, any line in a melonic subgraph $\cH$ is part of an internal face in $F( \cH )$. This means in particular that $\cH$ cannot be split in two vertex-connected parts connected by a single $1$-dipole line $e$, since the three faces running through $e$ would then necessarily be external to $\cH$. In other words:
\begin{lemma}\label{1PI}
Any melonic subgraph $\cH \subset \cG$ is $1$-particle irreducible.
\end{lemma}
From the point of view of renormalization theory, this is already interesting, as $2$-point divergences in particular will not require any further decomposition into $1$-particle irreducible components.

We now turn to specific properties of non-vacuum melonic subgraphs. In order to understand further their possible structures, it is natural to first focus on their rosettes. The following proposition shows that they cannot be arbitrary melopoles.

\begin{proposition}\label{face_rosette}
Let $\cH \subset \cG$ be a non-vacuum melonic subgraph. For any spanning tree $\cT$ in $\cH$, the rosette $\cH / \cT$ is face-connected. 
\end{proposition}
\begin{proof}
Let $\cT$ be a spanning tree in $\cH$ and let us suppose that $\cH / \cT$ has $k \geq 2$ face-connected components. In order to find a contradiction, one first remarks that the contraction of a tree conserves the number of faces, and even elementary face connections. That is to say: if $l_1 , l_2 \in L( \cH ) \setminus \cT$ share a face in $\cH$, they also share a face in $\cH / \cT$. Therefore, the lines of $L ( \cH ) \setminus \cT$ can be split into $k$ subsets, such that each one of them does not share any face of $\cH$ with any other. We give a pictorial representation of what we mean on the left side of Figure \ref{propo6_1}, with $k = 4$. The internal structure of the vertices is ommited, the $4$ subsets of lines are marked with different symbols, while the tree lines are left unmarked.    
\begin{figure}[h]
\begin{center}
\includegraphics[scale=0.6]{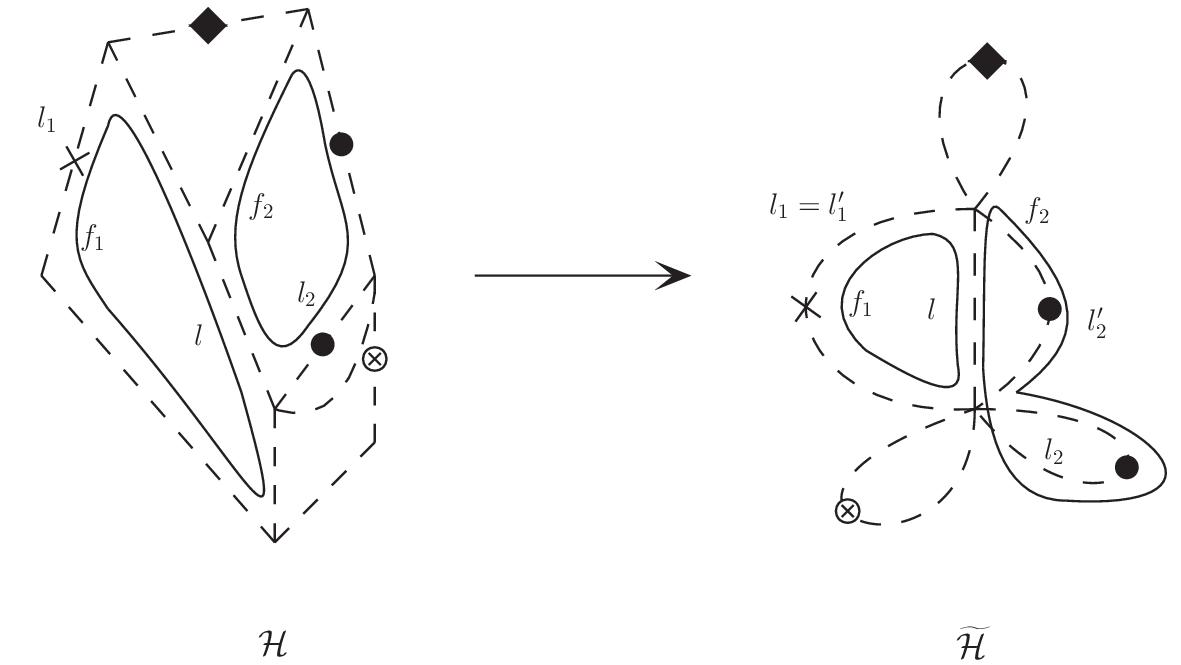}
\caption{Simplified representation of a melonic graph $\cH$ and its contraction $\widetilde{\cH}$.}
\label{propo6_1}
\end{center}
\end{figure}
The face-connectedness of $\cH$ is ensured by the tree lines, which must connect together these $k$ subsets. Incidentally, there must be at least one line $l \in \cT$ which is face-connected to two or more of these subsets. In particular\footnote{At this point we rely on $F(\cT) = \emptyset$, which holds because $\cT$ is a tree.}, we can find two faces $f_1$ and $f_2$ which are face-disconnected in $\cH / \cT$, and two lines $l_1 , l_2 \in \cH \setminus \cT$ such that: $l \in f_1 \cap f_2$, $l_1 \in f_1$ and $l_2 \in f_2$. See again the left side of Figure \ref{propo6_1}, where $f_1$ and $f_2$ are explicitly represented as undashed closed loops.  

\begin{figure}[h]
\begin{center}
\includegraphics[scale=0.6]{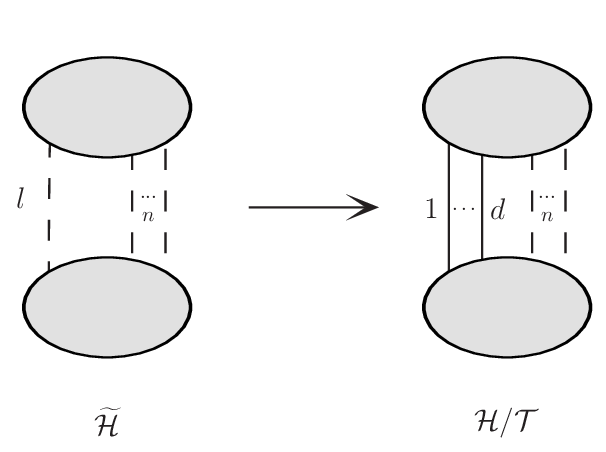}
\caption{Last step of a contraction of a spanning tree in a melonic subgraph.}
\label{melonic_last_step}
\end{center}
\end{figure}

Now, call $\widetilde{\cH}$ the subgraph obtained after contraction of all the tree lines but $l$, i.e. $\widetilde{\cH} \equiv \cH / (\cT \setminus \{ l \})$. $\widetilde{\cH}$ consists of two vertices, connected by $l$ and a certain number $n$ of lines from $L(\cH) \setminus \cT$ (see the right part of Figure \ref{propo6_1} and the left part of Figure \ref{melonic_last_step}). Through $l$ run at least two faces, $f_1$ and $f_2$. $l$ is their single connection, since they are disconnected in $\widetilde{\cH} / \{ l \} = \cH / \cT$. This requires the existence of two $1$-dipole lines $l_1'$ and $l_2'$ in $\widetilde{\cH} \setminus \{ l \}$, through which $f_1$ and $f_2$ respectively run. In Figure \ref{propo6_1} we see that $l_1 ' = l_1$, but because $l_2$ is a tadpole line in $\widetilde{\cH}$, we must choose $l_2 ' \neq l_2$. Otherwise, $f_1$ and $f_2$ could not close without being connected in $\cH  / \cT$. The colored extension of $\cH / \cT$ can thus be split into two groups of nodes, connected by $n \geq 2$ lines and $d$ colored lines (created by the contraction of $l$, see the right part of Figure \ref{melonic_last_step}). It is easy to understand that such a drawing cannot correspond to a melopole. Indeed, the number of colored lines connecting the two groups of nodes would need to be at least $n (d - 2) + 1$.\footnote{A simple way to understand this last point is the following. Suppose there are $p$ colored lines between the two groups of nodes. If none of the $n$ lines between the two groups of nodes are elementary melons, an elementary melon can be contracted in one of them, without affecting the $n$ lines nor the $p$ colored lines between them. If on the contrary one of the $n$ lines is an elementary melon, it can be contracted. This cancels $d-1$ colored lines connecting the two groups of nodes, and replaces it by a single one. Hence $n \to n-1$ and $p \to p - (d - 2)$. By induction, one must therefore have $p - n (d - 2) \geq 1$, where the $1$ on the right side is due to the last step $n = 1 \to n = 0$.} Hence $d \geq n (d - 2) + 1$, from which we deduce: 
\beq
d \leq \frac{2n - 1}{n -1}\,.
\eeq
When $n \geq 3$, this is incompatible with $d \geq 3$, and $n=2$ is also incompatible with $d \geq 4$. If $n = 2$ and $d = 3$, a contradiction also arises, thanks to the colors. In the process of elementary melon contractions, the first of the two lines to become elementary will delete $2$ colored lines, say with colors $1$ and $2$, and replace it by a color-$3$ line. One therefore obtains two groups of nodes connected by two color-$3$ lines and a single color-$0$ line, which cannot form an elementary melon. See Figure \ref{propo6_2}.
\begin{figure}[h]
\begin{center}
\includegraphics[scale=0.6]{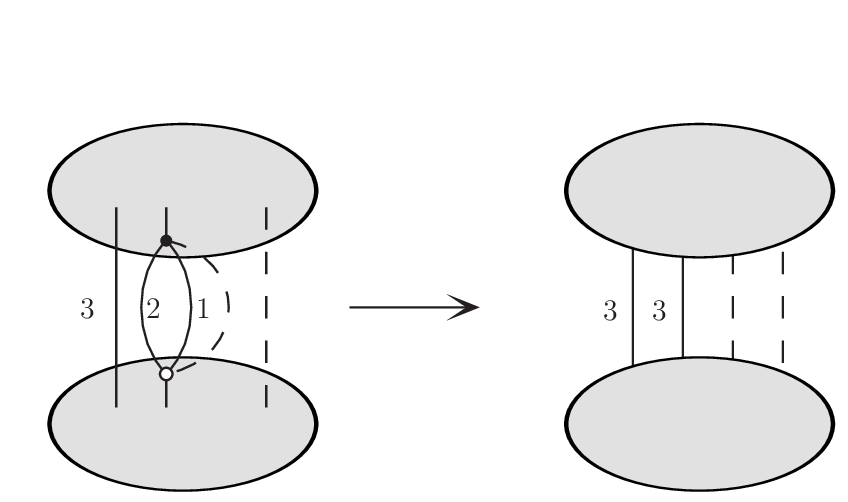}
\caption{Contraction of an elementary melon in a $3$-colored rosette with $n = 2$.}
\label{propo6_2}
\end{center}
\end{figure} 
\end{proof}

One immediately notices that this proposition also holds for any forest $\cF \subset \cH$, that is any set of lines without loops, be it a spanning tree or not. Indeed, any such $\cF$ is included in a spanning tree $\cT$. The contraction of $\cF$ on the one hand can only increase the number of face-connected components, and on the other hand the full contraction of $\cT$ leads to a single face-connected components, hence the contraction of $\cF$ also leads to a single face-connected component. 

We provide an illustration of this result in Figure \ref{rosette_propo6}, representing a melonic graph and one of its rosettes, which is face-connected.
\begin{figure}[h]
\begin{center}
\includegraphics[scale=0.6]{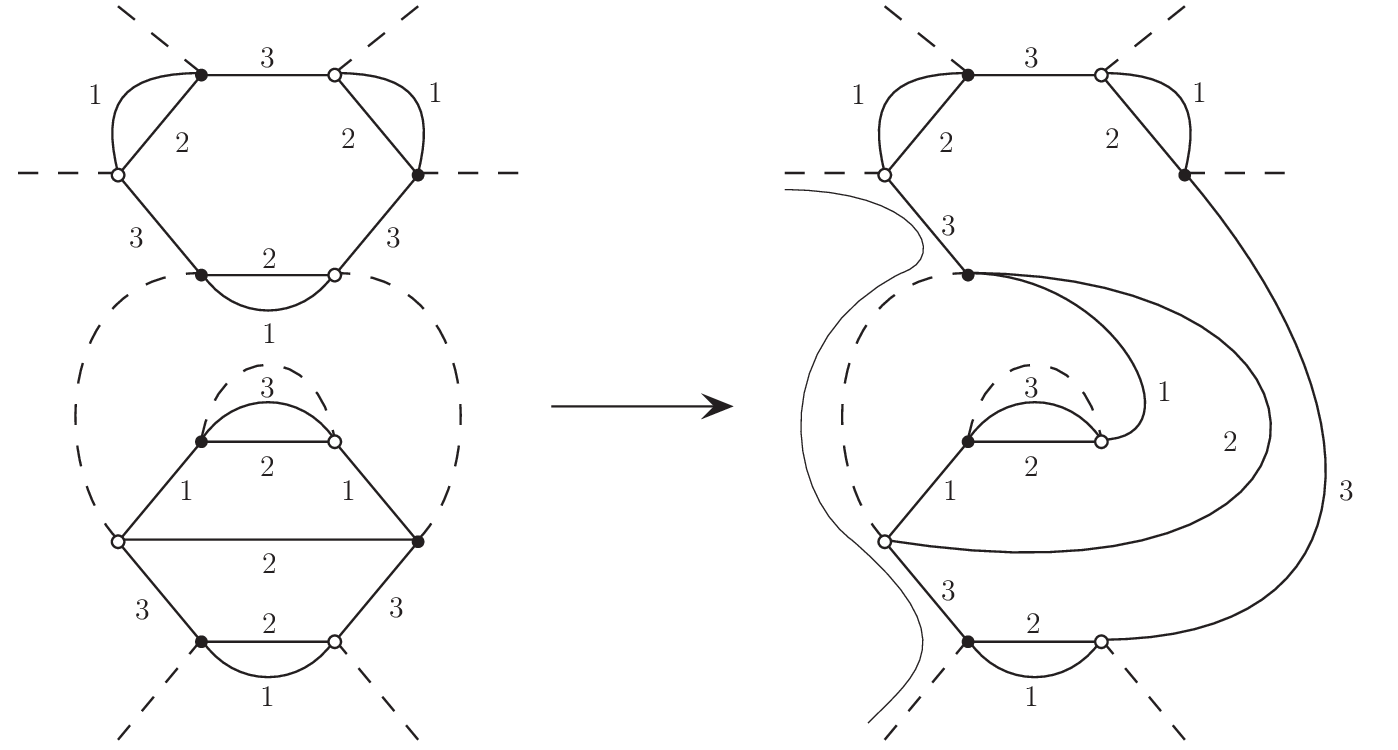}
\caption{A melonic graph (left) and one of its rosettes (right). The latter is face-connected (Proposition \ref{face_rosette}) and has a single external face (Corollary \ref{coro2}), represented as a thin line.}
\label{rosette_propo6}
\end{center}
\end{figure} 
\

A more important consequence of this statement is a restriction on the number of external faces of the rosettes:
\begin{corollary}\label{coro2}
Let $\cH \subset \cG$ be a non-vacuum melonic subgraph. For any spanning tree $\cT$ in $\cH$, $F_{ext}( \cH / \cT ) = 1$.
\end{corollary}
\begin{proof}
Let us prove that any face-connected melopole $\widetilde{H}$ has a single external face. $\cH / \cT$ being itself face-connected thanks to the previous proposition, the result will immediately follow. We proceed by induction on $L(\widetilde{\cH})$. The elementary melon has $d-1$ internal faces and $1$ external face, hence the property holds when $L(\widetilde{\cH}) = 1$. If $L(\widetilde{\cH}) \geq 2$, we can contract an elementary $d$-dipole line $l$ in $\widetilde{\cH}$. The subgraph $\{ l \}$ has $1$ external face, but it is internal in $\widetilde{\cH}$, otherwise the latter would not be face-connected. Hence $\widetilde{\cH}$ and $\widetilde{\cH} / \{ l \}$ have the same number of external faces. By the induction hypothesis, $\widetilde{\cH} / \{ l \}$ (which is a face-connected melopole) has a single external face, and so do $\widetilde{\cH}$.
\end{proof}

We illustrate this result again in Figure \ref{rosette_propo6}. Such restrictions on the rosettes constrain the face structure of the initial melonic graphs themselves.

\begin{proposition}\label{propo_color}
Let $\cH \subset \cG$ be a non-vacuum melonic subgraph. All the external faces of $\cH$ have the same color.
\end{proposition}
\begin{proof}
Suppose $F_{ext} ( \cH ) \geq 2$. Let us choose two distinct external faces $f_1$ and $f_2$, and show that they are of the same color. We furthermore select a line $l_1 \in f_1$, and a spanning tree $\cT$ in $\cH$ such that $l_1 \notin \cT$. This is possible thanks to lemma $\ref{1PI}$, and this guarantees that the unique external face of $\cH / \cT$ is $f_1$. This also means that in $\cH$, $f_2$ only runs through $\cT$, otherwise it would constitute a second face in $\cH / \cT$. We can in particular pick a line $l_2 \in f_2 \cap \cT$. See Figure \ref{propo7} for an example, in which we use the same simplified representation as before, except that the external faces we are interested in have open ends. 
\begin{figure}[h]
\begin{center}
\includegraphics[scale=0.6]{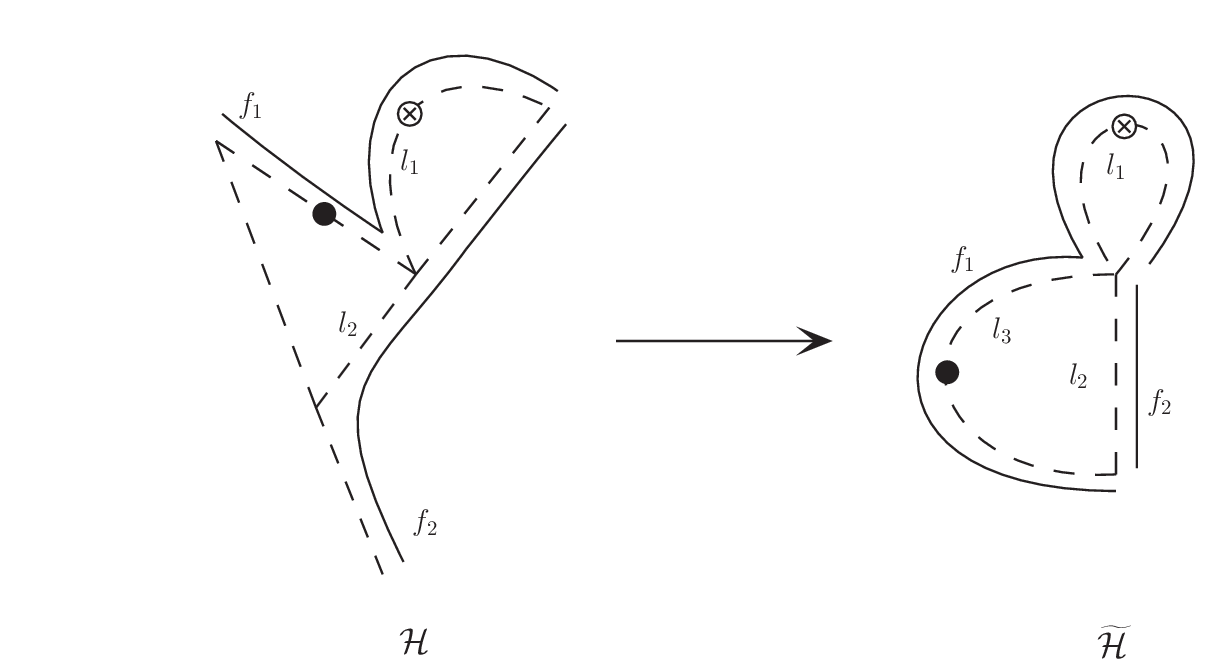}
\caption{A melonic graph $\cH$ and its contraction $\widetilde{\cH}$.}
\label{propo7}
\end{center}
\end{figure} 
 Similarly to the strategy followed in the proof of Proposition \ref{face_rosette}, define $\widetilde{\cH} \equiv \cH / ( \cT \setminus \{ l_2 \})$. As was already explained, $\widetilde{\cH}$ consists of two vertices, connected by $l_2$ and at most one extra line (see Figure \ref{melonic_last_step}, with $n = 1$). There cannot be just $l_2$ connecting these two vertices, because $\widetilde{\cH}$ is $1$-particle irreducible, hence there are exactly two such lines. Call $l_3$ the second of these lines (it is not necessarily possible to choose $l_3 = l_1$, see Figure \ref{propo7}). They must have at least $(d-1)$ internal faces in common, otherwise $\widetilde{\cH} / \{ l_2 \} = \cH / \cT$ would not be a melopole. They moreover cannot have $d$ internal faces in common, otherwise $\cH / \cT$ would be vacuum. This means that both appear in external faces of the same color. One of them is of course $f_2$ (which goes through $l_2$), and the second (which goes through $l_3$) is either $f_1$, again $f_2$, or yet another external face. $f_2$ is excluded because by construction it had no support on $\cH \setminus \cT$. Moreover, $\widetilde{\cH}$ must have exactly two external faces, since only one is deleted when contracting $l_2$ and the resulting rosette $\cH / \cT$ has itself a single external face (by Corollary \ref{coro2}). Hence the external face running through $l_3$ can only be $f_1$, and we conclude that it has the same color than $f_2$.
\end{proof}
This property is quite useful in practice because it implies a restrictive bound on the number of external faces of a melonic subgraph in terms of its number of external legs. 
\begin{corollary}\label{faces_legs}
A melonic subgraph with $N$ external legs has at most $\frac{N}{2}$ external faces.
\end{corollary}
\begin{proof}
In any vertex-connected graph with $N$ external legs, the number of external faces of a given color is bounded by $\frac{N}{2}$.
\end{proof}
Figure \ref{rosette_propo6} provides a good example of a melonic graph having more external faces that its rosettes: while the rosette on the right side has a single external face (in agreement with Corollary \ref{coro2}), the graph on the left side has two external faces, and they both have the same color $3$ (in agreement with Proposition \ref{propo_color}). 

\

Finally one would like to understand the inclusion and connectivity relations between all divergent subgraphs of a given non-vacuum graph. This is a very important point to address in view of the perturbative renormalization of such models, in which divergent subgraphs are inductively integrated out. As usual, the central notion in this respect is that of a "Zimmermann" forest, which we will generalize to our situation (where face-connectedness replaces vertex-connectedness) in Section \ref{sec:finiteness}. At this stage, we just elaborate on some properties of melonic subgraphs which will later on help simplifying the analysis of "Zimmermann" forests of divergent subgraphs. 
%



\begin{proposition}\label{curiosity}
Let $\cG$ be a non-vacuum vertex-connected graph. If $\cH_1 , \cH_2 \subset \cG$ are two melonic subgraphs, then: 
\begin{enumerate}[(i)]
\item $\cH_1$ and $\cH_2$ are line-disjoint, or one is included into the other.
\item If $\cH_1 \cup \cH_2$ is melonic, then: $\cH_1 \subset \cH_2$ or $\cH_2 \subset \cH_1$.
\end{enumerate} 
Moreover, any $\cH_1 , \ldots , \cH_k  \subset \cG$ melonic are necessarily face-disjoint if their union $\cH_1 \cup \ldots \cup \cH_k$ is also melonic.
\end{proposition}
\begin{proof}
Let us first focus on (i) and (ii). To this effect, we assume that: (i) $\cH_1 \cap \cH_2 \neq \emptyset$ (and in particular $\cH_1$ and $\cH_2$ are face-connected in their union); (ii) $\cH_1$ and $\cH_2$ are face-connected in $\cH_1 \cup \cH_2$, and the latter is also melonic. We need to prove that in these two situations, $\cH_1 \subset \cH_2$ or $\cH_2 \subset \cH_1$. In order to achieve this, we suppose that both $\widetilde{\cH}_1 \equiv \cH_1 \setminus ( \cH_1 \cap \cH_2 )$ and $\widetilde{\cH}_2 \equiv \cH_2 \setminus ( \cH_1 \cap \cH_2 )$ are non-empty, and look for a contradiction. 

Let $f_1$ be an arbitrary external face of $\cH_1$. Choose a line $l_1 \in f_1$, and a spanning tree $\cT_1$ in $\cH_1$, such that $l_1 \notin \cT_1$. Then the unique face of $\cH_1 / \cT_1$ is $f_1$. We want to argue that $( \cH_1 \cup \cH_2 ) / \cT_1 = ( \cH_1 / \cT_1 ) \cup \widetilde{\cH}_2$ is face-connected. In situation (ii), this is guaranteed by Proposition \ref{face_rosette} (applied to $\cH_1 \cup \cH_2$). In situation (i) on the other hand, one can decompose it as a disjoint union of subgraphs as follows:
\beq
( \cH_1 \cup \cH_2 ) / \cT_1 = \widetilde{\cH}_1 / (\cT_1 \cap \widetilde{\cH}_1) \sqcup (\cH_1 \cap \cH_2) / (\cT_1 \cap \cH_2) \sqcup \widetilde{\cH}_2 \,. 
\eeq 
The key thing to remark is that through each line of $\cH_1 \cap \cH_2$ run at least $d-1$ faces from $F(\cH_1)$, and at least $d-1$ from $F(\cH_2)$. Since at most a total of $d$ faces run through each line (and $d \geq 3$), we conclude that each line of $\cH_1 \cap \cH_2$ appears in at least one face of $F(\cH_1) \cap F(\cH_2)$. Therefore $(\cH_1 \cap \cH_2) / (\cT_1 \cap \cH_2)$ has at least one face, and is in particular non-empty. We also know that $\widetilde{\cH}_1 / (\cT_1 \cap \widetilde{\cH}_1) \sqcup (\cH_1 \cap \cH_2) / (\cT_1 \cap \cH_2) = \cH_1 / \cT_1$ is face-connected, as well as $(\cH_1 \cap \cH_2) / (\cT_1 \cap \cH_2) \sqcup \widetilde{\cH}_2 = \cH_2 / ( \cT_1 \cap \cH_2 )$. Therefore $( \cH_1 \cup \cH_2 ) / \cT_1$ is itself face-connected.
Finally, since $\widetilde{\cH}_2 \neq \emptyset$, this is only possible if an external face of $\cH_1 / \cT_1$ is internal in $( \cH_1 \cup \cH_2 ) / \cT_1$. We conclude that $f_1$ is internal to $( \cH_1 \cup \cH_2 ) / \cT_1$, hence to $\cH_1 \cup \cH_2$.

We have just shown that all the external faces of $\cH_1$ are internal to $\cH_1 \cup \cH_2$. Likewise, all the external faces of $\cH_2$ are internal to $\cH_1 \cup \cH_2$. Therefore $F_{ext} ( \cH_1 \cup \cH_2 ) = \emptyset$, which implies that $\cH_1 \cup \cH_2 = \cG$ is vacuum, and contradicts our hypotheses.

\

We can proceed in a similar way than for (ii) to prove the last statement. Assume $\cH_1 , \ldots , \cH_k$ to be melonic, line-disjoint, and face-connected in their union. The connectedness of $\cH_1 \cup \dots \cup \cH_k$ and any of its reduction by a forest implies that all the external faces of $\cH_i$ are internal in $\cH_1 \cup \dots \cup \cH_k$, for any $1 \leq i \leq k$. Therefore the latter is vacuum, and this again contradicts the fact that $\cG$ is not.  
\end{proof} 

\

\noindent{\bf{Example.}} Figure \ref{overlap} represents two non-trivial melonic graphs $\cH_1$ and $\cH_2$ which are line-disjoint but face-connected in their union. Accordingly, their union is not melonic, as can be checked explicitly. 


\section{The $\SU(2)$ model in three dimensions}
\label{su2_model}

In this section, the $\vphi^6$ model based on the group $\SU(2)$ of type A in Table \ref{theories} is precisely defined. A detailed proof of its renormalizability will follow in the next two sections. 

\subsection{Model, regularization and counter-terms}

From now on, $G = \SU(2)$ and $K_\alpha$ is the corresponding heat kernel at time $\alpha$, which explicitly writes
\beq
K_{\alpha} = \sum_{j \in \mathbb{N}/2} (2 j + 1)\e^{- \alpha j (j+1) } \chi_{j}
\eeq 
in terms of the characters $\chi_j$. We can introduce the cut-off covariance $C^{\Lambda}$
\beq \label{paracut}
C^{\Lambda}(g_1, g_2 , g_3 ; g_1' , g_2' , g_3' )
\equiv \int_{\Lambda}^{+ \infty} \extd \alpha \, \e^{- \alpha m^2} \int \extd h \prod_{\ell = 1}^{3} K_{\alpha} (g_\ell h g_\ell'^{\inv})\,,
\eeq
defined for any $\Lambda > 0$. This allows to define a UV regularized theory, with partition function
\beq
\cZ_{\Lambda} = \int \extd \mu_{C^\Lambda} (\vphi , \vphib) \, \e^{- S_{\Lambda} (\vphi , \vphib )}\,.
\eeq
According to our analysis of the Abelian divergence degree, $S_{\Lambda}$ can contain only up to $\vphi^6$ $d$-bubbles. This gives exactly $5$ possible patterns of contractions (up to color permutations): one $\vphi^2$ interaction, one $\vphi^{4}$ interaction, and three $\vphi^{6}$ interactions. They are represented in Figure \ref{int}. 

Among the three types of interactions of order $6$, only the first two can constitute melonic subgraphs. Indeed, an interaction of the type $(6,3)$ cannot be part of a melonic subgraph, therefore cannot give any contribution to the renormalization of coupling constants. Reciprocally, the contraction of a melonic subgraph in a graph built from vertices of the type $(2)$, $(4)$, $(6,1)$ and $(6,2)$ cannot create an effective $(6,3)$-vertex. This is due to the fact that a $(6,3)$-bubble  is dual to the triangulation of a torus, while the other four interactions represent spheres, and the topology of $d$-bubbles is conserved under contraction of melonic subgraphs \cite{crystallization}. 

Therefore, we can and we shall exclude interactions of the type $(6,3)$ from $S_{\Lambda}$ from now on. This is a very nice feature of the model, for essentially two reasons. First, from a discrete geometric perspective, $(6,3)$ interactions would introduce topological singularities that would be difficult to interpret in a quantum gravity context, so it is good that they are not needed for renormalization. Second, contrary to the other interactions, they are not positive and could therefore induce non-perturbative quantum instabilities.

\begin{figure}[h]
\begin{center}
\includegraphics[scale=0.5]{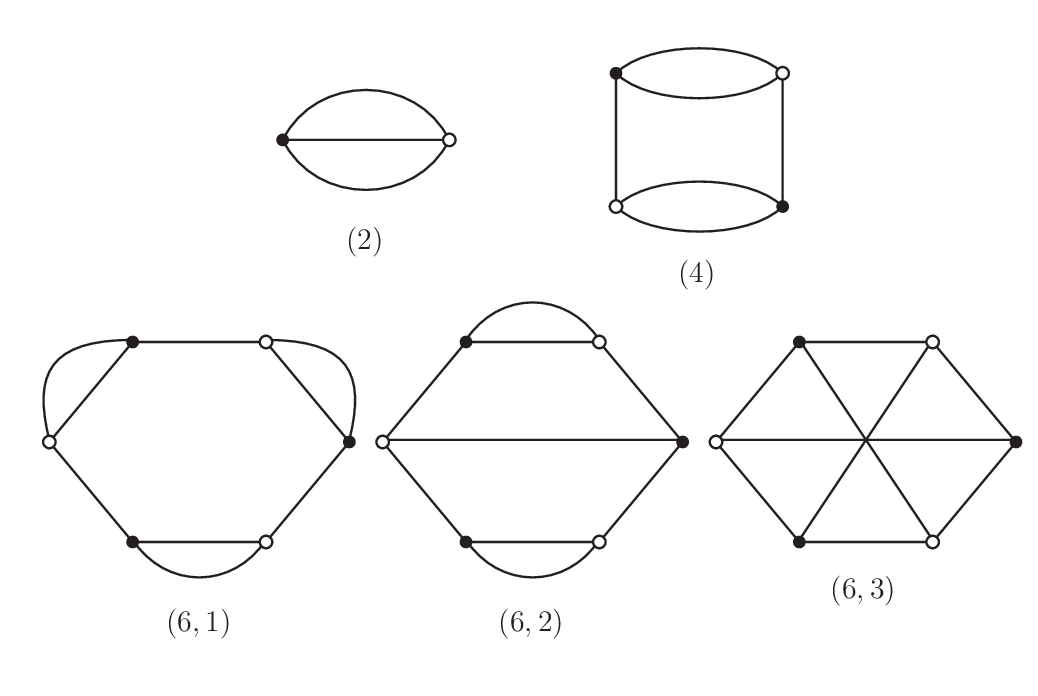}
\caption{Possible $d$-bubble interactions.}
\label{int}
\end{center}
\end{figure}

\
The $2$-point interaction is identical to a mass term, and will therefore be used to implement the mass renormalization counter-terms. Since the model will also generate quadratically divergent $2$-point functions, we also need to include wave function counter-terms in $S_\Lambda$. Finally, we require color permutation invariance of the $4$- and $6$-point interactions. All in all, this gives
\beq\label{drawing_sym}
S_\Lambda  = \frac{t_4^{\Lambda}}{2} S_{4} + \frac{t_{6,1}^{\Lambda}}{3} S_{6,1} + t_{6,2}^{\Lambda} S_{6,2} + CT_{m}^{\Lambda} S_{m} + CT_{\vphi}^{\Lambda} S_{\vphi}\,,
\eeq
where:
\bes\label{color_sym}
S_{4} (\vphi , \vphib) &=& \int [\extd g ]^6 \, \vphi(g_1 , g_2 , g_3 ) \vphib(g_1 , g_2 , g_4 ) \vphi(g_5 , g_6 , g_3 ) \vphib(g_5 , g_6 , g_4 ) + \; {\rm color} \; {\rm permutations}  \,,  \\
S_{6,1} (\vphi , \vphib) &=&  \int [\extd g ]^9 \, \vphi(g_1 , g_2 , g_7 ) \vphib(g_1 , g_2 , g_9 ) \vphi(g_3 , g_4 , g_9 ) \vphib(g_3 , g_4 , g_8 ) \vphi(g_5 , g_6 , g_8 ) \vphib(g_5 , g_6 , g_7 ) \\
&& + \; {\rm color} \; {\rm permutations}  \,, \nn \\
S_{6, 2} (\vphi , \vphib) &=& \int [\extd g ]^9 \, \vphi(g_1 , g_2 , g_3 ) \vphib(g_1 , g_2 , g_4 ) \vphi(g_8 , g_9 , g_4 ) \vphib(g_7 , g_9 , g_3 ) \vphi(g_7 , g_5 , g_6 ) \vphib(g_8 , g_5 , g_6 ) \\
&& + \; {\rm color} \; {\rm permutations} \,, \nn  \\
S_{m} (\vphi , \vphib) &=& \int [\extd g ]^3 \, \vphi(g_1 , g_2 , g_3 ) \vphib(g_1 , g_2 , g_3 )\,,\\
S_{\vphi} (\vphi , \vphib) &=& \int [\extd g ]^3 \, \vphi(g_1 , g_2 , g_3 ) \left( - \sum_{l = 1}^{3} \Delta_\ell \right) \vphib(g_1 , g_2 , g_3 )\,.
\ees

Two types of symmetries have to be kept in mind. 
In (\ref{color_sym}), we just averaged over color permutations. This gives a priori $6$ terms for each bubble type, but some of them are identical. It turns out that for each type of interaction, we have exactly $3$ distinct bubbles. 
Similarly, $S_\vphi$ is a sum of three term, which we can consider as new bubbles. 
With the mass term, we therefore have a total number of $13$ different bubbles in the theory. From now on, $\cB$ has to be understood in this extended sense. We could as well work with independent couplings for each bubble $b \in \cB$, but we decide to consider the symmetric model only, which seems to us the most relevant situation. However, it is convenient to work with notations adapted to the more general situations, because this allows to write most of the equations in a more condensed fashion. In the following, we will work with coupling constants $t_b^\Lambda$ for any $b \in \cB$, which has to be understood as $t_{4}^{\Lambda}$, $t_{6,1}^{\Lambda}$, $t_{6,2}^{\Lambda}$, $CT_m^{\Lambda}$ or $CT_\vphi^{\Lambda}$ depending on the nature of $b$.

\

In (\ref{drawing_sym}), we divided each coupling constant by a certain number of permutations of labels on the external legs of a bubble associated to this coupling. More precisely, it is the order of the subgroup of the permutations of these labels leaving the labeled colored graph invariant. Note that a first look at $(6,2)$ interactions suggests an order $2$ symmetry, but it is incompatible with any coloring. The role of such rescalings of the coupling constants is, as usual, to make the symmetry factors appearing in the perturbative expansions more transparent. The symmetry factor $s(\cG)$ associated to a Feynman graph $\cG$ becomes the number of its automorphisms. All these conventions will be useful when discussing in detail how divergences can be absorbed into new effective coupling constants.

\

Finally, the reader might wonder whether it is appropriate to include the $2$-point function counter-terms in the interaction part of the action, rather than associating flowing parameters to the covariance itself. This question is particularly pressing for wave-function counter-terms, since they break the tensorial invariance of the interaction action. One might worry that the degenerate nature of the covariance could prevent a Laplacian interaction with no projector from being reabsorbed in a modification of the wave-function parameter of the covariance. However, it is not difficult to understand that the situation is identical to that of a non-degenerate covariance. At fixed cut-off, modifying the covariance is not exactly the same as adding $2$-point function counter-terms in the action, but the two prescriptions coincide in the $\Lambda \to 0$ limit. Thus, it is perfectly safe to work in the second setting. Moreover, this has the main advantage of being compatible with a fixed slicing of the covariance according to scales, which is the central technical tool of the work presented in this article. 

\subsection{List of divergent subgraphs}

From the previous sections, and as we will confirm later on, the Abelian divergence degree of a subgraph $\cH$ will allow us to classify the divergences. When $\cH$ does not contain any wave-function counter-terms, one has\footnote{One also assumes $F(\cH) \geq 1$, as in the previous section.}:
\beq\label{def_omega1}
\omega (\cH) = 3 - \frac{N}{2} - 2 n_2 - n_4 + 3 \rho(\cH / \cT) \,.
\eeq
We will moreover see in the next section that wave-function counter-terms are neutral with respect to power-counting arguments. We can therefore extend the definition (\ref{def_omega1}) of $\omega$ to arbitrary subgraphs if $n_2$ is understood as the number of $2$-valent bubbles 
\textit{which are not of the wave-function counter-term type}, and the contraction of a tree is also understood in a general sense: $\cH / \cT$ is the subgraph obtained by \textit{first collapsing all chains of wave-function counter-terms}, and then contracting a tree $\cT$ in the collapsed graph. Alternatively, $\omega$ takes the generalized form:
\beq
\omega(\cH) = -2 ( L - W ) + 3 ( F - R ) \,,
\eeq
where $W$ is the number of wave-function counter-terms in $\cH$. \textit{This formula holds also when $F(\cH) = 0$}.

\
   
Let us focus on non-vacuum connected subgraphs with $F \geq 1$, which are the physically relevant ones. In this case $\rho = 0$ for melonic subgraphs and $\rho \leq -1$ otherwise. Therefore 
$$\omega (\cH) \leq - \frac{N}{2}$$
if $\cH$ is not melonic. As a result, divergences are entirely due to melonic subgraphs. They are in particular tracial, which means their Abelian power-counting is optimal. We therefore obtain an exact classification of divergent subgraphs, provided in table \ref{div}. It tells us that $6$-point functions have logarithmic divergences, $4$-point functions linear divergences as well as possible logarithmic ones, that will have to be absorbed in the constants $t_{4}^{\Lambda}$, $t_{6,1}^{\Lambda}$ and $t_{6,2}^{\Lambda}$. The full $2$-point function will be quadratically divergent, generating the constants $CT_{m}^{\Lambda}$ and $CT_{\vphi}^{\Lambda}$.

\begin{table}[h]
\centering
\begin{tabular}{| l | c | c | c || r |}
    \hline
    $N$ & $n_2$ & $n_4$ & $\rho$ & $\omega$  \\ \hline\hline
 6 & 0 & 0 & 0 & 0 \\ \hline
 4 & 0 & 0 & 0 & 1 \\ 
 4 & 0 & 1 & 0 & 0 \\ \hline
 2 & 0 & 0 & 0 & 2 \\ 
 2 & 0 & 1 & 0 & 1 \\
 2 & 0 & 2 & 0 & 0 \\
 2 & 1 & 0 & 0 & 0 \\ 
    \hline
  \end{tabular}
\caption{Classification of non-vacuum divergent graphs for $d=D=3$. All of them are melonic.}
\label{div}
\end{table}

{\bf Remark.}
There are a lot more cases to consider for vacuum divergences, including non-melonic contributions. However, they are irrelevant 
to perturbative renormalization.

\

In light of Corollary \ref{faces_legs}, we also notice that $2$-point divergent subgraphs, hence all degree $2$ subgraphs, have a single external face. This is a useful point to keep in mind as far as wave-function renormalization is concerned. As for $4$- and $6$-point divergent subgraphs, they have at most $2$ and $3$ external faces respectively. It is also not difficult to find examples saturating these two bounds, as shown in Figures \ref{div_4} and \ref{div_6}. 

\begin{figure}[h]
  \centering
  \subfloat[$\omega = 1$]{\label{div_4}\includegraphics[scale=0.6]{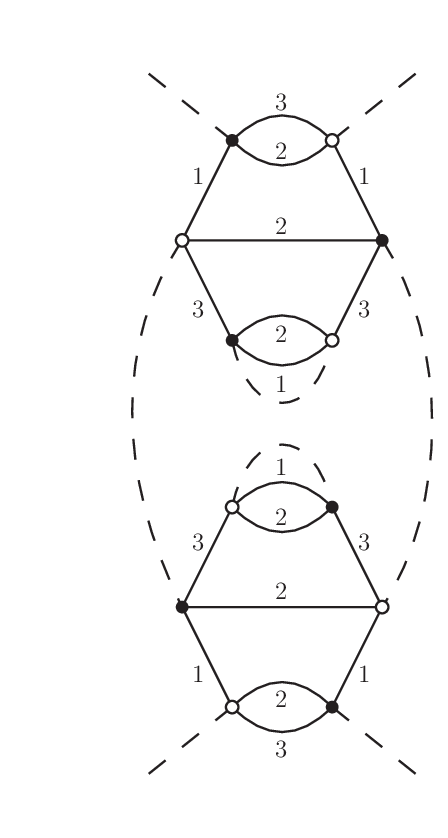}}   	       
  \subfloat[$\omega = 0$]{\label{div_6}\includegraphics[scale=0.6]{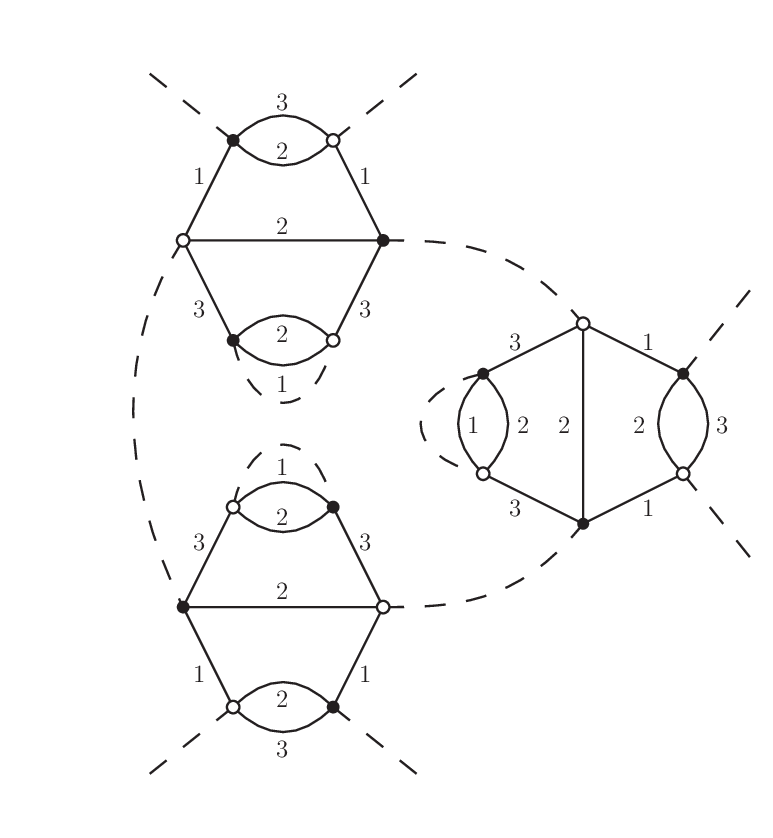}}
  \caption{Divergent subgraphs with respectively $2$ and $3$ external faces.}
\end{figure}

\section{Multi-scale expansion}\label{sec:multiscale}
\label{multiscale}

In this section, we use a multi-scale expansion \cite{Riv} to prove the claimed results concerning the applicability of the Abelian power-counting in the $\SU(2)$ case. We then show that divergent high subgraphs generate local counter-terms for the $2$-, $4$-, and $6$-point functions, supplemented by finite remainders.

\subsection{Multi-scale expansion}

The multi-scale expansion relies on a slicing of the propagator in the Schwinger parameter $\alpha$, according to a geometric progression. We fix an arbitrary constant $M > 1$ and for any integer $i \geq 0$, we define the slice of covariance $C_i$ as:
\bes
C_{0}(g_1, g_2 , g_3 ; g_1' , g_2' , g_3' ) &=& \int_{1}^{+ \infty} \extd \alpha \, \e^{- \alpha m^2} \int \extd h \prod_{\ell = 1}^{3} K_{\alpha} (g_\ell h g_\ell'^{\inv})\,, \\
\forall i \geq 1\,, \qquad C_{i}(g_1, g_2 , g_3 ; g_1' , g_2' , g_3' ) &=& \int_{M^{ - 2 i}}^{M^{ - 2(i-1)}} \extd \alpha \, \e^{- \alpha m^2} \int \extd h \prod_{\ell = 1}^{3} K_{\alpha} (g_\ell h
g_\ell'^{\inv})\,.
\ees  
In order to be compatible with the slicing, we choose a UV regulator of the form $\Lambda = M^{-2 \rho}$. In this context, we will use the simpler notation $C^{\rho}$ for $C^{M^{-2 \rho}}$ (see \eqref{paracut}):
\beq  \label{decomposi}
C^{\rho} = \sum_{0 \leq i \leq \rho} C_i \,.
\eeq  
We can then decompose the amplitudes themselves, according to scale attributions $\mu = \{ i_e \}$ where $i_e$ are integers associated to each line, determining the slice attribution of its propagator. The full amplitude $\cA_\cG$ of $\cG$ is then reconstructed from the sliced amplitudes $\cA_{\cG, \mu}$ by simply summing over the scale attribution $\mu$:
\beq
\cA_\cG = \sum_{\mu} \cA_{\cG , \mu}\,.
\eeq
The idea of the multi-scale analysis is then to bound sliced propagators, and deduce an optimized bound for each $\cA_{\cG , \mu}$ separately. To this effect, we first need to capture the peakedness properties of the propagators into Gaussian bounds.
They can be deduced from a general fact about heat kernels on curved manifolds: at small times, they look just the same as their flat counterparts, and can therefore be bounded by suitable Gaussian functions. In the case of $\SU(2)$, let us denote $\vert X \vert$ the norm of a Lie algebra element $X \in \su(2)$, and $|g|$ the geodesic distance between a Lie group element $g \in \SU(2)$ and the identity $\one$. We can prove the following bounds on $K_\alpha$ and its derivatives.
\begin{lemma}\label{heat}
There exists a set of constants $\delta > 0$ and $K_n > 0$ , such that for any $n \in \mathbb{N}$ the following holds:
\beq
\forall \alpha \in \left] 0 , 1 \right] , \quad \forall g \in \SU(2), \quad \forall X \in \su(2),\, |X| = 1, \qquad
\vert ( \cL_X )^n K_\alpha (g) \vert \leq K_n \alpha^{-\frac{3 + n}{2}} \e^{- \delta \frac{|g|^2}{\alpha}} 
\eeq
\end{lemma}
\begin{proof}
See the Appendix.
\end{proof}

As a consequence, the divergences associated to the propagators and their derivatives can be captured in the following bounds.
 \begin{proposition}
 There exist constants $K > 0$ and $\delta > 0$, such that for all $i \in \mathbb{N}$:
 \beq\label{propa_bound} 
 C_i (g_1, g_2 , g_3 ; g_1' , g_2' , g_3' ) \leq K M^{7 i } \int \extd h \,
 \e^{- \delta M^{i} \sum_{\ell =1}^{3} |g_\ell h g_\ell'^{\inv}| }\,\,.
 \eeq
 Moreover, for any integer $k \geq 1$, there exists a constant $K_k$, such that for any $i \in \mathbb{N}$, any choices of colors $\ell_p$ and Lie algebra elements $X_p \in \su(2)$ of unit norms ($1 \leq p \leq k$):
 \beq\label{deriv_bound}
\left( \prod_{p = 1}^{k} \cL_{X_{p} , g_{\ell_p}} \right) C_i (g_1, g_2 , g_3 ; g_1' , g_2' , g_3' ) \leq K M^{ (7 + k) i } \int \extd h \,
 \e^{- \delta M^{i} \sum_{\ell =1}^{3} |g_\ell h g_\ell'^{\inv}| }\,,
\eeq
where $\cL_{X_p , g_{\ell_p}}$ is the Lie derivative with respect to the variable $g_{\ell_p}$ in direction $X_p$.\footnote{We define the Lie derivative of a function $f$ as:
\beq
\cL_X f(g) \equiv \frac{\extd}{\extd t} f(g \e^{t X}) |_{t = 0} \,.
\eeq}
\end{proposition}
\begin{proof}
For $i \geq 1$, the previous lemma immediately shows that:
\bes
C_i (g_1, g_2 , g_3 ; g_1' , g_2' , g_3' ) &\leq& K_1 \int_{M^{- 2 i}}^{M^{- 2 (i- 1)}} \int \extd h \, \frac{\e^{- \frac{\delta_1}{\alpha} \sum_{\ell =1}^{3} |g_\ell h g_\ell'^{\inv}|^2 }}{\alpha^{9/2}} \\
&\leq& K_1 M^{- 2 (i- 1)} (M^{2 i})^{9/2} \int \extd h \, \e^{- \delta_1 M^{- 2 i} \sum_{\ell =1}^{3} |g_\ell h g_\ell'^{\inv}|^2 } \\
&\leq & K M^{7 i } \int \extd h \, \e^{- \delta M^{i} \sum_{\ell =1}^{3} |g_\ell h g_\ell'^{\inv}| }\,,
\ees
for some strictly positive constants $K_1$, $\delta_1$, $K$ and $\delta$. And similarly for Lie derivatives of $C_i$.

\
When $i = 0$, equations (\ref{nice_form}), (\ref{bound_F}) and (\ref{ext_bound}), together with the fact that $m \neq 0$ allow to bound the integrand of $C_0$ by an integrable function of $\alpha \in [ 1 , + \infty [$. $C_0$ is therefore bounded from above by a constant, and due to the compact nature of $\SU(2)$ we can immediately deduce a bound of the form
\beq   
C_0 (g_1, g_2 , g_3 ; g_1' , g_2' , g_3' ) \leq K \int \extd h \,
 \e^{- \delta \sum_{\ell =1}^{3} |g_\ell h g_\ell'^{\inv}| }\,.
\eeq
Again, the same idea allows to prove a similar bound on the Lie derivatives of $C_0$, which concludes the proof.
\end{proof}

Before stating the multi-scale power-counting theorem, we need an additional technical tool: the Gallavotti-Nicol\`{o} tree. It is the abstract tree encoding the inclusion order of \textit{high subgraphs} of a connected graph  $\cG$.
\begin{definition}
Let $\cG$ be a connected graph, with scale attribution $\mu$.
\begin{enumerate}[(i)]
\item Given a subgraph $\cH \in \cG$, one defines internal and external scales:
\beq
i_{\cH}(\mu) = \inf_{e \in L(\cH)} i_e (\mu)\,, \qquad e_{\cH}(\mu) = \sup_{e \in N_{ext}(\cH)} i_e (\mu)\,,
\eeq
where $N_{ext}(\cH)$ are the external legs of $\cH$ which are hooked to external faces.
\item A \textit{high subgraph} of $(\cG , \mu)$ is a connected subgraph $\cH \subset \cG$ such that
$e_{\cH}(\mu) < i_{\cH}(\mu)$. We label them as follows. For any $i$, $\cG_i$ is defined as the set of lines of $\cG$ with scales higher or equal to $i$. We call $k(i)$ its number of face-connected components, and $\{ \cG_{i}^{(k)} | 1 \leq k \leq k(i) \}$ its face-connected components. The subgraphs $\cG_{i}^{(k)}$ are exactly the high subgraphs.
\item Two high subgraphs are either included into another or line-disjoint, therefore the inclusion relations of the subgraphs $\cG_{i}^{(k)}$ can be represented as an abstract graph, whose root is the whole graph $\cG$. This is the \textit{Gallavotti-Nicol\`o tree} or simply \textit{GN tree}.
\end{enumerate}
\end{definition}

We can now extend the multi-scale power-counting of \cite{COR} to our non-Abelian model.

\begin{proposition}
There exists a constant $K > 0$, such that for any connected graph $\cG$ with scale attribution $\mu$, the following bound holds: 
\beq\label{fund}
 \vert \cA_{\cG, \mu}  \vert   \leq K^{L(\cG)}  \prod_{i \in \mathbb{N}} \prod_{ k \in \llbracket 1 , k(i) \rrbracket } M^{\omega [  \cG_{i}^{(k)}]}\,,
\eeq
where $\omega$ is the Abelian degree of divergence
\beq
\omega(\cH) = - 2 ( L(\cH) - W(\cH) ) + 3 ( F(\cH) - R(\cH) ) \,.
\eeq
\end{proposition}
\begin{proof}
Let us first assume $W(\cG) = 0$. In this case, we follow and adapt the proof of Abelian power-counting of \cite{COR}, about which we refer for additional details. 
We first integrate the $g$ variables in an optimal way, as was done in \cite{bgriv}. In each face $f$, a maximal tree of lines $T_f$ is
chosen to perform $g$ integrations. Optimality is ensured by requiring the trees $T_f$ to be compatible with the abstract GN tree, in the sense that $T_f \cap G_{i}^{(k)}$ has to be a tree itself, for any $f$ and $G_{i}^{(k)}$. This yields:
\bes\label{step1}
\vert \cA_{\cG, \mu} \vert &\leq& K^{L(\cG)}  \prod_{i \in \mathbb{N}} \prod_{ k \in \llbracket 1 , k(i) \rrbracket } M^{- 2 L( \cG_{i}^{(k)} ) + 3 F( \cG_{i}^{(k)} )} \\
&& \times \int [\extd h]^{L(\cG)} \prod_{f} \e^{- \delta M^{i(f)} \vert \overrightarrow{\prod_e} h_e^{\epsilon_{ef}} \vert}\,,
\ees
where $i(f) = \min \{ i_e \vert e \in f \}$.

\
The main difference with \cite{COR} is that variables are non-commuting, which prevents us from easily integrating out these variables. We can however rely on the methods developed in \cite{valentinmatteo}, which provide an exact power-counting theorem for $BF$ spin foam models. In particular, one can show that for any $2$-complex with $E$ edges and $F$ faces, the expression
\beq
\int [d g_e ]^E \exp \left( - \Lambda \sum_f | \prod_{e \in f} g_e^{\epsilon_{ef}}| \right)
\eeq
scales as $\Lambda^{- \rm{rk} \, \delta^{1}_{\phi}}$ when $\Lambda \to 0$. $\delta^{1}_{\phi}$ is the twisted boundary map associated to a (non-singular) flat connection $\phi$\footnote{The explicit construction of this map can be found in \cite{valentinmatteo}. With the notations of the present paper, it is defined as 
\bes
\delta^{1}_\phi : \quad E \otimes su(2) &\rightarrow& F \otimes su(2) \\
												 e \otimes X &\mapsto & \sum_f \epsilon_{ef} f \otimes Ad_{P \phi (v_e, v_f)} (X)
\ees
where $P \phi (v_e, v_f)$ is a path from a reference vertex $v_e$ in the edge $e$ to a reference vertex $v_f$ in the face $f$. The adjoint action encodes parallel transport with respect to $\phi$, and is full rank.}, which takes the non-commutativity of the group into account. Remarkably, this boundary map verifies:
\beq
\rm{rk} \, \delta^1_\phi \geq 3 \rm{rk} \, \epsilon_{ef}\,.
\eeq
As a result, the contribution of the closed face of a $\cG_i^{(k)}$ can be bounded by
\beq
\int [\extd g_e ]^{L(\cG_i^{(k)})} \exp \left( - M^{i(f)} \underset{f \in F(\cG_i^{(k)})}{\sum} | \prod_{e \in f} g_e^{\epsilon_{ef}}| \right) \leq K_1^{L(\cG_i^{(k)})} M^{- 3 R(\cG_i^{(k)}) i}\,.
\eeq
The power-counting (\ref{fund}) is recovered by recursively applying this bound, from the leaves to the root of the GN tree. 

\

The $W(\cG) \neq 0$ case is an immediate consequence of the $W(\cG) = 0$ one. Indeed, one just needs to understand how the insertion of a wave-function counter-term in a graph $\cG$ affects its amplitude $\cA_\cG$. While it adds one line to $\cG$, it does not change its number of faces, nor their connectivity structure, hence the rank $R$ is not modified either. The line being created is responsible for an additional $M^{-2 i}$ factor in the power-counting, with $i$ its scale. On the other hand, it is acted upon by a Laplace operator, that is two derivatives, which according to (\ref{deriv_bound}) generate an $M^{2 i}$. The two contributions cancel out, which shows that wave-function counter-terms are neutral to power-counting. The $L$ contribution to $\omega$ has therefore to be compensated by a $W$ term with the opposite sign.
\end{proof}

Notice that all the steps in the derivation of the bound are optimal, in the sense that we could find lower bounds with the same structure, except for the last integrations of face contributions. In this last step we discarded the fine effects of the noncommutative nature of $\SU(2)$, encoded in the rank of $\delta^{1}_{\phi}$. Remark however that no such effect is present for a contractible $\cG_{i}^{(k)}$, since the $2$-complex formed by its internal faces is simply connected \cite{valentinmatteo}. Indeed, such a subgraph supports a unique flat connection (the trivial one), which means that the integrand in equation (\ref{step1}) can be linearized around $h_e = \one$, showing the equivalence between Abelian and non-Abelian power-countings in this case. 
Since melonic subgraphs are contractible, this confirms our previous claim: the Abelian power-counting exactly captures the divergences of the set of models studied in this paper.

\subsection{Contraction of high melonic subgraphs}\label{sec:contraction}

We close this section with a discussion of the key ingredients entering the renormalization of this model, by explaining how local approximations to high melonic subgraphs are extracted from high slices to lower slices of the amplitudes. A full account of the renormalization procedure, including rigorous finiteness results, will be detailed in the next and final section. 
 
\

\begin{figure}[h]
\begin{center}
\includegraphics[scale=0.6]{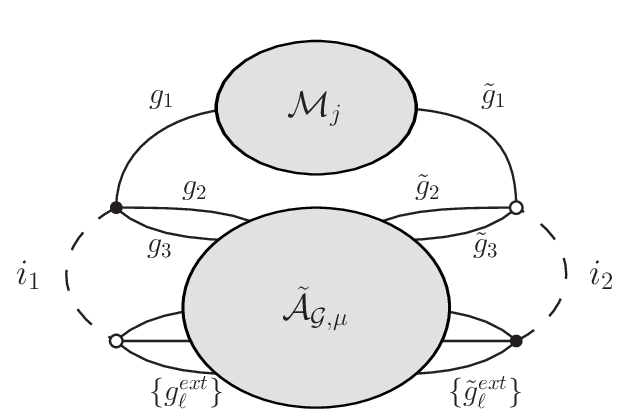}
\caption{A graph with a high melonic subgraph $\cM_j$: in this drawing, the grey bulb labeled $\cM_j$ represents the smallest part of the colored extension of $\cM_j$ which contains all its lines. The lines labeled by the scales $i_1$ and $i_2$ are the external legs of $\cM_j$.}
\label{contract_melon}
\end{center}
\end{figure}
Let us consider a non-vacuum graph with scale attribution $( \cG , \mu )$, containing a melonic high subgraph $\cM_j \subset \cG$ at scale $j$. For the convenience of the reader, we first focus on the case $F_{ext}(\cM_j) = 1$, which encompasses all the $2$-point divergent subgraphs, therefore all the degree $2$ subgraphs. We also first assume that no wave-function counter-terms is present in $\cM_j$.

\subsubsection{Divergent subgraphs with a single external face, and no wave-function counter-terms}

Since $F_{ext} (\cM_j) = 1$, $N_{ext} ( \cM_j )$ contains two external propagators, labeled by external variables $\{ g_\ell^{ext} \,, \ell = 1 , \ldots , 3 \}$ and $\{ \tilde{g}_\ell^{ext} \,, \ell = 1 , \ldots , 3 \}$, and scales $i_1 < j$ and $i_2 < j$ respectively. We can assume (without loss of generality) that the melonic subgraph $\cM_j$ is inserted on a color line of color $\ell = 1$. The amplitude of $\cG$, pictured in Figure \ref{contract_melon}, takes the form:
\bes
\cA_{\cG , \mu} &=& \int [\extd g_\ell ]^3 [\extd \tilde{g}_\ell ]^3 [\extd g_\ell^{ext} ]^3 [\extd \tilde{g}_\ell^{ext} ]^3 \widetilde{\cA}_{\cG , \mu}  ( g_2 , g_3 ; \tilde{g}_2 , \tilde{g}_3 ; \{ g_\ell^{ext} \} ; \{ \tilde{g}_\ell^{ext} \} ) \nn \\
&& \times C_{i_1} ( g_1^{ext} , g_2^{ext} , g_3^{ext} ; g_1 , g_2 , g_3 ) \cM_j ( g_1 , \tilde{g}_1 ) C_{i_2} ( \tilde{g}_1 , \tilde{g}_2 , \tilde{g}_3 , g_1^{ext} , g_2^{ext} , g_3^{ext} ) \, .
\ees
The idea is then to approximate the value of $\cA_{\cG , \mu}$ by an amplitude associated to the contracted graph $\cG / \cM_j$. This can be realized by "moving" one of the two external propagators towards the other. In practice, we can use the interpolation\footnote{$X_g$ denotes the Lie algebra element with the smallest norm such that $\e^{X_g} = g$.}
\bes
g_1 (t) = \tilde{g}_1 e^{t X_{\tilde{g}_1^{\inv} g_1}} \,, \qquad t \in [0 , 1]\, 
\ees 
and define:
\bes\label{sep_scales}
\cA_{\cG , \mu} (t) &=& \int [\extd g_\ell ]^3 [\extd \tilde{g}_\ell ]^3 [\extd g_\ell^{ext} ]^3 [\extd \tilde{g}_\ell^{ext} ]^3 \widetilde{\cA}_{\cG , \mu}  ( g_2, g_3 ; \tilde{g}_2 , \tilde{g}_3 ; \{ g_\ell^{ext} \} ; \{ \tilde{g}_\ell^{ext} \} ) \nn \\
&& \times C_{i_1} ( g_1^{ext} , g_2^{ext} , g_3^{ext} ; g_1 (t) , g_2 , g_3 ) \cM_j ( g_1 , \tilde{g}_1 ) C_{i_2} ( \tilde{g}_1 , \tilde{g}_2 , \tilde{g}_3 , \tilde{g}_1^{ext} , \tilde{g}_2^{ext} , \tilde{g}_3^{ext} ) \,.
\ees
This formula together with a Taylor expansion allows to approximate $\cA_\cG = \cA_\cG (1)$ by $\cA_\cG (0)$ and its derivatives. The order at which the approximation should be pushed is determined by the degree of divergence $\omega(\cM_j )$ of $\cM_j$ and the power-counting theorem: we should use the lowest order ensuring that the remainder in the Taylor expansion has a convergent power-counting. Roughly speaking each derivative in $t$ decreases the degree of divergence by $1$, therefore the Taylor expansion needs to be performed up to order $\omega(\cM_j )$:
\beq\label{taylor}
\cA_{\cG , \mu} = \cA_{\cG , \mu} (1) = \cA_{\cG , \mu} (0) + \sum_{k = 1}^{\omega(\cM_j )} \frac{1}{k !} \cA_{\cG , \mu}^{(k)} (0) + \int_{0}^{1} \extd t \, \frac{(1 - t)^{\omega(\cM_j )}}{\omega(\cM_j ) !} \cA_{\cG , \mu}^{(\omega(\cM_j ) + 1)} (t) .
\eeq 

\

Before analyzing further the form of each of these terms, we point out a few interesting properties verified by the function $\cM_j$. First, since by definition the variables $g_1$ and $\tilde{g}_1$ are boundary variables for a same face (and because the heat-kernel is a central function), $\cM_j( g_1 , \tilde{g}_1)$ can only depend on $\tilde{g}_1^{\inv} g_1$. From now on, we therefore use the notation:
\beq
\cM_j( g_1 , \tilde{g}_1) = \cM_j ( \tilde{g}_1^{\inv} g_1 )\,.
\eeq
We can then prove the following lemma.
\begin{lemma}
\begin{enumerate}[(i)]
\item $\cM_j$ is invariant under inversion:
\beq
\forall g \in \SU(2) \, , \qquad \cM_j (g^{\inv}) = \cM_j (g) \,. 
\eeq 
\item $\cM_j$ is central:
\beq
\forall g \, , h \in \SU(2) \, , \qquad \cM_j ( h g h^{\inv}) = \cM_j (g) \,.
\eeq
\end{enumerate}
\end{lemma}
\begin{proof}
We can proceed by induction on the number of lines $\widetilde{L}$ of a rosette of $\cM_j$. 

\
When $\widetilde{L} = 1$, $\cM_j$ can be cast as an integral over a single Schwinger parameter $\alpha$ of an integrand of the form:
\beq
\int \extd h \,  K_\alpha ( \tilde{g}_1^{\inv}  g_1 h) \, (K_\alpha ( h ))^2 \,.
\eeq
By invariance of the heat kernels and the Haar measure under inversion and conjugation, the invariance of $\cM_j$ immediately follows. 

\
Suppose now that $\widetilde{L} \geq 2$. A rosette of $\cM_j$ can be thought of as an elementary melon decorated with two melonic insertions of size strictly smaller than $\widetilde{L}$ (at least one of them being non-empty). We therefore have:
\beq
\cM_j ( \tilde{g}_1^{\inv} g_1 ) = \int [ \extd g_2 \extd \tilde{g}_2 \extd g_3 \extd \tilde{g}_3 ] \, m^{(2)} ( g_2^{\inv} g_2 ) \, m^{(3)} ( g_3^{\inv} g_3 ) \int \extd \alpha \int \extd h \, K_\alpha ( \tilde{g}_1^{\inv}  g_1 h) K_\alpha ( \tilde{g}_2^{\inv}  g_2 h ) K_\alpha ( \tilde{g}_3^{\inv}  g_3 h) \,,
\eeq
in which we did not specify the integration domain of $\alpha$, since it does not play any role here. $m^{(2)}$ and $m^{(3)}$ are associated to melonic subgraphs of size strictly smaller than $L$, we can therefore assume that they are invariant under conjugation and inversion\footnote{If one $m^{(i)}$ is an empty melon, then $m^{(i)} (g) = \delta(g)$, and is trivially invariant.}. Using again the invariance of the heat kernels and the Haar measure, we immediately conclude that $\cM_j$ itself is invariant. 
   
\end{proof}

\

We now come back to (\ref{taylor}). The degree of divergence being bounded by $2$, it contains terms $\cA_{\cG , \mu}^{(k)} (0)$ with $k \leq 2$. We now show that $\cA_{\cG , \mu} (0)$ gives mass counter-terms, $\cA_{\cG , \mu}^{(1)} (0)$ is identically zero, and $\cA_{\cG , \mu}^{(2)} (0)$ implies wave-function counter-terms. This is stated in the following proposition.

\begin{proposition}
\begin{enumerate}[(i)]
\item $\cA_{\cG , \mu} (0)$ is proportional to the amplitude of the contracted graph $\cG / \cM_j$, with the same scale attribution:
\beq
\cA_{\cG , \mu} (0) = \left( \int \extd g \cM_j (g) \right) \cA_{\cG / \cM_j , \mu} \,.
\eeq
\item Due to the symmetries of $\cM_j$, $\cA_{\cG , \mu}^{(1)} (0)$ vanishes: 
\beq
\cA_{\cG , \mu}^{(1)} (0) = 0\,.
\eeq
\item $\cA_{\cG , \mu}^{(2)}$ is proportional to an amplitude in which a Laplace operator has been inserted in place of $\cM_j$:
\bes\label{wf_ct}
\cA_{\cG , \mu}^{(2)} (0) &=& \left( \frac{1}{3} \int \extd g \cM_j (g) | X_g | ^{2} \right) \nn \\ 
&& \times \int [\extd g_\ell ]^3 [\extd \tilde{g}_\ell ]^3 \int [\extd g_\ell^{ext} ]^3 [\extd \tilde{g}_\ell^{ext} ]^3 \widetilde{\cA}_{\cG , \mu}  ( g_2 , g_3 ; \tilde{g}_2 , \tilde{g}_3 ; \{ g_\ell^{ext} \} ; \{ \tilde{g}_\ell^{ext} \} ) \nn \\
&& \times \left( \Delta_{\tilde{g}_1} C_{i_1} ( g_1^{ext} , g_2^{ext} , g_3^{ext} ; \tilde{g}_1 , g_2 , g_3 ) \right) C_{i_2} ( \tilde{g}_1 , \tilde{g}_2 , \tilde{g}_3 , \tilde{g}_1^{ext} , \tilde{g}_2^{ext} , \tilde{g}_3^{ext} )  \,.
\ees
\end{enumerate}
\end{proposition} 
\begin{proof}

\begin{enumerate}[(i)]
\item One immediately has:
\bes
\cA_{\cG , \mu} (0) &=& \int [\extd g_\ell ]^3 [\extd \tilde{g}_\ell ]^3 [\extd g_\ell^{ext} ]^3 [\extd \tilde{g}_\ell^{ext} ]^3 \widetilde{\cA}_{\cG , \mu}  ( g_2 , g_3 ; \tilde{g}_2 , \tilde{g}_3 ; \{ g_\ell^{ext} \} ; \{ \tilde{g}_\ell^{ext} \} )\nn \\
&& \times C_{i_1} ( g_1^{ext} , g_2^{ext} , g_3^{ext} ; \tilde{g}_1 , g_2 , g_3 ) \cM_j ( \tilde{g}_1^{\inv} g_1  ) C_{i_2} ( \tilde{g}_1 , \tilde{g}_2 , \tilde{g}_3 , \tilde{g}_1^{ext} , \tilde{g}_2^{ext} , \tilde{g}_3^{ext} ) \nn \\
&=& \int [\extd g_\ell ]^3 [\extd \tilde{g}_\ell ]^3  [\extd g_\ell^{ext} ]^3 [\extd \tilde{g}_\ell^{ext} ]^3 \widetilde{\cA}_{\cG , \mu}  ( g_2 , g_3 ; \tilde{g}_2 , \tilde{g}_3 ; \{ g_\ell^{ext} \} ; \{ \tilde{g}_\ell^{ext} \} ) \nn \\
&& \times C_{i_1} ( g_1^{ext} , g_2^{ext} , g_3^{ext} ; \tilde{g}_1 , g_2 , g_3 ) \cM_j ( g_1 ) C_{i_2} ( \tilde{g}_1 , \tilde{g}_2 , \tilde{g}_3 , \tilde{g}_1^{ext} , \tilde{g}_2^{ext} , \tilde{g}_3^{ext} ) \nn \\
&=& \left( \int \extd g \cM_j (g) \right) \cA_{\cG / \cM_j , \mu} \,,
\ees
where from the first to the second line we made the change of variable $g_1 \to \tilde{g}_1 g_1$.
\item For $\cA_{\cG , \mu}^{(1)} (0)$, a similar change of variables yields:
\bes
\cA_{\cG , \mu}^{(1)} (0) &=& \int [\extd g_\ell ]^2 [\extd \tilde{g}_\ell ]^3 [\extd g_\ell^{ext} ]^3 [\extd \tilde{g}_\ell^{ext} ]^3 \widetilde{\cA}_{\cG , \mu}  ( g_2 , g_3 ; \tilde{g}_2 , \tilde{g}_3 ; \{ g_\ell^{ext} \}; \{ \tilde{g}_\ell^{ext} \} ) \nn \\
&& \times \left( \int \extd g  \, \cM_j (g) \, \cL_{X_g , \tilde{g}_1} \, C_{i_1} ( g_1^{ext} , g_2^{ext} , g_3^{ext} ; \tilde{g}_1 , g_2 , g_3 )  \right) C_{i_2} ( \tilde{g}_1 , \tilde{g}_2 , \tilde{g}_3 , \tilde{g}_1^{ext} , \tilde{g}_2^{ext} , \tilde{g}_3^{ext} ) \,. \nn 
\ees
But by invariance of $\cM_j$ under inversion, one also has
\beq
\int \extd g  \, \cM_j (g) \, \cL_{X_g} = - \int \extd g  \, \cM_j (g) \, \cL_{X_g} \; \Rightarrow \; \int \extd g  \, \cM_j (g) \, \cL_{X_g} = 0\,,
\eeq 
hence $\cA_{\cG , \mu}^{(1)} (0) = 0$. 

\item Finally, $\cA_{\cG , \mu}^{(2)} (0)$ can be expressed as:
\bes
\cA_{\cG , \mu}^{(2)} (0) &=& \int [\extd g_\ell ]^2 [\extd \tilde{g}_\ell ]^3 [\extd g_\ell^{ext} ]^3 [\extd \tilde{g}_\ell^{ext} ]^3 \widetilde{\cA}_{\cG , \mu}  ( g_2 , g_3 ; \tilde{g}_2 , \tilde{g}_3 ; \{ g_\ell^{ext} \} ; \{ \tilde{g}_\ell^{ext} \} ) \nn \\
&& \times \left( \int \extd g  \, \cM_j (g) \, (\cL_{X_g , \tilde{g}_1})^{2} \, C_{i_1} ( g_1^{ext} , g_2^{ext} , g_3^{ext} ; \tilde{g}_1 , g_2 , g_3 )  \right) C_{i_2} ( \tilde{g}_1 , \tilde{g}_2 , \tilde{g}_3 , \tilde{g}_1^{ext} , \tilde{g}_2^{ext} , \tilde{g}_3^{ext} ) \,. \nn 
\ees
We can decompose the operator $\int \extd g  \, \cM_j (g) \, (\cL_{X_g})^{2}$ into its diagonal and off-diagonal parts with respect to an orthonormal basis $\{ \tau_k \,, k = 1\,, \ldots , 3\}$ in $\su(2)$. The off-diagonal part writes
\beq
\sum_{k \neq l} \int \extd g \, \cM_j (g ) \, X^{k}_g X^{l}_g \, \cL_{\tau_k} \cL_{\tau_l}
\eeq
and can be shown to vanish. Indeed, let us fix $k \neq l$, and $h \in \SU(2)$ such that:
\beq
X^{k}_{h g h^{\inv}} = X^{l}_{g} \; \; ; \qquad X^{l}_{h g h^{\inv}} = - X^{k}_{g} \,.  
\eeq
It follows from the invariance of $\cM_j$ under conjugation that 
\beq
\int \extd g \, \cM_j (g ) \, X^{k}_g X^{l}_g = - \int \extd g \, \cM_j (g ) \, X^{k}_g X^{l}_g
\; \Rightarrow \; \int \extd g \, \cM_j (g ) \, X^{k}_g X^{l}_g = 0\,. 
\eeq
Hence all off-diagonal terms vanish. One is therefore left with the diagonal ones, which contribute in the following way:
\beq
\int \extd g  \, \cM_j (g) \, (\cL_{X_g})^{2} = \sum_{k=1}^{3} \int \extd g  \, \cM_j (g) \, ( X^{k}_g )^{2} (\cL_{\tau_k})^{2} \,.
\eeq
Again, by invariance under conjugation, $\int \extd g  \, \cM_j (g) \, ( X^{k}_g )^{2}$ does not depend on $k$. This implies:
\bes
\int \extd g  \, \cM_j (g) \, (\cL_{X_g})^{2} &=& \int \extd g  \, \cM_j (g) \, ( X^{1}_g )^{2} \sum_{k=1}^{3} (\cL_{\tau_k})^{2} \\
&=& \left( \frac{1}{3} \int \extd g  \, \cM_j (g) \, ( X^{k}_g )^{2} \right) \Delta \,. 
\ees
\end{enumerate}
\end{proof}

\subsubsection{Additional external faces and wave-function counter-terms}

Let us first say a word about how the previous results generalize to more external faces, still assuming the absence of wave-function counter-terms. According to Corollary \ref{faces_legs}, the only two possibilities are $F_{ext}(\cM_j) = 2$ or $F_{ext}(\cM_j) = 3$, and in both cases $N \geq 4$. Incidentally, $\omega( \cM_j ) = 0$ or $1$. Moreover, since all the faces have the same color, we always have $N_{ext} (\cM_j) = 2 F_{ext}(\cM_j)$. One defines $\cA_{\cG , \mu} (t)$ by interpolating between the end variables of the external faces, which consist of $F_{ext}(\cM_j)$ pairs of variables, with one variable per propagator in $N_{ext} (\cM_j)$. Assuming their color to be $1$, for instance, the amplitude $\cA_{\cG , \mu} (t)$ can be written as:
\bes
\cA_{\cG , \mu} (t) &=& \int [\extd g_\ell^{k} \extd \tilde{g}_\ell^{k} ] [\extd g_\ell^{ext, k} \extd \tilde{g}_\ell^{ext, k} ]  \widetilde{\cA}_{\cG , \mu}  ( g_2^{k} , g_3^{k} ; \tilde{g}_2^{k} , \tilde{g}_2^{k} ; \{ g_\ell^{ext,k} \} ; \{ \tilde{g}_\ell^{ext,k} \} ) \\
&&\prod_{k = 1}^{F_{ext}(\cM_j)} C_{i_k} ( g_1^{ext, k} , g_2^{ext, k} , g_3^{ext , k} ; g_1^{k} (t) , g_2^{k} , g_3^{k} )  \cM_j ( \{ g_1^{k} , \tilde{g}_1^{k} \} ) C_{i_k '} ( \tilde{g}_1^{k} , \tilde{g}_2^{k} , \tilde{g}_3^{k} , \tilde{g}_1^{ext, k} , \tilde{g}_2^{ext, k} , \tilde{g}_3^{ext, k} ) \,, \nn
\ees
with 
\beq
g_1^{k} (t) = \tilde{g}_1^k e^{t X_{(\tilde{g}_1^{k})^\inv g_1^{k}}} \,, \qquad t \in [0 , 1]\,. 
\eeq
Moreover, we know that under a spanning tree contraction, the external faces of $\cM_j$ get disconnected. This means that the function $\cM_j$ can be factorized as a product
\beq
\cM_j ( \{ g_1^{k} , \tilde{g}_1^{k} \} ) = \prod_{k = 1}^{F_{ext} (\cM_j)} \cM_j^{(k)} ( g_1^{k} , \tilde{g}_1^{k} ) \,,
\eeq
such that each $\cM_j^{(k)}$ verifies all the invariances discussed in the previous paragraph. Thus, the part of the integrand of $\cA_{\cG , \mu} (t)$ relevant to $\cM_j$ is factorized into $k$ terms similar to the integrand appearing in the $F_{ext} = 1$ case. It is then immediate to conclude that all the properties which were proven in the previous paragraph hold in general. Indeed, the Taylor expansions to check are up to order $0$ or $1$ at most. The zeroth order of a product is trivially the product of the zeroth orders. As for the first order, it cancels out since the derivative of each one of the $k$ terms is $0$ at $t=0$. 

\

The effect of wave-function counter-terms is even easier to understand. Indeed, they essentially amount to insertions of Laplace operators. But the heat-kernel at time $\alpha$ verifies
\beq
\Delta K_\alpha = \frac{\extd K_\alpha}{\extd \alpha}\,,
\eeq
therefore all the invariances of $K_\alpha$ on which the previous demonstrations rely also apply to $\Delta K_\alpha$. 

\

All in all, the conclusions drawn in the previous paragraph hold for all non-vacuum high divergent subgraphs $\cM_j$.

\subsubsection{Notations and finiteness of the remainders}
\label{sec:remainders}

In the remainder of this paper, it will be convenient to use the following notations for the local part of the Taylor expansions above:
\beq
\tau_{\cM_j} \cA_{\cG , \mu} = \sum_{k = 0}^{\omega(\cM_j )} \frac{1}{k !} \cA_{\cG , \mu}^{(k)} (0) \,.
\eeq
$\tau_{\cM_j}$ projects the full amplitude $\cA_{\cG , \mu}$ onto effectively local contributions which take into account the relevant contributions of the subgraph $\cM_j \subset \cG$. To confirm that this is indeed the case, one needs to prove that in the remainder 
\beq
R_{\cM_j} \cA_{\cG , \mu} \equiv \int_{0}^{1} \extd t \, \frac{(1 - t)^{\omega(\cM_j )}}{\omega(\cM_j ) !} \cA_{\cG , \mu}^{(\omega(\cM_j ) + 1)} (t) \, ,
\eeq
the (non-local) part associated to $\cM_j$ is power-counting convergent. According to (\ref{sep_scales}), we have:
\bes
R_{\cM_j} \cA_{\cG , \mu} &=& \int_{0}^{1} \extd t \, \frac{(1 - t)^{\omega(\cM_j )}}{\omega(\cM_j ) !} \int [\extd g_\ell ]^3 [\extd \tilde{g}_\ell ]^3 [\extd g_\ell^{ext} ]^3 [\extd \tilde{g}_\ell^{ext} ]^3 \widetilde{\cA}_{\cG , \mu}  ( g_2 , g_3 ; \tilde{g}_2 , \tilde{g}_3 ; \{ g_\ell^{ext} \} ; \{ \tilde{g}_\ell^{ext} \} ) \nn \\
&&\times (\cL_{X_{\tilde{g_1}^{\inv} g_1 } , g_1 (t) })^{\omega(\cM_j) + 1} C_{i_1} ( g_1^{ext} , g_2^{ext} , g_3^{ext} ; g_1 (t) , g_2 , g_3 ) \nn \\
&& \times \cM_j ( g_1 , \tilde{g}_1 ) C_{i_2} ( \tilde{g}_1 , \tilde{g}_2 , \tilde{g}_3 , \tilde{g}_1^{ext} , \tilde{g}_2^{ext} , \tilde{g}_3^{ext} ) \,,
\ees
and therefore:
\bes\label{remainder_scales}
\vert R_{\cM_j} \cA_{\cG , \mu} \vert &\leq& \int_{0}^{1} \extd t \, \frac{\vert 1 - t \vert^{\omega(\cM_j )}}{\omega(\cM_j ) !} \int [\extd g_\ell ]^3 [\extd \tilde{g}_\ell ]^3 [\extd g_\ell^{ext} ]^3 [\extd \tilde{g}_\ell^{ext} ]^3 \widetilde{\cA}_{\cG , \mu}  ( g_2 , g_3 ; \tilde{g}_2 , \tilde{g}_3 ; \{ g_\ell^{ext} \} ; \{ \tilde{g}_\ell^{ext} \} ) \nn \\
&& \times \vert X_{\tilde{g_1}^{\inv} g_1} \vert^{\omega(\cM_j ) + 1} (\cL_{ \tilde{X}_{\tilde{g_1}^{\inv} g_1 } , g_1 (t) })^{\omega(\cM_j) + 1} C_{i_1} ( g_1^{ext} , g_2^{ext} , g_3^{ext} ; g_1 (t) , g_2 , g_3 ) \nn \\
&& \times \cM_j ( g_1 , \tilde{g}_1 ) C_{i_2} ( \tilde{g}_1 , \tilde{g}_2 , \tilde{g}_3 , \tilde{g}_1^{ext} , \tilde{g}_2^{ext} , \tilde{g}_3^{ext} ) \,,
\ees
where $\tilde{X}_{\tilde{g_1}^{\inv} g_1 }$ is the unit vector of direction $X_{\tilde{g_1}^{\inv} g_1 }$. We can now analyze how the power-counting of expression (\ref{remainder_scales}) differs from that of the amplitude $( \cG , \mu )$. There are two competing effects. The first is a loss of convergence due to the $\omega(\cM_j) + 1$ derivatives acting on $C_{i_1}$. According to (\ref{deriv_bound}), these contributions can be bounded by an additional $M^{( \omega( \cM_{j}) + 1 ) i_1 }$ term. This competes with the second effect, according to which the non-zero contributions of the integrand are concentrated in the region in which $\tilde{g_1}^{\inv} g_1$ is close to the identity. More precisely, the fact that $\cM_j$ contains only scales higher than $j$ imposes that 
\beq
\vert X_{\tilde{g_1}^{\inv} g_1} \vert \leq K M^{- j}
\eeq
where the integrand is relevant. The first line of \eqref{remainder_scales} therefore contributes to the power-counting with a term bounded by $M^{- (\omega( \cM_{j}) + 1 )) j}$. And since by definition $j > i_1$, one concludes that the degree of divergence of the remainder is bounded by:
\beq
\omega ( \cM_j ) + (\omega ( \cM_j ) + 1 ) ( i_1 - j ) \leq - 1 \,.  
\eeq

\section{Renormalization at all orders}\label{sec:finiteness}
\label{finiteness}

We conclude this paper by establishing a BPHZ theorem for the renormalized series. As in other kinds of field theories, this proof relies on forest formulas, and a careful separation between its high, divergent, and quasi-local parts from additional useless finite contributions.

\

We begin with a (standard) discussion about the compared merits of the renormalized expansion on the one hand, and the effective expansion on the other hand. 

So far we have discussed the renormalization of our model in the spirit of the latter, where each renormalization step (one for each slice) generates effective local couplings at lower scales. It perfectly fits Wilson's conception of renormalization:
in this setting, one starts with a theory with UV cut-off $\Lambda = M^{- 2 \rho}$, and tries to understand the physics in the IR, whose independence from UV physics is ensured by the separation of scales with respect to the cut-off. In order to compute physical processes involving external scales $i_{IR} < \rho$, one can integrate out all the fluctuations in the shell $i_{IR} < i \leq \rho$, resulting in an effective theory at scale $i_{IR}$. 

Because our model is renormalizable, we know that the main contributions in this integration are associated to quasi-local divergent subgraphs, therefore the effective theory can be approximated by a local theory of the same form as the bare one. According to Wilson's renormalization group treatment, in order to better handle the fact that only the high parts of divergent subgraphs contribute to this approximation, one should proceed in individual and iterated steps $i \to i - 1$ instead of integrating out the whole shell $i_{IR} < i \leq \rho$ at once. At each step one can absorb the ultimately divergent contributions (when the cut-off will be subsequently removed) into new effective coupling constants. This procedure, then, naturally generates one effective coupling constant per renormalizable interaction and per scale, as opposed to a single renormalized coupling per interaction in the renormalized expansion. This might look like a severe drawback, but on the other hand a main advantage is that the finiteness of the effective amplitudes becomes clear: the Taylor expansions of the previous section together with the finiteness of the remainder guarantee that all divergences are tamed. In particular, there is no problem of overlapping divergences since high subgraphs at a given scale cannot overlap.   

\

If one wants to be able to work with single renormalized couplings in the Lagrangian, one has to resort to a cruder picture in which the whole renormalization trajectory is approximated by a unique integration step from $\rho$ to $i_{IR}$. The price to pay is that one has no way anymore to isolate  the high (truly divergent) contributions of divergent subgraphs, which will result in additional finite contributions to the renormalized amplitudes. These contributions can build up over scales, explaining the appearance of renormalons, i.e. amplitudes which grow as a factorial of the number of vertices. This should be contrasted with the effective approach, in which amplitudes grow at most exponentially in the number of vertices. 

While they are not a big issue in perturbative expansions at low orders, renormalons are very problematic in non-perturbative approaches to quantum field theory such as the constructive program. 
They may be all the more problematic in TGFTs, if one expects continuum spacetime physics to show up in a regime dominated by large graphs, and thus to depend on non-perturbative effects\footnote{Here, we only mean non-perturbative in the sense of the perturbative expansion for small coupling constants, around the TGFT Fock vacuum, that we considered in this paper.}. And this seems in turn unavoidable if one interprets the perturbative expansion we deal with here as an expansion around the 'no spacetime' vacuum (corresponding to the TGFT Fock vacuum). 

A second related drawback of the renormalized expansion is the problem of overlapping divergences. Here again, the effective expansion appears to be very helpful. Not only overlapping divergences do not show up in this framework, but this also elucidates their treatment in the renormalized expansion. Indeed, at each step in the trajectory of the renormalization group, divergences are indexed by disconnected subgraphs. When one iterates the process, from high to lower scales, one finds that the divergent subgraphs of a given amplitude $\cA_{\cG , \mu}$ which contribute organize themselves into a forest. This is obvious once we understood that these graphs are high, and therefore correspond to nodes of the Gallavotti-Nicol\`o tree of $(\cG , \mu)$. 
In order to pack all these contributions into renormalized couplings for the whole trajectory of the renormalized group, it is therefore necessary to index the counter-terms by all the possible forests of divergent subgraphs (irrespectively of them being high or not), called Zimmermann's forests. Seen from this perspective, it is only when unpacking the renormalized amplitudes by appropriately decomposing them over scale attributions that one makes transparent why and how the Zimmermann's forest formula cures all divergences. Here again, the situation in TGFTs with respect to usual QFTs would suggest to resort to the effective expansion: due to the finer notion of connectedness which indexes the divergent subgraphs (face-connectedness), overlapping divergences are enhanced, the internal structure of the vertices being an additional source of difficulties. This is for instance manifest in super-renormalizable examples of the type \cite{COR}, in which overlapping contributions already enter the renormalization of tadpoles (leading to the notion of \lq melonic Wick ordering\rq). 

\

Despite the two generic drawbacks of the renormalized series, we choose a conservative approach in the following, and decide to outline in some details the proof of finiteness of the usual renormalized amplitudes. We will however start with a sketch of the recursive definition of the effective coupling constants. Since vertex-connectedness lies at the core of the Wilsonian effective expansion, we cannot take full advantage of face-connectedness in this context. This is similar to ordinary quantum field theories, where $1$-particle reducible graphs need to be taken into account in the effective expansion but can be dispensed with in the renormalized expansion. This is the main motivation for resorting to the renormalized expansion, where counter-terms are indexed by forests of divergent subgraphs. We will then decompose the amplitudes over scales and check that all contributions from high divergent subgraphs are correctly cured by the appropriate counter-terms. We will finally perform the sum over scale attributions, showing why the result is finite, and how useless counter-terms can build up to form renormalons.

\subsection{Effective and renormalized expansions}

As briefly explained before, the effective expansion is a reshuffling of the bare theory (with cut-off $\Lambda = M^{-2 \rho}$), in terms of recursively defined effective coupling constants. We therefore start from the connected Schwinger functions decomposed over scale attributions compatible with the cut-off:
\beq\label{bare_s}
\cS_N^\rho = \sum_{\cG , \mu | \mu \leq \rho} \frac{1}{s(\cG)} \left( \prod_{b \in \cB} ( - t_{b}^{\rho} )^{n_b (\cG)} \right) \cA_{\cG , \mu } \,.
\eeq 
In this formula, the sum runs over connected graphs, and $b$ spans all possible interactions, including mass and wave-function counter-terms. Starting from the highest scale $\rho$, we want to construct a set of $\rho + 1$ effective coupling constants per interaction $b$, called $t_{b , i}^{\rho}$ with $0 \leq i \leq \rho$. They will be formal power series in the bare coupling constants $t_{b , \rho}^{\rho} \equiv t_{b}^{\rho}$, such that $t_{b , i}^{\rho}$ is obtained from $t_{b , i + 1}^{\rho}$ by adding to it all the counter-terms associated to high subgraphs at scale $i + 1$. In order to make this statement more precise, it is useful to define $i_b (\cG , \mu)$ as the scale of a vertex $b$ in a graph $(\cG , \mu)$ as:
\beq\label{scale_coupling}
i_b ( \cG , \mu ) \equiv \max \{ i_l (\mu) \vert l \in L_b (\cG) \} \,,
\eeq
where $L_b (\cG)$ is the set of lines of $\cG$ which are hooked to $b$. We aim at a re-writing of \eqref{bare_s} of the form:
\beq\label{s_eff}
\cS_N^\rho = \sum_{\cG , \mu | \mu \leq \rho} \frac{1}{s(\cG)} \left( \prod_{b \in \cB(\cG) } ( - t_{b , i_b ( \cG , \mu )}^{\rho} ) \right) \cA_{\cG , \mu}^{eff} \,,
\eeq
in which the bare coupling constants have been substituted by effective ones at the scale of the bubbles making a graph $(\cG , \mu)$, and the new effective amplitudes are free of divergences. Thanks to the multiscale analysis, we know exactly which face-connected subgraphs are responsible for the divergences of a bare amplitude $\cA_{\cG , \mu}$: they are the high divergent subgraphs, which is a subset of all the quasi-local subgraphs. Unfortunately, they cannot play the leading role in the effective expansion: the divergences in a slice $i+1$ must be packaged into vertex-connected components, and \emph{reabsorbed in effective vertices with external propagators at scales lower or equal to $i$}. This condition on the external scales makes it impossible to act on a face-connected divergent subgraph independently of what it is vertex-connected to. Our language is therefore not adapted to the effective expansion. In order to make this point clearer, \emph{let us assume for the moment that the divergent subgraphs, the GN tree and the $\tau$ contraction operators are defined on the basis of vertex-connectedness}. In this provisional acceptation of the terms, let us moreover call $D_\mu (\cG)$, the \textit{forest of high divergent subgraphs} of $(\cG , \mu)$. The effective amplitudes are then deduced from the bare ones by subtracting the local part of each high divergent subgraph \cite{Riv}:
\beq\label{a_eff}
\cA_{\cG , \mu}^{eff} = \prod_{m \in D_\mu (\cG)} ( 1 - \tau_m ) \cA_{\cG , \mu}\,.
\eeq
Finiteness of $\cA_{\cG , \mu}^{eff}$ in the limit of infinite cut-off is then guaranteed. 
In order to make this prescription consistent, we need to reabsorb contributions of the form
\beq
\tau_m \cA_{\cG , \mu}
\eeq
into the effective coupling constants. 
This can be made more precise by defining an inductive version of (\ref{s_eff}):
\beq\label{s_eff_i}
\cS_N^\rho = \sum_{\cG , \mu | \mu \leq \rho} \frac{1}{s(\cG)} \left( \prod_{b \in \cB(\cG)} ( - t_{b , \sup( i , i_b ( \cG , \mu ))}^{\rho} ) \right) \cA_{\cG , \mu}^{eff , i} \,,
\eeq
with
\beq\label{a_eff_i}
\cA_{\cG , \mu}^{eff , i} \equiv \prod_{m \in D_\mu^i (\cG)} ( 1 - \tau_m ) \cA_{\cG , \mu}
\eeq
and
\beq
D_\mu^i (\cG) \equiv \{ m \in D_{\mu} (\cG) \vert i_m > i \} \,.
\eeq
We now proceed to prove (\ref{s_eff_i}), by induction on $i$, which at the same time will provide the recursive relation for the effective coupling constants. For $i = \rho$, (\ref{s_eff_i}) coincides with the bare expansion (\ref{bare_s}), and therefore holds true. Assuming that it holds at rank $i + 1$, let us then see how to prove it at rank $i$. The difference between $\cA_{\cG , \mu}^{eff , i}$ and $\cA_{\cG , \mu}^{eff , i + 1}$ amounts to counter-terms in $D_{\mu}^{i} (\cG) \setminus D_{\mu}^{i + 1} (\cG) = \{ M \in D_{\mu} (\cG) \vert i_M = i + 1 \}$, hence:
\beq
\cA_{\cG , \mu}^{eff , i} - \cA_{\cG , \mu}^{eff , i + 1} = \sum_{S \subset D_{\mu}^{i} (\cG) \setminus D_{\mu}^{i + 1} (\cG) \atop S \neq \emptyset} \prod_{M \in S} (- \tau_M ) \prod_{m \in D_{\mu}^{i + 1} (\cG)} ( 1 - \tau_m ) \cA_{\cG , \mu} .
\eeq
Adding and subtracting this quantity to $\cA_{\cG , \mu}^{eff , i + 1}$ in the equation (\ref{s_eff_i}) at rank $i+1$, one obtains a new equation which now involves $\cA_{\cG , \mu}^{eff , i}$ (thanks to the term added), together with a sum over subsets $S$ (due to the term subtracted). In condensed notations, this can be written as:
\beq\label{s_eff_intermediate}
\cS_N^\rho = \sum_{ (\cG , \mu , S) , \mu \leq \rho \atop S \subset D_{\mu}^{i} (\cG) \setminus D_{\mu}^{i + 1} (\cG) }  \frac{1}{s(\cG)} \left( \prod_{b \in \cB(\cG)} ( - t_{b , \sup( i + 1 , i_b ( \cG , \mu ))}^{\rho} ) \right) \cA_{\cG , \mu , S}^{eff , i} \,,
\eeq
where
\beq
\cA_{\cG , \mu , S}^{eff , i} \equiv - \prod_{M \in S} (- \tau_M ) \prod_{m \in D_{\mu}^{i + 1} (\cG)} ( 1 - \tau_m ) \cA_{\cG , \mu} 
\eeq
when $S \neq \emptyset$ and $\cA_{\cG , \mu , \emptyset}^{eff , i} \equiv \cA_{\cG , \mu}^{eff , i}$. The elements in a set $S$ being vertex-disjoint, we can contract them independently, and absorb the terms associated to $S \neq \emptyset$ into effective coupling constants at scale $i$. 

However, in order to correctly take wave function counter-terms into account, one needs to slightly generalize the notion of contraction previously defined for strictly tensorial interactions. While $\tau_M$ extracts amplitudes of contracted graphs times a pre-factor when $\omega(M) = 0$ or $1$, it is rather a sum of two terms when $\omega(M) = 2$: a zeroth order term proportional to a contracted amplitude, and a second order term proportional to a contracted amplitude supplemented with a Laplacian insertion as in (\ref{wf_ct}). In the latter case, one shall therefore decompose the operators as sums of two operators
\beq
\tau_M = \tau_M^{(0)} + \tau_M^{(2)}
\eeq
corresponding to the two types of counter-terms\footnote{Similarly, one defines $\tau_M \equiv \tau_M^{(0)}$ if $\omega(M) = 0$ or $1$.}. Developing these products, one ends up with a formula akin to (\ref{s_eff_intermediate}), provided that sets $S$ are generalized to
\beq
\hat{S} \equiv \{ (M , k_M) \vert M \in S, \, k_M \in \{ 0 , 2\}, \, k_M \leq \omega(M) \} \,,
\eeq
and that $\tau$ operators are replaced by $\tau_{\hat{M}} \equiv \tau_M^{(k_M)}$ for $\hat{M} = (M , k_M)$.
 Taking the scale attributions into account, we are lead to define the \textit{collapse} $\phi_i$, which sends triplets $( \cG , \mu , \hat{S} )$ with $S \subset D_{\mu}^{i} (\cG) \setminus D_{\mu}^{i + 1} (\cG)$ to its contracted version $(\cG' , \mu' , \emptyset)$. $\cG' \equiv \cG / \hat{S}$ is the graph obtained after the elements of $\hat{S}$ have been contracted, understood in a generalized sense: $\cG / (M , k_M)$ is equivalent to $\cG / M$ when $k_M = 0$ or $1$, and is a graph in which the $2$-point divergent subgraph $M$ has been replace by a Laplace operator if $k_M = 2$. As for $\mu'$, it is simply the restriction of $\mu$ to lines of $\cG'$. The bubbles of $\cG'$ are thought of as new effective interactions, obtained from contractions of vertex-connected graphs 
 . We can therefore factorize the sum in (\ref{s_eff_intermediate}) as:
\beq\label{s_eff_intermediate-2}
\cS_N^\rho = \sum_{\cG' , \mu' } \sum_{ (\cG , \mu , \hat{S}) , \mu \leq \rho \atop \phi_i (\cG , \mu , \hat{S}) = (\cG', \mu', \emptyset) }  \frac{1}{s(\cG)} \left( \prod_{b \in \cB(\cG)} ( - t_{b , \sup( i + 1 , i_b ( \cG , \mu ))}^{\rho} ) \right) \cA_{\cG , \mu , \hat{S}}^{eff , i} \,,
\eeq
with
\beq
\cA_{\cG , \mu , \hat{S}}^{eff , i} \equiv - \prod_{\hat{M} \in \hat{S}} (- \tau_{\hat{M}}) \prod_{m \in D_{\mu}^{i + 1} (\cG)} ( 1 - \tau_m ) \cA_{\cG , \mu} \,
\eeq
when $S \neq \emptyset$. In this last equation, one can act first with $\underset{\hat{M} \in \hat{S}}{\prod} (- \tau_{\hat{M}})$ on $\cA_{\cG , \mu}$. 
After reorganizing all the terms in (\ref{s_eff_intermediate-2}), it is easy to understand how the coupling constants at scale $i$ must be defined. For instance, assuming that $\cG$ has no quadratic divergences to avoid overloaded notations, one notices that (\ref{s_eff_intermediate-2}) reduces to (\ref{s_eff_i}) provided that:
\bes\label{rec}
\frac{1}{s(\cG')} \prod_{b' \in \cB(\cG')} ( - t_{b' , \sup( i , i_{b'} ( \cG' , \mu ))}^{\rho} ) &=& 
\sum_{ (\cG , \mu , S) , \mu \leq \rho \atop \phi_i (\cG , \mu , S) = (\cG', \mu', \emptyset) }  \frac{1}{s(\cG)} 
\prod_{b \in \cB(\cG)} ( - t_{b , \sup( i + 1 , i_b ( \cG , \mu ))}^{\rho} ) \nn \\
&& \prod_{m \in D_\mu^{i + 1} (\cG)} (1 - \tau_m ) \prod_{M \in S } (- \tau_{M}) \cA_{M , \mu} .
\ees
Exploiting the fact that the symmetry factor of a graph $g$, $s(g)$, is the number of permutations of its external legs leaving its colored structure invariant, we can write
\beq
s(\cG) = s(\cG') \prod_{M \in D_\mu^{i} (\cG) \setminus D_\mu^{i + 1} (\cG)} s(M)\,.
\eeq
It is a consequence of the factorization of the symmetry group associated to $\cG$ into a product of: the symmetry group associated to $\cG'$, and the symmetry groups associated to all the subgraphs $M$ inserted in $\cG'$.
We can therefore readily extract a solution for \eqref{rec}, in the form of a definition of the effective coupling constants at rank $i$:
\bes\label{scale_induction}
- t_{b , i}^{\rho} &=& - t_{b , i + 1}^{\rho} - \sum_{ (\cH , \mu , \{M\}) , \mu \leq \rho \atop  \phi_i (\cH , \mu , \{M\}) = ( b , \mu , \emptyset)} \frac{1}{s(\cH)} \left( \prod_{b' \in \cB(\cH)} ( - t_{b' , i_{b'} ( \cH , \mu )}^{\rho} )  \right) \nn \\
&\times& \; \left( \prod_{m \in D_{\mu} (\cH) \setminus \{M\}} ( 1 - \tau_m ) \right) (- \tau_{M} ) \, \cA_{M , \mu}\,.
\ees
This concludes the proof of the existence of the effective expansion when vertex-connectedness is used to organize the counter-terms. Had we relied on face-connectedness instead, equation (\ref{rec}) would have had coupling constants at scale $i+1$ also on the left-hand side, which would have made the whole scheme inconsistent with definition (\ref{scale_coupling}). 

\
By construction, $\{ t_{b , \rho}^{\rho} \, , \, b \in \cB \}$ are interpreted as the bare coupling constants. Accordingly, the renormalized constants are to be found at the other end of the scale ladder, namely in the last infrared slice, which corresponds to external legs. This is compatible with a renormalized coupling being defined as the full amputated function corresponding to the type of interaction considered. It can be checked that the latter amounts to set 
\beq
t_{b , ren}^{\rho} \equiv t_{b , -1}^{\rho} \,,
\eeq
and we could look for yet another reshuffling of the Schwinger functions, this time as multi-series in $\{ t_{b , ren}^{\rho} \}$. 

However, we follow a different strategy for the renormalized expansion, and close the vertex-connected parenthesis. Divergent graphs and contraction operators are now again understood in the face-connected sense we advocate in this paper. 
The natural induction with respect to scales (\ref{scale_induction})
is not available anymore, but can be partially encapsulated into the definition of counter-terms according to an induction with respect to the number of vertices in a diagram. This is nothing but the well-known Bogoliubov induction, which provides the infinite set of counter-terms to be added to the bare Lagrangian. In our case, the induction takes the form:
\beq\label{bogo}
c_\cG = \sum_{\{ g_1, \ldots , g_k \} } \prod_{m \in S} ( - \tau_m ) \cA_{m / \{ g \}} \prod_{i = 1}^{k} c_{g_i} \,, 
\eeq
where $\cG$ is a vertex-connected graph with all its face-connected components $m \in S$ divergent, $c_\cG$ its associated counter-term, and $\{ g_1 , \ldots , g_k \}$ runs over all possible families of disjoint vertex-connected divergent subgraphs of $\cG$, for which counter-terms $\{ c_{g_i} \}$ have been defined at an earlier stage of the induction. Note also that $\cA_{m / \{ g \}}$ is a short-hand notation for the part of the amplitude associated to $m$, once the $g_i$'s it contains have been contracted. Each of these counter-terms will contribute to the renormalization of a coupling constant (or several when quadratically divergent subgraphs are present). More precisely, one has:
\beq
t_b^\rho = t_{b , ren}^{\rho} + \sum_{n = 1}^{+ \infty} c_n^{b} (t_{b , ren}^{\rho})^n \,,
\eeq
where $c_n^b$ is the sum of all the counter-terms $c_\cG$ at order $n$ of the type $b$ \footnote{The same subtlety as in the previous discussion occurs for quadratically divergent contributions: one has to split the counter-terms $c_\cG$ into mass and wave-function contributions. We kept this step implicit here in order to lighten the notations.}.

It is then a well-known fact that a (formal) perturbative expansion in these new variables generates renormalized amplitudes expressed by Zimmermann's forest formula. The forests appearing in this formula can be called \textit{inclusion forests}, since they are sets of subgraphs $\cF$ with specific inclusion properties: for any $h_1 , h_2 \in \cF$, either $h_1$ and $h_2$ are line-disjoint, or one is included into the other. In this model, the relevant forests are inclusion forests of vertex-connected subgraphs with all their face-connected components divergent. Since each of the graphs in the forests is acted upon by a product of contraction operators $( - \tau_{m} )$, one for each face-connected component, and since in addition face-connectedness is a finer notion than vertex-connectedness, one can actually work with inclusion forests of face-connected subgraphs. Moreover, one needs to strengthen their definition by emphasizing face-disjointness rather than line-disjointness. To avoid any terminology confusion with the usual notion of inclusion forest, we call this new type of forests \textit{strong inclusion forests}.
 
\begin{definition}
Let $\cH \subset \cG$ be a subgraph.
A \textit{strong inclusion forest} $\cF$ of $\cH$ is a set of non-empty and face-connected subgraphs of $\cH$, such that: 
\begin{enumerate}[(i)]
\item for any $h_1 , h_2 \in \cF$, either $h_1$ and $h_2$ are line-disjoint, or one is included into the other; 
\item any line-disjoint $h_1 , \ldots , h_k \in \cF$ are also face-disjoint.
\end{enumerate} 
\end{definition}

A few remarks are in order. First, a strong inclusion forest $\cF$ is always an inclusion forest (condition (i)), hence the nomenclature. 
Second, it is important to understand that the Zimmermann forests relevant to our model are \textit{strong} inclusion forests. To this effect, notice for instance that if $g_1, \ldots ,  g_k \subset \cG$ appear in a same term of the Bogoliubov recursion (\ref{bogo}) for some intermediate subgraph $\cH \subset \cG$, then they form $k$ distinct face-connected components in $\cH$. The existence of such a subgraph is equivalent to the face-disjointness of $g_1 , \ldots ,  g_k$. Third, we point out that this strong notion of inclusion forest was already hinted at in \cite{COR}, through the specific case of \textit{meloforests}, which are 
forests of face-connected melopoles. In the wider context of the present paper, we modify slightly this terminology, and call \textit{meloforest} any strong inclusion forest of \textit{melonic} subgraphs. 
Finally, we simply call \textit{divergent forest} a strong inclusion forest of divergent subgraphs, and note $\cF_{D} (\cG)$ the set of divergent forests of a graph $\cG$ (including the empty forest). In the $\SU(2)$, $d=3$ model, divergent forests are also meloforests, but the converse is not true.

\

In this language, the renormalized amplitudes are related to the bare ones through:
\beq\label{a_ren}
\cA_\cG^{ren} = \left( \sum_{\cF \in \cF_D (\cG)} \prod_{m \in \cF} \left( - \tau_{m} \right) 
\right) \cA_\cG \,.
\eeq 

\

In order to prove the finiteness of the renormalized amplitudes, one should rely on the refined understanding of the divergences provided by the multi-scale expansion. To this effect, we will expand equation (\ref{a_ren}) over scales. For fixed scale attribution, contraction operators acting on high divergent subgraphs will provide a convergent power-counting. The sum over scales will finally be achieved thanks to an adapted classification of divergent forests, which is the purpose of the next section.

%

\subsection{Classification of forests}

Before discussing the classification in details, we point out an intriguing property of this model. In light of Proposition \ref{curiosity}, we notice that the melonic subgraphs of a given non-vacuum graph $\cG$ organize themselves into an inclusion forest. It would be therefore tempting to conjecture that they also form a strong inclusion forest (i.e. a meloforest). However, we can actually find examples of overlapping melonic subgraphs, showing that this is incorrect (see Figure \ref{overlap}). Still, and again by Proposition \ref{curiosity}, we notice that the union of two melonic subgraphs cannot be itself melonic, hence cannot be divergent. Therefore, if we restrict our attention to divergent forests, we can actually prove that the previous conjecture hold.

\begin{figure}[h]
\begin{center}
\includegraphics[scale=0.5]{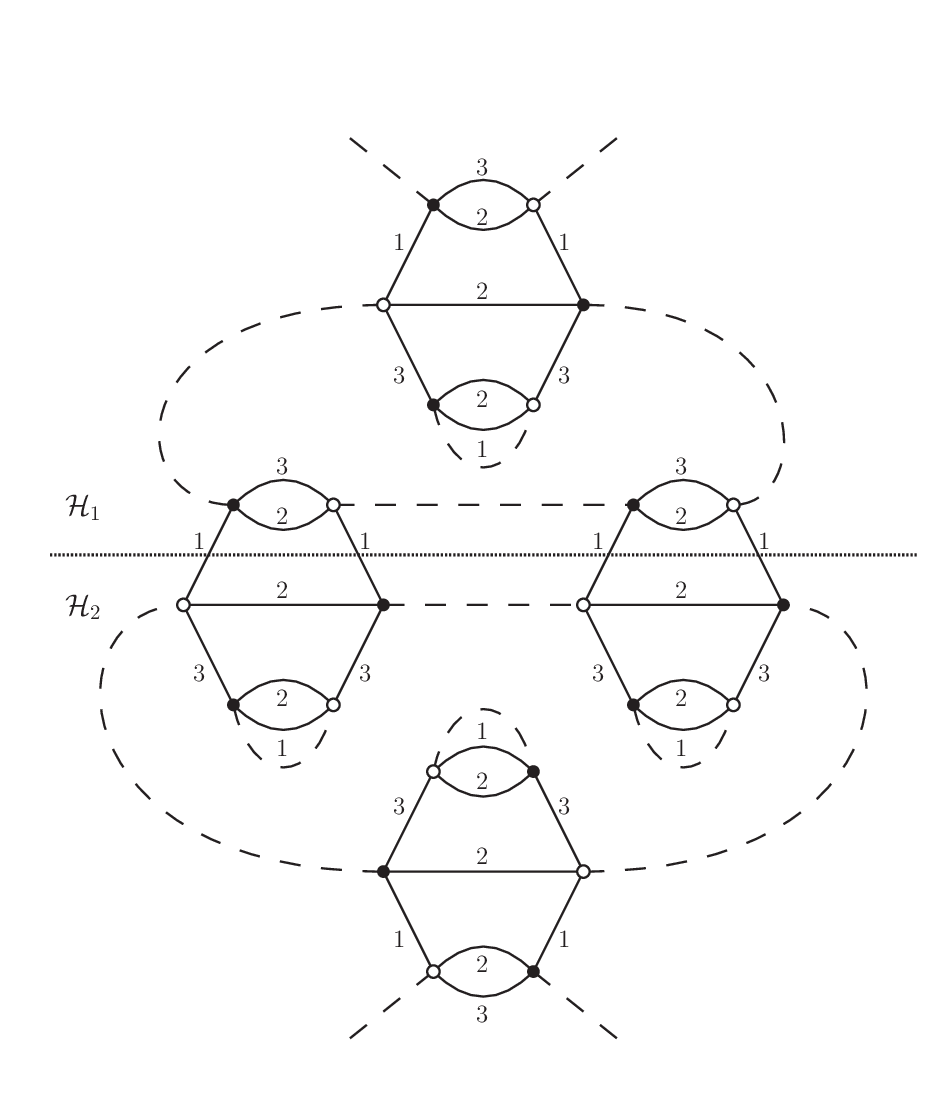}
\caption{Two melonic subgraphs $\cH_1$ and $\cH_2$ ($\cH_2$ being even divergent) which are face-connected in their union.}
\label{overlap}
\end{center}
\end{figure}

\begin{proposition}
Let $\cG$ be a non-vacuum graph. The set of divergent subgraphs of $\cG$ is a strong inclusion forest. We denote it $D(\cG)$.
\end{proposition} 
\begin{proof}
Thanks to proposition \ref{curiosity}, we already know that $D(\cG)$ is an inclusion forest. 
To conclude, we need to show that there exists no subset of line-disjoint subgraphs in $D(\cG)$ which are not also face-disjoint. If this would not be the case, we could certainly find line-disjoint subgraphs $\cH_1 , \ldots , \cH_k \in D(\cG)$ which are face-connected in their union. 
This face-connectedness is necessarily ensured by external faces of the $\cH_1 , \ldots , \cH_k$ which arrange together into internal faces of $\cH_1 \cup \dots \cup \cH_k$. Because their intersection is empty, this can only be achieved if some vertices of $\cH_1 \cup \dots \cup \cH_k$ are shared by several subgraphs $\cH_1 , \ldots , \cH_k$. Let us call $p_{2i}^s$ ($1 \leq i \leq 3$, $2 \leq s \leq 3$) the number of vertices of valency $2i$ which are shared by exactly $s$ subgraphs $\cH_k$. They can be related to the valency of $\cH_1, \ldots , \cH_k$ and $\cH_1 \cup \dots \cup \cH_k$ by the formula:
\beq
N (\cH_1 \cup \dots \cup \cH_k) = \sum_{j = 1}^{k} N (\cH_j) - \sum_{i =1}^3 \sum_{s = 2}^{3} (2 i) (s - 1) p_{2 i}^s \,.
\eeq
This just says that when summing all the individual valencies, one needs to subtract all the contributions of external legs of the connecting vertices, which have been over counted, in order to find the valency of the full subgraph. If a connecting vertex $v$ is connected to exactly $s$ subgraphs $\cH_k$, its external legs have been counted exactly $s - 1$ too many times. Furthermore, the conditions $\omega(\cH_j) \geq 0$ can be summed to yield\footnote{The $\rho$ contributions are all $0$ since $\cH_1, \ldots , \cH_k$ are non-vacuum and melonic}:
\beq
\sum_{j = 1}^k N(\cH_j) \leq 6 k - 4 \sum_{j = 1}^k n_2 (\cH_j) - 2 \sum_{j = 1}^k n_4 (\cH_j) \,. 
\eeq
Remarking that $\sum_{s = 2}^3 s p_{2i}^s \leq \sum_{j = 1}^k n_{2i} (\cH_j)$ for all $i$, we can finally deduce from the two previous inequalities that:
\beq
N(\cH_1 \cup \dots \cup \cH_k) \leq 6 k  - 6 \sum_{i = 1}^{3} \sum_{s = 2}^3 (s -1) p_{2 i}^{s} - 2 \sum_{s = 2}^3 p_{4}^{s} - 4 \sum_{s = 2}^3 p_{2}^{s}\,.
\eeq
We immediately notice that whenever
\beq
\sum_{i = 1}^{3} \sum_{s = 2}^3 (s -1) p_{2 i}^{s} \geq k \,,
\eeq
$\cH_1 \cup \dots \cup \cH_k$ is vacuum, which contradicts the hypothesis that $\cG$ is not. If the previous inequality is not verified, one has instead
\beq\label{tree_ineq}
\sum_{i = 1}^{3} \sum_{s = 2}^3 s p_{2 i}^{s} \leq k + \sum_{i = 1}^{3} \sum_{s = 2}^3 p_{2 i}^s - 1 \,.
\eeq
In order to understand the meaning of this inequality, let us introduced an abstract graph $G$: its nodes are the subgraphs $\cH_1 , \ldots , \cH_k$ and all the vertices shared by more than one subgraph; two nodes are linked by one line in $G$ if and only if one of them is a subgraph, and the other a vertex contained in this subgraph. In equation (\ref{tree_ineq}), on the left-hand side one finds the number of links in $G$, and on the right side its number of nodes minus $1$. Therefore, when the inequality is saturated $G$ is a tree, and when the inequality is strict it is not connected. The latter case is contradictory with our hypotheses. As for when $G$ is a tree, one can find spanning trees $\cT_1 \subset \cH_1 , \ldots , \cT_k \subset \cH_k$ such that there union $\cT \equiv \cT_1 \cup \dots \cup \cT_k$ is a spanning tree of $\cH_1 \cup \dots \cup \cH_k$. But in such a situation, $( \cH_1 \cup \dots \cup \cH_k ) / \cT = (\cH_1 / \cT_1 ) \cup \dots \cup ( \cH_k / \cT_k)$ would be a melopole, contradicting the fact that $\cH_1 \cup \dots \cup \cH_k$ cannot be melonic (see proposition \ref{curiosity}).
\end{proof}

At this stage, we tend to see this result as a curiosity of the specific model we are considering, and only a detailed study of other just renormalizable models of this type could confirm it to have a wider validity. What is sure is that it is by no means essential to the classification of forests. Still, it allows significant simplifications, which we will take advantage of in the following, in the notations and proofs, since the divergent forests of $\cG$ are exactly the subsets of $D(\cG)$.

\

Recall that $D_\mu (\cG)$ denotes the set of truly divergent subgraphs in $\cG$ for the scale attribution $\mu$. It is a divergent forest, as we already knew from the fact that, modulo $\cG$ itself, it consists exactly in the subgraphs appearing in the GN tree of $(\cG , \mu)$. Furthermore, we now know it to be a subforest of $D(\cG)$. We call its complementary part
\beq
I_\mu (\cG) \equiv D(\cG) \setminus D_\mu (\cG)
\eeq
the \textit{innofensive part} of $D (\cG)$ at scale $\mu$, since it is the set of divergent subgraphs of $\cG$ which do not appear in the GN tree of $(\cG , \mu)$, and therefore do not contribute to divergences at this scale. 

In this model, where the disjoint decomposition $D(\cG) = I_\mu (\cG) \cup D_\mu (\cG)$ involves three sets which are themselves divergent forests, the classification of forest is as trivial as saying that choosing a forest in $D(\cG)$ amounts to choosing a forest in $I_\mu (\cG)$ and a forest in $D_\mu (\cG)$, namely:
\beq
\cF_D (\cG) = \{ \cF_1 \cup \cF_2 \vert \cF_1 \subset I_\mu (\cG) \, , \cF_2 \subset D_\mu (\cG)\}\,.
\eeq
We can use this simple fact in the decomposition of equation (\ref{a_ren}) over scale attributions
\bes
\cA_\cG^{ren} &=& \sum_{\mu} \sum_{\cF \in \cF_D (\cG)} \prod_{m \in \cF} ( - \tau_{m} ) \cA_{\cG, \mu } \\
&=& \sum_{\mu} \sum_{\cF_1 \subset I_\mu (\cG)} \sum_{\cF_2 \subset D_\mu (\cG)} \prod_{m \in \cF_1 \cup \cF_2} ( - \tau_{m} ) \cA_{\cG, \mu }\,.
\ees
We then exchange the first two sums:
\beq
\cA_\cG^{ren} = \sum_{\cF_1 \subset D(\cG)} \, \sum_{\mu \vert \cF_1 \subset I_\mu (\cG)} \prod_{m \in \cF_1} ( - \tau_{m} ) \sum_{\cF_2 \subset D_\mu (\cG)} \prod_{h \in \cF_2} ( - \tau_{m} ) \cA_{\cG, \mu }\,.
\eeq
We can finally reorganize the contraction operators associated to graphs of $D_\mu (\cG)$ to obtain:
\bes
\cA_\cG^{ren} &=& \sum_{\cF \subset D(\cG)} \cA_{\cG , \cF}^{ren}\,,\\
\cA_{\cG , \cF}^{ren} &\equiv& \sum_{\mu | \cF \subset I_\mu (\cG)} \prod_{m \in \cF} (- \tau_m ) \prod_{h \in D_\mu (\cG)} (1 - \tau_h ) \cA_{\cG , \mu}\,.
\ees

This way of splitting the contributions of the different forests according to the scales is in phase with the multi-scale analysis. We shall explain in the next two paragraphs why $\cA_{\cG , \cF}^{ren}$ is convergent. To this effect, we first use the contraction operators indexed by elements of $D_\mu (\cG)$ to show that the renormalized power-counting is improved with respect to the bare one, in such a way that all divergent subgraphs become power-counting convergent. In a second step, we will explain how these decays can actually be used to perform the sum over scale attributions.

\subsection{Convergent power-counting for renormalized amplitudes}

We fix a divergent forest $\cF \in D(\cG)$ and a scale attribution $\mu$ such that $\cF \subset I_\mu (\cG)$. We want to find a multi-scale power-counting bound for 
\beq
\prod_{m \in \cF} (- \tau_m ) \prod_{h \in D_\mu (\cG)} (1 - \tau_h ) \cA_{\cG , \mu}\,.
\eeq
Since contraction operators commute, we are free to first act on $\cA_{\cG , \mu}$. In order to properly encode the two possible Taylor orders in $2$-point divergences, we should reintroduce the generalized notations $\hat{m}$ and $\tau_{\hat{m}}$, together with a generalized notion of divergent forest $\hat{\cF}$. Since the argument we are about to make is insensible to such subtleties, and its clarity would be somewhat affected by the heavy notations, we decide instead to assume that $\cF$ does not contain any quadratically divergent subgraph. It is easily understood that the action of the product of contraction operators disconnects parts of the amplitudes, yielding a product of pieces of the integrand integrated on their internal variables. The exact formula is
\beq
\prod_{m \in \cF} \tau_m \cA_{\cG , \mu} = \cA_{\cG / A_\cF (\cG) } \prod_{m \in \cF} \nu_\mu ( m / A_\cF (m) ) \,,
\eeq
where $A_\cF (m) \equiv \{ g \subset m \vert g \in \cF \}$ is the set of \textit{descendants} of $m$ in $\cF$, and $\nu_\mu ( m / A_\cF (m) )$ is the \textit{amputated amplitude}\footnote{We mean by that that the contributions of external faces are discarded.}
of $m$ contracted by its descendants. The power-counting of each subgraph appearing on the right-hand side of this formula is known, yielding:
\beq
\vert \prod_{m \in \cF} (- \tau_m ) \cA_{\cG , \mu} \vert
\leq K^{L(\cG)} \prod_{m \in \cF \cup \{ \cG \}} \prod_{(i , k)} M^{\omega[ (m / A_\cF (m))_i^{(k)}]} \,. 
\eeq
As expected, we see that the contraction operators associated to inoffensive forests does not improve the power-counting, and are in a sense useless.

On the other hand, we have also seen in section \ref{sec:remainders}, that $(1 - \tau_h )$ operators effectively render subgraphs $h \subset D_\mu (\cG)$ power-counting convergent. We can use this improved power-counting in each $m / A_\cF (m)$ to prove the following proposition:
\begin{proposition}
There exists a constant $K$, such that for any divergent forest $\cF \in D( \cG )$:
\beq\label{improved}
\vert \cA_{\cG , \cF}^{ren} \vert \leq K^{L(\cG)} \sum_{\mu \vert \cF \subset I_\mu (\cG)} \prod_{m \in \cF \cup \{ \cG \}} \prod_{(i , k)} M^{\omega'[ (m / A_\cF (m))_i^{(k)} ]} \, ,
\eeq
where
\beq 
\omega'[ ( m / A_\cF (m) )_i^{(k)} ] = \min \{ -1 , \, \omega[ (m / A_\cF (m))_i^{(k)} ] \}
\eeq
except when $m \in \cF$ and $(m / A_\cF (m))_i^{(k)} = m / A_\cF (m)$, in which case $\omega'[ m / A_\cF (m) ] = 0$.
\end{proposition}
\begin{proof}
If $m$ is compatible with $\cF$ (i.e. $\cF \cup \{ m \}$ is also a strong inclusion forest), we denote by $B_\cF (m)$ the ancestor of $m$ in $\cF \cup \{ m \}$. This notion allows to decompose the product of useful contraction operators as
\beq
\prod_{h \in D_\mu (\cG)} (1 - \tau_h ) = \prod_{m \in \cF \cup \{ \cG\}} \prod_{h \in D_\mu (\cG) \atop B_\cF (h ) = m} (1 - \tau_h )\,.
\eeq
When multiplying this expression by $\underset{m \in \cF}{\prod} (- \tau_m )$, one obtains
\bes
\vert \prod_{m \in \cF} (- \tau_m ) \prod_{h \in D_\mu (\cG)} (1 - \tau_h ) \cA_{\cG , \mu} \vert &=& \left( \prod_{h \in D_\mu (\cG) \atop B_\cF ( h ) = \cG} (1 - \tau_{h / A_\cF (\cG)} ) \vert \cA_{\cG / A_\cF (\cG) , \mu} \vert \right) \nn \\
&& \times \left( \prod_{m \in \cF } \prod_{h \in D_\mu (\cG) \atop B_\cF (h ) = m} (1 - \tau_{h / A_\cF (m)} ) \, \vert \nu_\mu ( m / A_\cF (m) ) \vert \right)
\ees
We recognize in this formula all the useful contractions associated to high divergent subgraphs in each $m / A_\cF (m)$, for which the new degree is at most $-1$, except possibly for the roots $m = m / A_\cM (m)$\footnote{This root can indeed itself be divergent.} when $m \neq \cG$. But because the corresponding amplitudes are amputated, they contribute to the power-counting with a degree $0$. 
\end{proof}

\subsection{Sum over scale attributions}

The improved power-counting (\ref{improved}) allows to decompose renormalized amplitudes into fully convergent\footnote{We call fully convergent a graph whose face-connected subgraphs all have convergent power-counting.} pieces associated to the contracted subgraphs $m / A_\cF (m)$. We therefore decompose the task of summing over scale attributions into two steps: we will first recall how this can be performed maintaining a bound in $K^n$ for a fully convergent graph $\cG$; we will then explain how this generalizes to arbitrary renormalized amplitudes, the price to pay being possible factorial growths in $n$ due to contraction operators associated to the inoffensive forests $I_\mu (\cG)$.

\

Let $\cG$ be a fully convergent, vertex-connected, and non-vacuum graph. For any face-connected subgraph $\cH \subset \cG$ such that $\cF(\cH) \neq 0$, we have seen that
\beq
\omega (\cH ) \leq -\frac{N(\cH)}{2}\,.
\eeq
Moreover, $\omega(\cH) = -2$ and $N(\cH) \leq 10$ when $F(\cH) = 0$, therefore one can use a slower decay in $-N (\cH) / 5$ and write
\beq\label{cv_pc}
\cA_{\cG , \mu} \leq K^{L (\cG)} \prod_{(i , k)} M^{- N (\cG_i^{(k)}) / 5}
\eeq
for any scale attribution $\mu$. In order to extract a sufficient decay in $\mu$ from (\ref{cv_pc}), it is crucial to focus on the scales associated to the vertices of $\cG$. Let us therefore introduce $L_b (\cG)$ the set of external lines of a bubble $b \in \cB (\cG)$, and define:
\beq
i_b (\mu) = \sup_{l \in L_b (\cG)} i_l (\mu) \, , \qquad e_b (\mu) = \inf_{l \in L_b (\cG)} i_l (\mu)\,.
\eeq  
The main interest of these two scales lies in the two following facts: a) $b$ touches a high subgraph $\cG_i^{(k)}$ if and only if $i \leq i_b (\mu)$; b) moreover, $b$ is an external vertex of $\cG_i^{(k)}$ if and only if $e_b (\mu) < i \leq i_b (\mu)$. Accordingly, and because $b$ touches at most $6$ high subgraphs, one can distribute a fraction of the decay in the number of lines of high subgraphs to the vertices of $\cG$:
\beq
\prod_{(i , k)} M^{- N (\cG_i^{(k)}) / 5} \leq \prod_{(i , k)} \prod_{b \in \cB(\cG_i^{(k)}) \vert e_b (\mu) < i \leq i_b (\mu)} M^{- 1 / 30} \,.
\eeq
Exchanging the two products yields the interesting bound:
\beq
\cA_{\cG , \mu} \leq K^{L (\cG)} \prod_{b \in \cB (\cG)} \prod_{(i , k) \vert e_b (\mu) < i \leq i_b (\mu)} M^{- \frac{i_b (\mu) - e_b (\mu)}{30}}  \,.
\eeq
Finally, we can distribute the decays among all possible pairs of external legs of each vertex. Since there are at most $6 \times 5 / 2 = 15$ such pairs, we get:
\beq
\cA_{\cG , \mu} \leq K^{L (\cG)} \prod_{b \in \cB (\cG)} \prod_{(l, l') \in L_b (\cG)  \times L_b (\cG) } M^{- \frac{\vert i_l (\mu) - i_l' (\mu) \vert }{450}}  \,.
\eeq
This bound implies the finiteness of $\cA_\cG$. To see this, we can choose a total ordering of the lines $L(\cG) = \{ l_1 , \ldots , l_{L(\cG)} \}$ such that $l_1$ is hooked to an external vertex of $\cG$, and $\{ l_1 , \ldots , l_m \}$  is connected for any $m \leq L(\cG)$. This allows to construct a map $j'$ on the indices $2 \leq j \leq L(\cG)$, such that $1 \leq j'(j) < j$, and\footnote{By convention, one also defines $i_{l_{j'(1)}} = - 1$.}:
\beq
\prod_{b \in \cB (\cG)} \prod_{(l, l') \in L_b (\cG)  \times L_b (\cG) } M^{- \frac{\vert i_l (\mu) - i_l' (\mu) \vert }{450}} \leq \prod_{j = 1}^{L(\cG)} M^{- \vert i_{l_j} (\mu) - i_{l_{j'(j)}} (\mu) \vert / 450 } \,.
\eeq
The sum over $\mu = \{ i_{l_1}, \ldots , i_{l_{L(\cG)}} \}$ of such a sum is uniformly bounded by a constant to the power $L(\cG)$, which proves the following theorem:
\begin{theorem}
There exists a constant $K>0$ such that, for any fully convergent, vertex-connected, and non-vacuum graph $\cG$:
\beq
\cA_\cG \leq K^{L(\cG)}\,.
\eeq
\end{theorem}

\

We can apply the same reasoning to the general power-counting (\ref{improved}). Let us fix $\cF$ a divergent forest. The only difference is that graphs $g / A_\cF (g)$ do not have any decay associated to their external legs. One therefore gets one additional scale index to sum over per element of $\cF$. But we can bound them by the maximal scale $i_{max} (\mu)$ in $\mu$ and write:
\bes
\vert \cA_{\cG , \cF}^{ren} \vert &\leq& K^{L(\cG)} \sum_{\mu \vert \cF \subset I_\mu (\cG)} \prod_{m \in \cF \cup \{ \cG \}} \prod_{(i , k)} M^{\omega'[ (m / A_\cF (m))_i^{(k)} ]} \\
&\leq& {K_1}^{L(\cG)} \sum_{i_{max} (\mu) } (i_{max} (\mu))^{\vert \cF \vert} M^{\delta i_{max} (\mu)} \, , 
\ees
where $\delta > 0$ and $K_1 > 0$ are some constants, and $\vert \cF \vert$ is the cardinal of $\cF$. The last sum over $i_{max} (\mu)$ can finally be bounded by $|\cF|! K^{|\cF|}$ for some constant $K > 0$. The final sum on $\cF \subset D(\cG)$ can be absorbed into a redefinition of the constants, since the number of divergent forests is simply bounded by $2^{| D(\cG) |}$. This concludes the proof of the BPHZ theorem.

\begin{theorem}
For any vertex-connected and non-vacuum graph $\cG$, the renormalized amplitude $\cA_\cG^{ren}$ has a finite limit when the cut-off $\Lambda$ is sent to $0$. More precisely, there exists a constant $K>0$ such that the following uniform bound holds:
\beq
\vert \cA_\cG^{ren} \vert \leq K^{L(\cG)} \vert D(\cG) \vert ! 
\eeq
\end{theorem} 
While this theorem proves the renormalizability of the model, it does not preclude the existence of renormalons, since the uniform bound we could find is only factorial. However, we notice that such an unreasonable growth can only exist because of the contractions operators associated to subforest of $I_\mu (\cG)$. On the contrary, if we were to focus on the effective expansion, where no unnecessary counter-terms enter the definitions, one would find instead a uniform bound like the one for fully convergent graphs.

\section*{Conclusion}
\addcontentsline{toc}{section}{Conclusion}

Let us summarize the main achievements of this article. We focused our attention on a particular class of group field theories, named {\it tensorial group field theories} which are characterized by: a) an infinite set of interactions, labeled by colored bubbles, based on a tensorial symmetry principle; b) non-trivial propagators implementing a gauge invariance condition on the fields, supplemented with a Laplace operator which softly breaks the tensorial invariance of the interaction. The first ingredient is suggested by recent work on tensor models and characterizes also the effective theory obtained by integrating out fields in colored group field theories based on simplicial interactions. The gauge invariance condition turns the Feynman amplitudes into lattice gauge theories and is one of the two main ingredients of group field theories for gravity (the other being, in 4d, the so-called simplicity constraints). The Laplace operator launches the renormalization group flow as seems also to be produced by quantum corrections in simpler topological models with ultralocal propagators. The rank $d$ of the tensors, as well as the dimension $D$ of the compact group indexing the tensors, were in a first stage kept arbitrary. A detailed analysis of the power-counting of such models allowed to derive stringent restrictions on $d$ and $D$ in order to achieve renormalizability. In particular, it was shown that only five combinations of such parameters can potentially support (interacting) just-renormalizable models. Among these, only $(d , D) = (3 , 3)$ can be directly related to a spacetime theory, namely topological BF theory or 3d quantum gravity, with $G$ the symmetry group for Lorentzian or Riemannian spaces of dimension $d$. In particular, the case $(4,6)$, that would correspond to the 4d topological BF theory from which one obtains gravitational models by imposing simplicity constraints, is found to be non-renormalizable.  We then went on to study in detail this particular model, in the Riemannian case $G = \SU(2)$. In order to classify the divergences, proven to be all melonic, we used multi-scale techniques. The tensorial interactions were shown to be renormalizable up to order $6$, and to generate up to quadratically divergent subgraphs. The same multi-scale techniques could then be used to reabsorb divergences into tensorial effective coupling constants, as well as wave-function counter-terms, thus defining renormalized amplitudes. Computed as sums over particular types of Zimmermann forests, they could finally be proven finite at all orders of perturbation, which is the main result of this paper. Along the way, many useful technical results could be gathered about melonic subgraphs, which will certainly be relevant to future works such as $\beta$-functions calculations. Additionally, divergent forests were found to be unexpectedly rigid in their structure, which helped simplifying some aspects of the proof of renormalizability.

\

The present study provides a few lessons which in our opinion will have to be kept in mind in the construction and renormalization analysis of more elaborate models, in particular models for 4d quantum gravity. First, concerning TGFTs \textit{per se}, the message we would like to convey is that, in order to efficiently index the divergences, the most appropriate notion of connectedness is face-connectedness rather than vertex-connectedness. This is particularly true in models implementing the gauge invariance condition, in which the amplitudes are functions of holonomies around faces. The natural coarse-graining procedure in this situation is indeed to erase "high energy" faces rather than internal lines, and this can be consistently implemented in what we called tracial subgraphs. Face-connectedness also crucially enters the power-counting theorem of such models, through the rank of the incidence matrix between faces and lines of a given graph. While a proof of renormalizability can certainly be achieved with a notion of vertex-connected high divergent subgraphs only, the face-connected high divergent subgraphs we relied on in this paper capture the fine structure of the divergences, and henceforth avoid many redundancies in the renormalization. This is exemplified by the fact that divergent subgraphs in the sense of face-connectedness do not overlap (in non-vacuum graphs), but rather organize themselves into a strong inclusion forest. Had we worked in the coarser vertex-connectedness picture, overlapping divergences would have been generic, and redundant counter-terms would have been introduced. An intriguing question to ask, in this respect, is whether face-connectedness might prove more fundamental in simpler TGFTs as well, for example in the original model \cite{bgriv}, where no connection degrees of freedom are introduced. On the other hand, we also remarked that the usual vertex-connected divergent graphs remain at the root of the effective expansion, hence we cannot take full advantage of face-connectedness in this Wilsonian context. This suggests an interesting analogy between face-connected graphs in TGFTs and $1$-particle irreducible graphs in ordinary quantum field theory, which deserves further investigation.

\

Let us now turn to the hard question of the renormalizability of quantum gravity models in four dimensions. While we do not have any definitive statement to make on this issue, since, as we explained earlier, the analysis presented here does not immediately generalize to models involving simplicity constraints, as the latter GFTs, beside the fact that they are not in the class of models considered in this paper, are also not based on group manifolds as such, but rather submanifolds of the Lorentz group. Still, it seems to us that the three dimensional $\SU(2)$ model studied in this paper suggests to reconsider and improve the current spin foam models for quantum gravity in two essential ways, before attempting any complete study of renormalizability. The first concerns the much debated nature of scales in such models, and the definition of non-trivial propagators which decay in the UV. In particular, we think that the results of the present paper suggest  that, in general, semi-classical reasoning interpreting the large-$j$ limit of spin foam models (the Feynman amplitudes of GFTs) as the IR general relativistic limit should be taken with care. Indeed, perturbative divergences being associated to the same large-$j$ sector (or equivalently small Schwinger parameter $\alpha$), it actually plays the role of the UV in our TGFT setting. In the Wilsonian point of view, it is therefore only for boundary states with scales much lower than the cut-off that the theory retains some predictive power. This points in the direction of large boundary geometries having to be constructed as collections of many small cells rather than a few big ones. And in practice, this means that one will have to address the question of approximate effective schemes, in order to control such regimes with large numbers of particles. An intriguing possibility would be the occurrence of one or several phase transitions along the renormalization flow. This scenario might already start to be tested in the three dimensional case, the first step being the computation of $\beta$-functions. In any case, we need to understand how to choose non-trivial kernels for propagators in four dimensional models. While there are some hints \cite{josephvalentin} that the Laplace-Beltrami operator is naturally generated by the quantum dynamics, when no simplicity constraints are imposed, whether the same is true in the presence of simplicity constraints is unclear at present. We believe this is the first open question to address as far as the renormalizabilty of TGFTs for four dimensional quantum gravity is concerned. The second point which deserves similar attention is the possible interplay between tensorial invariance and simplicity constraints, as it is not immediately clear whether the geometric meaning and motivations for such constraints, as well as the details of their implementation, straightforwardly generalize to bubble interactions. If these two important questions can be elucidated, one might try to apply the techniques used in this paper to determine whether four dimensional TGFT models for quantum gravity with such simplicity constraints are renormalizable or not, and up to which order of interactions. 

\section*{Acknowledgements}
S.C. thanks Joseph Ben Geloun for interesting discussions, and for inviting him at the Perimeter Institute while this article was under completion. 
This work is partially supported by a Sofja Kovalevskaja Award by the A. von Humboldt Stiftung, which is gratefully acknowledged.


\appendix

\section*{Appendix}

\section{Heat Kernel}

Consider the $S_3$ representation of $SU(2)$, the identity $\one$ being at the north pole,
$H_0 = S_2$ being the equator and $- \one$ being the south pole. The north and south 
open hemispheres are noted respectively as $H_N$ and $H_S$. 

The heat kernel between two points $g$ and $g'$ is:
\beq\label{heat_k}
K_\alpha(g,g') = \sum_{j \in {\mathbb N}/2}(2j + 1)e^{-j(j+1)\alpha} \frac{\sin((2j + 1) \psi(g'g^{-1})}{\sin \psi(g'g^{-1})} \,,
\eeq
where $\psi(g) \in [0, \pi]$ is the class angle of $g \in SU(2)$, which is 0 at $\one$
and $\pi$ at $-\one$.
It is also the sum over Brownian paths in $SU(2)$ from $g$ to $g'$
$$
K_\alpha(g,g') = \int \extd P_{\alpha} (g,g')[\omega]
$$
where $\extd P_{\alpha} (g,g')[\omega]$ is the Wiener measure over Brownian paths $\omega$ going from
$g$ to $g'$ in time $\alpha$.

The northern heat kernel with Dirichlet boundary conditions, called
$K^{N,D}_\alpha(g,g')$  is the same integral, but in which the 
Brownian paths are constrained to lie entirely in $H_N$, except
possibly their end points $g$ and $g'$, which  are allowed to belong to the {\emph{closed}}
hemisphere $\bar H_N$. Obviously:
\beq\label{dirichlet}
K^{N,D}_\alpha(g,g')  \le K_\alpha(g,g')
\eeq
since there are less paths in the left hand side than in the right hand side.

From the Markovian character of the heat kernel 
$K_\alpha$ we have a convolution equation for 
with $g \in H_S$, in terms of the first hitting point $g'$
where the path visits the equatorial boundary: 

\beq\label{conv}
K_\alpha  (I, g)  = \int_0^{\alpha} \extd \alpha' \int_{g' \in H_0}  \extd g' K^{N,D}_{\alpha'} (I, g') K_{\alpha- \alpha'} (g', g) .
\eeq

\section{Proof of heat kernel bounds}

In order to prove lemma \ref{heat}, we first re-express the heat kernel on $\SU(2)$ in terms of the third Jacobi $\theta$-function
\beq
\theta_3 (z , t) = 1 + 2 \sum_{n = 1}^{+ \infty} \e^{{\rm i} \pi n^2 t} \cos (2 \pi n z)\,,
\eeq
defined for any $(z , t) \in \mathbb{C} \times \mathbb{R}$. We note $\theta_3'$ its derivative with respect to $z$:
\beq
\theta_3' (z , t) = - 4 \pi \sum_{n = 1}^{+ \infty} n \e^{{\rm i} \pi n^2 t} \sin (2 \pi n z)\,.
\eeq 
From equation (\ref{heat_k}), we deduce that:
\beq
K_{\alpha} (g) = \frac{- \e^{\alpha / 4}}{4 \pi \sin \psi(g)} \theta_3' \left( \frac{\psi(g)}{2 \pi} , \frac{i \alpha}{4 \pi} \right)\,. 
\eeq
The main interest of this expression is that $\theta_3$ transforms nicely under the modular group, and in particular\footnote{This is a consequence of the Poisson summation formula, so one might as well directly use this theorem instead of introducing $\theta$.}:
\beq
\theta_3 ( \frac{z}{t} , \frac{- 1}{t}) = \sqrt{- {\rm i} t } \e^{\frac{{\rm i} \pi z^2 }{t}} \theta_3 ( z , t)\,.
\eeq
Differentiation with respect to $z$ yields:
\beq
\theta_3' (z , t) = \frac{\e^{- \frac{{\rm i} \pi z^2 }{t}}}{t \sqrt{- {\rm i} t }} \left( \theta_3' \left( \frac{z}{t} , \frac{- 1}{t} \right) - 2 \pi {\rm i} z \theta_3 \left( \frac{z}{t} , \frac{- 1}{t} \right) \right)\,,
\eeq
and allows to express the heat kernel as
\beq
K_{\alpha} (g) = \frac{\e^{- \frac{ {\psi(g)}^2 }{\alpha}}}{ \alpha^{3/2} } 
\times \frac{ (4 \pi)^{1/2} {\rm i} \e^{\alpha / 4}}{ \sin \psi(g)} \left( \theta_3' \left( \frac{ - 2 {\rm i} \psi(g) }{\alpha} , \frac{4 {\rm i} \pi}{\alpha} \right) - {\rm i} \psi(g) \theta_3 \left(  \frac{ - 2 {\rm i} \psi(g) }{\alpha} , \frac{4 {\rm i} \pi}{\alpha} \right) \right) \,.
\eeq

Using the explicit expressions of $\theta_3$ and $\theta_3'$, we finally obtain:

\beq\label{nice_form}
K_{\alpha} (g) = K_{\alpha}^{0} (g) \frac{ \sqrt{4 \pi} \e^{\alpha / 4} \psi(g)}{\sin{\psi(g)}} F_\alpha (\psi(g))\,,
\eeq
where
\bes
K_{\alpha}^{0} (g) &\equiv& \frac{\e^{- \frac{ {\psi(g)}^2 }{\alpha}}}{ \alpha^{3/2} } \,, \\
F_{\alpha} (g) &\equiv&  1 + \sum_{n = 1}^{+ \infty} \e^{- 4 \pi^2 n^2 / \alpha} 
\left( 2 \cosh( \frac{4 \pi n \psi(g)}{\alpha} ) - \frac{4 \pi n}{\psi(g)} \sinh( \frac{4 \pi n \psi(g)}{\alpha}
) \right) \,.
\ees

\

This formula is suitable for investigating the behavior of $K_\alpha$ away from $- \one$. In particular, simple integral bounds on $F_\alpha$ allow to prove that:
\beq
K_\alpha (g) \underset{\alpha \to 0}{\sim} \frac{\e^{- \frac{ {\psi(g)}^2 }{\alpha}}}{ \alpha^{3/2} } 
\frac{ \sqrt{4 \pi} \psi(g)}{\sin{\psi(g)}} 
\eeq
uniformly on any compact $H$ such that $- \one \notin H$. We shall therefore first study the behavior of $K_\alpha$ and its derivatives on the fixed compact $H_{\frac{3 \pi}{4}} = \{ g \in \SU(2) \vert \psi(g) \leq \frac{3 \pi}{4}\}$. Relying on convolution properties of the heat kernel, we will then extend these results to all of $\SU(2)$. 

\

\noindent {\bf Bounds on $H_{\frac{3 \pi}{4}}$}

\
$K_{\alpha}^{0}$ is easy to analyze, as it is nothing but the flat version of $K_{\alpha}$. The function $\psi \mapsto \frac{\psi(g)}{\sin \psi(g)}$ is analytic on $H_{\frac{3 \pi}{4}}$, therefore its contributions to $K_\alpha$ and its derivatives will be uniformly bounded. The non trivial point of the proof consists in proving that $F_\alpha$ and all its derivatives are also uniformly bounded, by constants independent of $\alpha \in ] 0 , 1]$. By expanding the hyperbolic functions, we can first write:
\bes
F_\alpha (\psi) &=& 1 + \sum_{n = 1}^{+ \infty} \e^{- 4 \pi^2 n^2 / \alpha}
	\sum_{p = 0}^{\infty} a_p (n , \alpha) \psi^{2 p} ,  \\
a_p (n , \alpha) &\equiv& \frac{2}{(2 p) !} \left( \frac{4 \pi n}{ \alpha } \right)^{2 p} \left[ 1 - \frac{16 \pi^2 n^2}{(2 p + 1) \alpha} \right]	.
\ees
We can fix $0 < \epsilon < 1$, and find a constant $K_\epsilon$ such that:
\beq
\e^{- 4 \pi^2 n^2 / \alpha} |a_p (n , \alpha ) | \leq K_\epsilon \e^{- 4 \pi^2 n^2 (1 - \epsilon) / \alpha} \frac{2}{(2 p) !} \left( \frac{4 \pi n}{ \alpha } \right)^{2 p} \,.
\eeq
This implies the following bounds, for any $k \in \mathbb{N}$:
\beq
\vert F_\alpha^{(k)} (\psi) \vert \leq \frac{\partial^{k}}{\partial \psi^k} \left( 1 + 2 K_\epsilon \sum_{n = 1}^{+ \infty} \e^{- 4 \pi^2 n^2 (1 - \epsilon) / \alpha} \cosh ( \frac{4 \pi n \psi}{\alpha} ) \right) \,.
\eeq
When $n \geq 1$, we can use the fact that
\beq
\frac{\partial^{k}}{\partial \psi^k} \cosh ( \frac{4 \pi n \psi}{\alpha} ) \leq  ( \frac{4 \pi n}{\alpha} )^{k} \cosh ( \frac{4 \pi n \psi}{\alpha} ) \,,
\eeq
and the exponential decay in $n^2 / \alpha$ to deduce bounds without derivatives. All in all, we see that for any $\epsilon$, we can find constants $K_\epsilon^{(k)}$ such that:
\bes
\vert F_\alpha (\psi) \vert &\leq& 1 + K_\epsilon^{(0)} \sum_{n = 1}^{+ \infty} \e^{- 4 \pi^2 n^2 (1 - \epsilon) / \alpha} \cosh ( \frac{4 \pi n \psi}{\alpha} ) \,, \\
\vert F_\alpha^{(k)} (\psi) \vert &\leq& K_\epsilon^{(k)} \alpha^{\frac{- k}{2}} \sum_{n = 1}^{+ \infty} \e^{- 4 \pi^2 n^2 (1 - \epsilon) / \alpha} \cosh ( \frac{4 \pi n \psi}{\alpha} ) \,.
\ees
 Following \cite{bgriv}, let us assume that $\epsilon \leq \frac{5 \pi}{8}$, in order to ensure that the function $$
x \mapsto \e^{- 4 (1 - \epsilon) \pi^2 x^2 / \alpha} \cosh ( \frac{4 \pi x \psi} {\alpha} ) 
$$
decreases on $[1, + \infty [$ for any $\psi \in [ 0 , \frac{3 \pi}{4} ]$. This provides us with the following integral bound:
\beq
\sum_{n = 1}^{+ \infty} \e^{- 4 (1 - \epsilon) \pi^2 n^2 / \alpha} \cosh ( \frac{4 \pi n \psi}{\alpha} ) \leq \int_{1}^{ + \infty} \e^{- 4 (1 - \epsilon) \pi^2 x^2 / \alpha} \cosh ( \frac{4 \pi x \psi} {\alpha} ) \extd x \,.
\eeq
Putting the latter in Gaussian form yields an expression in terms of the error function $\rm{erfc}(x) \equiv \int_{x}^{+ \infty} \e^{- t^2} \extd t \leq \e^{- x^2}$:
\bes
\int_{1}^{ + \infty} \e^{- 4 (1 - \epsilon) \pi^2 x^2 / \alpha} \cosh ( \frac{4 \pi x \psi} {\alpha} ) \extd x 
&=& \frac{ \e^{\frac{\psi^2}{(1 - \epsilon) \alpha}}\sqrt{ \pi \alpha}}{8 \pi \sqrt{1 - \epsilon}} \left[ \rm{erfc}\left( 2 \pi \sqrt{\frac{1 - \epsilon}{\alpha}}\pi + \frac{\psi}{\sqrt{(1 - \epsilon) \alpha}}\right) \right. \nn \\
&& \qquad \left. + \; \rm{erfc}\left( 2 \pi \sqrt{\frac{1 - \epsilon}{\alpha}}\pi - \frac{\psi}{\sqrt{(1 - \epsilon) \alpha}}\right) \right] \\
&\leq& \frac{ \e^{\frac{\psi^2}{(1 - \epsilon) \alpha}}\sqrt{ \pi \alpha}}{8 \pi \sqrt{1 - \epsilon}} \left[ \e^{- (2 \pi \sqrt{\frac{1 - \epsilon}{\alpha}}\pi + \frac{\psi}{\sqrt{(1 - \epsilon) \alpha}})^{2}} \right. \nn \\
&& \qquad \left. + \; \e^{- (2 \pi \sqrt{\frac{1 - \epsilon}{\alpha}}\pi - \frac{\psi}{\sqrt{(1 - \epsilon) \alpha}})^{2}} \right] \\
&\leq& \frac{\sqrt{ \pi \alpha}}{8 \pi \sqrt{1 - \epsilon}} \e^{- 4 \pi^2 \frac{1 - \epsilon}{\alpha}}
 \left[ \e^{\frac{4 \pi \psi}{\alpha}} + \e^{- \frac{4 \pi \psi}{\alpha}} \right]\,.
\ees
The last expression is bounded by a constant independent of $\alpha$ and $\psi \in [0 , \frac{3 \pi}{4}]$ provided that $\epsilon \leq \frac{1}{4}$, which is an admissible choice. This concludes the proof of the existence of constants $K^{(k)}$ such that:
\bes\label{bound_F}
\vert F_\alpha (\psi) \vert &\leq& 1 + K^{(0)} \sqrt{\alpha}  \,, \\
\vert F_\alpha^{(k)} (\psi) \vert &\leq& K^{(k)} \alpha^{\frac{1 - k}{2}} \,,
\ees
on $H_{\frac{3 \pi}{4}}$. Using equation (\ref{nice_form}), it is then easy to prove that when $\alpha \in ] 0 , 1 ]$, $K_\alpha$ verifies the same type of bounds as $K_\alpha^0$ on $H_{\frac{3 \pi}{4}}$, therefore concluding the proof of lemma \ref{heat} on this subset.

%

\

\noindent {\bf Extension to $\SU(2)$}

\

Suppose that $g \in \SU(2) \setminus H_{3 \pi /4}$. We can use formula (\ref{conv}) and (\ref{dirichlet}) to write:
\beq
K_\alpha (g) \leq \int_0^{\alpha} \extd \alpha' \int_{g' \in H_0}  \extd g' K_{\alpha'} (g') K_{\alpha- \alpha'} (g'^{\inv} g) \,.
\eeq
This upper bound involves only heat kernels evaluated in $H_{3 \pi /4}$. Moreover, $K_{\alpha'} (g')$ does not depend on the particular value of $g' \in H_0$, the squared distance to $\one$ of the latter and of $g'^{\inv} g$ being bounded from below by a constant $c > 0$. 

From the discussion above, we know that there exists constants $\delta_1$ and $K_1$ such that:
\beq
K_\alpha (g) \leq K_1 \int_{g' \in H_0} \extd g' \int_0^{\alpha} \extd \alpha'  \frac{ \e^{- \delta_1 \vert g' \vert^2 / \alpha'}}{\alpha'^{3/2}} \frac{\e^{- \delta_1 \vert g'^{\inv} g \vert^2 / (\alpha- \alpha')}}{(\alpha- \alpha')^{3/2}} \,.
\eeq
To take care of the singularities in $\alpha' = 0$ and $\alpha' = \alpha$, we can decompose the integral over $\alpha'$ into two components: from $0$ to $\alpha/2$, and from $\alpha/2$ to $\alpha$. Each of these integrals can then be bounded independently, for instance:
\bes
\int_0^{\alpha / 2} \extd \alpha'  \frac{ \e^{- \delta_1 \vert g' \vert^2 / \alpha'}}{\alpha'^{3/2}} \frac{\e^{- \delta_1 \vert g'^{\inv} g \vert^2 / (\alpha- \alpha')}}{(\alpha- \alpha')^{3/2}}
&\leq& \int_0^{\alpha / 2} \extd \alpha' \frac{ \e^{- \delta_1 c / \alpha'}}{\alpha'^{3/2}} \frac{\e^{- \delta_1 c / \alpha}}{(\alpha / 2)^{3/2}} \\
&=& K_2 \frac{\e^{- \delta_1 c / \alpha}}{\alpha^{3/2}} \,.
\ees
We can bound the second integral in the same way, and therefore conclude that:
\beq\label{ext_bound}
K_{\alpha}(g) \leq K \frac{\e^{- \delta_2 / \alpha}}{\alpha^{3/2}} \leq K \frac{\e^{- \delta \vert g \vert^{2} / \alpha}}{\alpha^{3/2}}
\eeq
for some constants $K > 0$ and $\delta > 0$.

\

We can proceed in a similar way for the derivatives of $K_\alpha$. We fix $k \geq 1$, and a normalized Lie algebra element $X$. From (\ref{conv}), we deduce 
\bes
\vert (\cL_X)^n K_\alpha  (g) \vert  &=& \vert \int_0^{\alpha} \extd \alpha' \int_{H_0}  \extd g' K^{N,D}_{\alpha'} (g') (\cL_X)^k K_{\alpha- \alpha'} (g'^{\inv} g) \vert \\
&\leq& \int_0^{\alpha} \extd \alpha' \int_{H_0}  \extd g' K_{\alpha'} (g') \vert (\cL_X)^k K_{\alpha- \alpha'} (g'^{\inv} g) \vert \\
&\leq&  K_1 \int_{g' \in H_0} \extd g' \int_0^{\alpha} \extd \alpha'  \frac{ \e^{- \delta_1 \vert g' \vert^2 / \alpha'}}{\alpha'^{3/2}} \frac{\e^{- \delta_1 \vert g'^{\inv} g \vert^2 / ( \alpha- \alpha')}}{(\alpha- \alpha')^{(3 + k)/2}}\,,
\ees
for some constants $K_1$ and $\delta_1$. 
The same method as before allows to show that 
\beq
\vert (\cL_X)^k K_\alpha  (g) \vert \leq  K \frac{\e^{- \delta \vert g \vert / \sqrt{\alpha}}}{\alpha^{(3+ k)/2}}\,,
\eeq 
for some constants $K > 0$ and $\delta > 0$.


{}

\end{document}